%% file: main.tex
\documentclass[11pt,lettersize]{article}
\usepackage[dvipsnames]{xcolor}
\usepackage[margin=1in]{geometry}
\usepackage[utf8]{inputenc}
\usepackage[T1]{fontenc}
\usepackage{microtype}
\usepackage{algorithm}
\usepackage[noend]{algpseudocode}
\usepackage{amssymb}
\usepackage{amsmath}
\usepackage[framemethod=TikZ]{mdframed}
\usepackage{amsthm}
\usepackage{graphicx}
\usepackage[inline]{enumitem}
\usepackage{multirow}
\usepackage{mathtools,thm-restate}
\usepackage{pifont}
\usepackage{cite}
\usepackage{subcaption}
\usepackage{pifont}
\usepackage{centernot}
\usepackage{wrapfig}
\usepackage{adjustbox}
\usepackage[hyphenbreaks]{breakurl}
\newtheorem{theorem}{Theorem}
\newtheorem{claim}[theorem]{Claim}
\newtheorem{lemma}[theorem]{Lemma}

\newtheorem{property}{Property}
\newtheorem{definition}{Definition}
\newtheorem{invariant}{Invariant}
\newtheorem{assumption}{Assumption}

\newcommand{\set}[1]{\left\{#1 \right\}}
\newcommand{\can}{\textsf{can}}
\DeclarePairedDelimiter{\tup}{\langle}{\rangle}
\DeclarePairedDelimiter{\parens}{(}{)}
\newcommand{\CalU}{\mathcal{U}}
\newcommand{\insr}{\textsf{insert}}
\newcommand{\del}{\textsf{delete}}
\newcommand{\lookup}{\textsf{lookup}}

\newcommand{\localind}[1]{\overset{#1}{\sim}}

\newcommand{\rank}{\mathit{rank}}
\newcommand{\dist}{\mathit{dist}}
\newcommand{\mem}{\mathit{val}}
\newcommand{\lookahead}{\mathit{next}}
\newcommand{\flag}{\mathit{mark}}

\newcommand{\LL}{\textsf{LL}}
\newcommand{\SC}{\textsf{SC}}
\newcommand{\DSC}{\textsf{DSC}}

\newcommand{\VL}{\textsf{VL}}
\newcommand{\seq}{\mathit{seq}}
\newcommand{\pos}{\mathit{pos}}
\newcommand{\resp}{\textit{resp}}

\newcommand{\pt}{\textit{pt}}
\newcommand{\canmult}{\can_{\textit{mult}}}
\newcommand{\canla}{\can^2}

\newcommand{\helpop}{\textsf{help\_op}}

\newcommand{\prop}{\textsf{propagate}}

\newcommand{\ncell}{m}

\newcommand{\nelem}{n}
\newcommand{\ndomain}{u}
\newcommand{\prio}[2]{#1_{p_{#2}}}
\newcommand{\Stable}{\mathit{S}}
\newcommand{\I}{\mathit{I}}
\newcommand{\D}{\mathit{D}}

\newcommand{\abort}{\textsf{abort}}
\newcommand{\posit}{\mathit{pos}}
\newcommand{\calU}{\CalU}
\newcommand{\Op}{\mathit{Op}}
\newcommand{\E}{\mathbb{E}}
\newcommand{\invref}[1]{Invariant~\hyperref[#1]{\ref*{#1}}}
\newcommand{\invenumref}[2]{Invariant~\hyperref[#1#2]{\ref*{#1}\ref*{#1#2}}}

\newcommand{\remove}[1]{}

\newcommand{\ns}[1]{{\color{blue} #1}}

\algnewcommand{\IfThenElse}[3]{
  \State \algorithmicif\ #1\ \algorithmicthen\ #2\ \algorithmicelse\ #3}

\newcommand{\newlinearrow}{\textcolor{gray}{\(\hookrightarrow\)}}
\algnewcommand\algorithmiccommentof[1]{\hfill \newlinearrow \quad \textcolor{blue}{\(\triangleright\) #1}}
\algnewcommand\Commentof{\Statex \algorithmiccommentof}
\algnewcommand{\LineComment}[1]{\Statex\textcolor{blue}{\(\triangleright\) #1}}

\title{History-Independent Concurrent Hash Tables}

\author{
  Hagit Attiya\\
  Technion, Israel\\
  \texttt{hagit@cs.technion.ac.il}
  \and
  Michael A. Bender\\
  Stony Brook University, USA \\
  \texttt{bender@cs.stonybrook.edu}
  \and
  Mart{\'\i}n Farach-Colton\\
  New York University, USA \\
  \texttt{martin@farach-colton.com}
  \and
  Rotem Oshman\\
  Tel-Aviv University, Israel\\
  \texttt{roshman@tau.ac.il}
  \and
  Noa Schiller\\
  Tel-Aviv University, Israel\\
  \texttt{noaschiller@mail.tau.ac.il}
}
\date{}

\begin{document}

\maketitle

\input{abstract}

\input{intro}

\input{preliminaries}
\input{impl-sketch}
\input{impl-llsc}
\input{amortized_sketch}

\input{amortized}

\input{lower-bounds}
\input{related}

\section*{Acknowledgments}
Hagit Attiya is partially supported by the Israel Science Foundation (grant number 22/1425).
Rotem Oshman and Noa Schiller are funded by NSF-BSF Grant No. 2022699.
Michael Bender and Mart{\'\i}n Farach-Colton are funded by NSF Grants 
CCF-2420942, CCF-2423105, 
CCF-2106827, 
CCF-2247577. 


\bibliographystyle{plain}
\bibliography{cite}

\end{document}

%% file: abstract.tex
\begin{abstract}
A \emph{history-independent} data structure does not reveal the history of operations applied to it, 
only its current logical state, even if its internal state is examined.
This paper studies history-independent concurrent \emph{dictionaries}, 
in particular, hash tables, 
and establishes inherent bounds on their space requirements. 

This paper shows
that there is a lock-free history-independent 
concurrent hash table, in which each memory cell stores two elements and two bits,
based on Robin Hood hashing. 
Our implementation is linearizable, 
and uses the shared memory primitive $\LL$/$\SC$.
The expected amortized step complexity of the hash table is $O(c)$,
where $c$ is an upper bound on the number of concurrent operations that access the 
same element, assuming the hash table is not overpopulated.
We complement this positive result by showing that even if we have only two concurrent processes,
no history-independent 
concurrent dictionary 
that supports sets of any size,
with \emph{wait-free} membership queries and \emph{obstruction-free} insertions and deletions,
can store only two elements of the set and a constant number of bits in each memory cell.
This holds even if the step complexity of operations on the dictionary is unbounded.

\end{abstract}

%% file: intro.tex
\section{Introduction}\label{sec:intro}

A data structure is said to be \emph{history independent} (HI)~\cite{Micciancio97,NaorTe01} if its representation (in memory or on disk) depends only the current logical state of the data structure and does not reveal the history of operations that led to this state.
History independence is a privacy property: for example,
a history-independent file system does not reveal that a file previously existed in the system but was then deleted,
or that one file was created before another.
The notion of history independence was introduced by Micciancio~\cite{Micciancio97},
and
Naor and Teague~\cite{NaorTe01} then formalized two now-classical notions of history independence: a data structure is \emph{weakly history independent} (WHI)  if it leaks no information
to an observer
who sees the representation a single time;
it is \emph{strongly history independent} (SHI) 
if it leaks no information
even to an observer who 
sees the representation at multiple points in the 
execution.
History independence has been extensively studied in sequential 
data structures~\cite{Micciancio97, NaorTe01, HartlineHoMo05,HartlineHoMo02,AcarBlHa04, BuchbinderPe03, BlellochGo07, NaorSeWi08, Golovin09, Golovin10,BenderBeJo16} (see Section~\ref{sec:related}),
and the foundational algorithmic work on history independence has 
found its way into secure storage systems and voting machines~\cite{BethencourtBoWa07, BajajSi13b,BajajChSi16,BajajChSi15, BajajSi13a,PoddarBoPo16,RocheAvCh16,ChenSi15,GKMT16}.
Surprisingly, despite the success of history independence
in the sequential setting, 
there are very few results on history independence in \emph{concurrent} data structures~\cite{AtBeFaCoOsSc24,ShunBlellochSPAA14}.
In this work we initiate the study of \emph{history-independent concurrent hash tables}.
A hash table implements a \emph{dictionary},
an abstract data type that supports insertions, deletions and lookups of items in a set.
Hash tables make for an interesting test case for history-independent concurrent data structures:
on the one hand,
they naturally disperse contention between threads
and allow for true parallelism~\cite{HerlihyShavitBook, PurcellH05, HerlihyShavitTzafrir08, ShalevShavit03, Michael02, Click2007}. %
On the other hand, many techniques that are useful in concurrent data structures are by nature not history independent:
for example,
\emph{timestamps} and
\emph{process identifiers} 
are often encoded into the data structure
to help manage contention between threads,
but they reveal information about the order of operations invoked in the past
and which process invoked them;
and \emph{tombstones} are sometimes used to mark an element as deleted without physically removing it, 
decreasing contention, 
but they disclose information about deleted elements
that are no longer in the dictionary.
In addition to not being history independent,
these standard tools of the trade require \emph{large memory cells},
often impractically so,
whereas hash tables are meant to serve as a space-efficient, lightweight dictionary.
Ideally,
if the universe size is $|\calU| = u$,
then
each memory cell should be of width $O(\log u)$,
storing a single element in the set and not much more.




\paragraph{Our goal.} 

%
It is tricky to define history independence in a concurrent setting~\cite{AtBeFaCoOsSc24}.
The bare minimum that one could ask for is \emph{correctness} 
(i.e., linearizability~\cite{HerlihyWing90}) in concurrent executions, 
and \emph{history independence} when restricted to sequential executions.  
To date there has been no hash table satisfying even this minimal notion.
A slightly stronger definition was proposed in~\cite{AtBeFaCoOsSc24}: that the data structure be linearizable 
in concurrent executions,
and be history independent when no state-changing operation is in progress.  This is called \emph{state-quiescent history independence (SQHI)}.
The open problem that we consider is:

\begin{quote}
Is there a linearizable lock-free SQHI hash table that
 uses $O(\ncell)$ memory cells, of width $O(\log u)$ each, to represent a set of $\ncell$ elements from a domain of size $u$, independent of the number of concurrent operations?
\end{quote}

\paragraph{Our results.}
Our main result is a positive one:

\begin{restatable}{theorem}{positive}
    \label{thm:upper}
    There is a lock-free%
    \footnote{\emph{Lock-freedom} requires that at any point in time, if we wait long enough, \emph{some} operation will complete.} 
    SQHI concurrent  hash table using the shared memory primitive $\LL$/$\SC$,%
    \footnote{See Section~\ref{sec:prelim} for the definition of $\LL$/$\SC$.}
    where each memory cell stores two elements and two bits
    (i.e., $2\lceil \log u\rceil  + 2$ bits). 
    The expected amortized step complexity of the hash table is $O(c)$, where $c$ is an upper bound on the number of concurrent operations that access the same element,
    assuming that the load%
    \footnote{The \emph{load} is the ratio of the number of elements stored in the hash table to the number of memory cells.}
     on the hash table is bounded away from 1.
\end{restatable}

Our construction shows that issues of concurrency that are often handled by heavyweight mechanisms such as global sequence numbers, maintaining a list of all ongoing operations,
and so on,
can be resolved by maintaining in each memory cell a \emph{lookahead} that stores the element in the next cell,
and two bits indicating whether an insertion (resp.\ deletion) is ongoing in this cell.

Our implementation is based on Robin Hood hashing~\cite{CelisLaMu85},
a linear-probing hashing scheme
where 
keys that compete for the same location $i$ in the hash table are sorted by the distance of $i$ from their hash value 
(see Section~\ref{sec:hash_overview}).  We call the distance of cell $i$ from the key's hash value the \emph{priority} of the key in cell $i$.
Priority-based 
linear-probing hash tables have been used for history independence in a sequential~\cite{NaorTe01,BlellochGo07} or partially-concurrent~\cite{ShunBlellochSPAA14} setting,%
\footnote{The hash table constructed in~\cite{ShunBlellochSPAA14} is \emph{phase-concurrent}: it can handle concurrent operations of the same type (insert, delete or lookup), but not concurrent operations of different types (e.g., inserts and lookups).}
but these constructions are insensitive to the choice of priority mechanism.
We show that among all linear probing hash tables based on priority mechanisms with a lookahead, \emph{only} the priority mechanism used in Robin Hood hashing 
has the property that we can conclude an element is not in the table by only reading a single cell---a key property that our algorithm relies on.

As we said above, our hash table stores two keys and two extra bits in each cell.
We complement our construction with a negative result showing that it is \emph{necessary} to store extra information in the cells of the table beyond just a single key,
in order to have both SQHI and wait freedom:

\begin{restatable}{theorem}{tradeoff}[Informal]
\label{thm:natural-assig}
    There is no concurrent SQHI hash table representing sets of up to $\ncell < \ndomain$
    elements using m memory cells, each of which stores either a key that is in the set or $\bot$, and
    supporting wait-free lookups and obstruction-free insertions and deletions.%
    \footnote{
        \emph{Wait-freedom} requires that every operation must terminate 
        in a finite number of steps by the executing thread, 
        regardless of concurrency and contention from other threads. \emph{Obstruction-freedom} only requires an operation to terminate if the process executing it runs by itself for sufficiently long.}
    Furthermore, if each cell also includes $\lceil \log u \rceil$ and 
    constant number of additional bits,
    this impossibility holds for sufficiently large $u$ and $\ncell < \sqrt{\ndomain}$.
\end{restatable}

This result separates sequential hash tables from concurrent ones: sequential hash tables \emph{can} be made history independent while storing one element per cell~\cite{NaorTe01,BlellochGo07}.
Moreover, the impossibility result makes no assumptions about the step complexity of the hash table, so it applies even to highly inefficient implementations (e.g., implementations where individual operations are allowed to read or modify the entire hash table).
The assumption that the dictionary can store up to $\ncell$ elements using $\ncell$ cells holds for our construction from Theorem~\ref{thm:upper}, and it is necessary in Theorem~\ref{thm:natural-assig} to avoid empty cells that are used \emph{only} for synchronization, not for storing values from the set.
We also prove that Theorem~\ref{thm:upper} is tight in the sense that we cannot achieve a \emph{wait-free} history-independent implementation using a one-cell lookahead plus $O(1)$ bits of metadata
(only lock-freedom is possible).

%% file: preliminaries.tex
\section{Preliminaries}
\label{sec:prelim}


\paragraph{The asynchronous shared-memory model.}
We use the standard 
model, in which a finite number of concurrent processes
communicate through 
shared memory consisting of 
$\ncell$
\emph{memory cells}.
The shared memory is accessed by executing \emph{primitive operations}.
Our implementation uses the primitive load-link/store-conditional ($\LL$/$\SC$),
which supports the following operations:
$\LL(x)$ returns the value stored in
memory cell $x$, and $\SC(x, \mathit{new})$ writes the value $\mathit{new}$ to $x$,
if it was not written since the 
last time the process performed an $\LL(x)$ operation (otherwise, $x$ is not modified). $\SC$ returns \textit{true} if it writes successfully, and \textit{false} otherwise.
We also use the \emph{validate} instruction, $\VL(x)$, which checks whether $x$ has been written to since the last time the process performed $\LL(x)$.

An implementation of an abstract data type (ADT) specifies
a program
for each process
and operation of the ADT;
when receiving an \emph{invocation} of an operation,
the process takes \emph{steps} according to this program. 
Each step by a process consists of some local computation,
followed by a single primitive operation on the shared memory.
The process may change its local state after a step, and it
may return a \emph{response} to the operation of the ADT.
We focus on implementations of a \emph{dictionary}, representing an unordered set of elements.%
\footnote{For simplicity, we focus on storing the keys and ignore the values that are sometimes associated with the keys.}
It supports the operations 
$\insr(v)$, $\del(v)$ and $\lookup(v)$, each operation takes an input element $v$ and returns \emph{true} or \emph{false}, indicating whether the element is present in the set based on the specific operation.

A \emph{configuration} $C$ specifies the state of every process and of every memory cell.
The \emph{memory representation} of a configuration $C$
is a vector specifying the state of each memory cell; 
this does not include local private variables held by each process, only the shared memory.
The system may have several initial configurations.
An \emph{execution} $\alpha$ is an alternating sequence of configurations and steps,
which can be finite or infinite.
We say that an operation \emph{completes} in
execution $\alpha$
if $\alpha$ includes both the invocation and response of the operation;
if $\alpha$ includes the invocation of an operation, but no matching response, 
then the operation is \emph{pending}.
We say that $\mathit{op}_1$ \emph{precedes} $\mathit{op}_2$ in execution $\alpha$, if $\mathit{op}_1$ response precedes $\mathit{op}_2$ invocation in $\alpha$.

The standard correctness condition for concurrent implementations is \emph{linearizability}~\cite{HerlihyWing90}: intuitively,
it requires that each operation appears to take place instantaneously at some point between its invocation and its return.
Formally, an execution $\alpha$ is \emph{linearizable} if there is a sequential permutation $\pi$ of the completed, and possibly some uncompleted, operations in $\alpha$, that matches the sequential specification of the ADT, with uncompleted operations assigned some output value.
Additionally,
if $\mathit{op}_1$ precedes $\mathit{op}_2$ in $\alpha$, then $\mathit{op}_1$ also precedes $\mathit{op}_2$ in $\pi$; namely, the permutation $\pi$ respects the real-time order of non-overlapping operations in $\alpha$. We call this permutation a \emph{linearization} of $\alpha$.
An implementation is linearizable 
if all of its executions are linearizable.

An implementation is \emph{obstruction-free} if an operation by process $p$ returns 
in a finite number of steps by $p$, if $p$ runs solo.
An implementation is \emph{lock-free} if whenever there is a pending operation, 
 some operation returns in a finite number of steps of all processes.
Finally, an implementation is \emph{wait-free} if whenever there is a pending operation by process $p$, 
this operation returns in a finite number of steps by $p$.



\paragraph{History independence.}
In this paper we consider implementations where the randomness is fixed up-front;
after initialization,
the implementation is deterministic.
An example of such an initialization is choosing a hash function.
It is known that in this setting, strong {history independence} is equivalent to requiring,
for an object with state space $Q$,
that
every logical state $q\in Q$ corresponds to a unique \emph{canonical memory representation}, 
fixed at initialization~\cite{HartlineHoMo02,HartlineHoMo05}.

Several notions of history independence for concurrent implementations were explored in~\cite{AtBeFaCoOsSc24}.
In this paper we adopt the notion of \emph{state-quiescent history independence} (SQHI), which requires that the 
memory representation be in its canonical representation at any point in the execution where \emph{no state-changing operation} (in our case, $\insr$ or $\del$) is pending.
The logical state of the object is determined using a linearization of the execution up to the point in which the observer inspects the memory representation.
(It is shown in~\cite{AtBeFaCoOsSc24} that stronger notions of history independence are impossible or prohibitively expensive;
in particular, if we require history independence 
at \emph{all} points in the execution,
an open-addressing hash table over a universe $\calU$ must have
size at least $|\CalU|$.)


%% file: impl-sketch.tex
\section{Overview of the Hash Table Construction}
\label{sec:hash_overview}

In this section we describe our main result, the construction of a concurrent history-independent lock-free hash table where each memory cell stores two elements plus two bits.
We give a high-level overview, glossing over many subtle details in the implementation.

Our hash table is based on \emph{Robin Hood hashing}~\cite{CelisLaMu85}, a type of linear-probing hash table.
In Robin Hood hashing, each element $x$ is inserted as close to its hash location $h(x)$ as possible,
subject to the following priority mechanism: 
for any cell $i$
and elements $x \neq y$,
element $x$ has \emph{higher priority in cell $i$} than $y$,
denoted $x \prio{>}{i} y$,
if cell $i$ is farther from $x$'s hash location than it is from $y$'s:
$i - h(x) > i - h(y)$
(or, to break ties,
if $i - h(x) = i - h(y)$ and $x > y$).\footnote{Here and throughout, we assume modular arithmetic that wraps around the $\ncell$-th cell.}

Robin Hood hash tables maintain the following invariant:%
\footnote{
    We assume that in all cells, $\bot$, indicating an empty cell, has lower priority than all other elements.}
\begin{invariant}[Ordering Invariant~\cite{ShunBlellochSPAA14}]
\label{inv:order-invar}
    If an element $v$ is stored in cell $i$, then for any cell $j$ between $h(v)$ and $i$,
    the element $v'$ stored in cell $j$ has priority higher than or equal to $v$ (that is, $v' \prio{\geq}{j} v$).
\end{invariant}
Note that we use `higher than or equal to' instead of `higher than', even though all elements are distinct.
This is because we later use this invariant in a table that includes duplicates of the same element (see Section~\ref{sec:impl}).
The invariant is achieved as follows.

Let $A[0,\ldots,\ncell-1]$ be the hash table.
To insert an element $v$,
we first try to insert it at its hash location, $h(v)$.
If $A[h(v)]$ is occupied by an element $v'$,
then whichever element has lower priority in cell $h(v)$
is \emph{displaced}, that is, pushed into the next cell.
If the next cell is not empty, 
then the displaced element either displaces the element stored there or is displaced again,
creating a chain of displacements
that ends when we reach the \emph{end of the run}:
the first empty cell encountered in the probe sequence.
(A \emph{run} is a maximal consecutive sequence of occupied cells.)
We place the last displaced element in the empty cell, and return.
We refer to the process of shifting elements forward as \emph{propagating the insertion}.

Deletions are also handled in a way that preserves the invariant:
after deleting an element $v$ from cell $i$ of the hash table,
we check whether the element $v'$ in cell $i+1$ ``prefers''
to shift backwards, that is, whether cell $i$ is closer to $h(v')$ than cell $i+1$.
If so, we shift $v'$ backwards into cell $i$ and examine the next element. We continue shifting elements backwards until we reach either the end of the run (an empty cell), or an element that is already at its hash location.
This is referred to as \emph{propagating the deletion}.

Finally, to perform a $\lookup(v)$ operation,
we scan the hash table starting at location $h(v)$,
until we either find $v$ and return \textit{true}, reach the end of the run and return \textit{false},
or find an element with lower priority than $v$ (that is,
reach a cell $j$ such that $A[j] \prio{<}{j} v$).
In the latter case, the invariant allows us to conclude
that $v$ is not in the hash table, and we return \textit{false}.


The crucial property of Robin Hood hashing that makes it suitable for our concurrent implementation is that following an insertion or deletion, \emph{no element in the hash table moves by more than one cell}.
In Section~\ref{sec:robin_hood_lookahead}, we discuss why, among all priority mechanisms, only Robin Hood hashing is suitable for our implementation.

As we mentioned in Section~\ref{sec:intro}, several traditional mechanisms for ensuring progress in the face of contention are unavailable to us,  because they compromise history independence or increase the cell size too much (or both). We begin by outlining the challenges involved in constructing a concurrent history-independent Robin Hood hash table, and then describe our implementation and how it overcomes them.

\paragraph{Handling displaced elements.}
In Robin Hood hashing, insertions may result in a chain of displacements, with multiple higher-priority elements ``bumping'' (displacing) lower-priority elements to make room for a newly-inserted element.
    In a sequential implementation, 
    if we wish to ``bump'' a lower-priority element $y$
    in favor of a higher-priority element $x$,
    we simply store $y$ in local memory,
    overwrite the cell containing $y$ in the hash table to place $x$ there instead,
    and then proceed to find a place for $y$ in the hash table.
    While a new place for $y$ is sought, it exists only in the local memory of the operation performing the insertion.
    In a \emph{concurrent} implementation this approach is dangerous: if a displaced element $y$ does not physically exist in the hash table, then a concurrent $\lookup(y)$ operation may mistakenly report that $y$ is not in the hash table. This is not allowed in a linearizable implementation.%
    \footnote{Unless $y$ itself is concurrently being inserted or deleted; however, Robin Hood hashing displaces elements upon insertion or deletion of \emph{other} elements, so we are not guaranteed this.}

\paragraph{Avoiding moves that happen behind a lookup's back.}
In sequential Robin Hood hashing,
if a $\lookup(x)$ operation is invoked when $x$ is not in the set,
the operation traverses the table starting from location $h(x)$, until it reaches an empty cell or a cell $i$ with a value $y \prio{<}{i} x$ and returns \textit{false}.
However, in a concurrent implementation this is risky:
a $\lookup(x)$ operation might scan from location $h(x)$
until it reaches a cell $i$ with a value $y \prio{<}{i} x$,
      never seeing element $x$,
\emph{even though $x$ is in the set the entire time}.
This can happen if $x$ is initially ahead of the location currently being examined by the $\lookup(x)$ operation,
but then,
between steps of the $\lookup(x)$ operation,
multiple elements stored between $x$ and location $h(x)$ are deleted,
pulling $x$ backwards to a location 
that the $\lookup(x)$ operation has already passed.

\paragraph{Synchronization without timestamps, process IDs, and operation announcements.}
In our implementation we wish to store as little information as possible in each cell, both for the sake of memory-efficiency and also to facilitate history independence (extraneous information must be eventually wiped clear, and this can be challenging).
To ensure progress, operations must be able to ``clear the way'' for themselves if another operation is blocking the way, without \emph{explicitly knowing} what operation is blocking them, or which process invoked that operation, as we do not wish to store that information in the hash table.



\subsection{Robin Hood Hashing with 1-Lookahead and Lightweight Helping}
\label{sec:robin_hood_lookahead}

To overcome the challenges described above we introduce a \emph{lookahead} mechanism.
Each cell $i$ has the form $A[i] = \tup{ \mathit{value}, \mathit{lookahead}, \mathit{mark} }$, 
consisting of three slots:
the \emph{value} slot stores the element currently occupying cell $i$, or $\bot$ for an empty cell;
the \emph{lookahead} slot is intended to store the element occupying cell $i+1$, when no operation is in progress;
and the \emph{mark} slot indicates
the status of the cell, and can take on three values, $\mathit{mark} \in \set{\Stable, \I, \D}$,
with $\Stable$ standing for ``stable'', indicating that no operation is working on this cell,
and $\I, \D$ indicating that an insertion or deletion (resp.) are working on this cell.
    If the cell is stable ($\Stable$),
    then its \emph{lookahead} is consistent with the next cell:
    if $A[i] = \tup{a, b, \Stable}$ and
    $A[i+1] = \tup{c, d, M}$,
    then $b = c$.
    If a cell is marked with $\I$ or $\D$,
    its \emph{lookahead} {might be} inconsistent.

\paragraph{The role of the lookahead.}
Adding a lookahead serves two purposes.
First, it allows lookups to safely conclude that an element is not in the hash table by reading a single cell:
if we are looking for element $x$,
and we reach a cell $A[i] = \tup{v, v', \ast}$
such that either $i = h(x)$ and $x \prio{>}{i} v$, or $v \prio{>}{i} x \prio{>}{i+1} v'$, 
then $x$ can safely be declared to not be in the hash table.
This resolves the concern that an element may be moved ``back and forth'' behind a lookup's back, so that the lookup can never safely decide that the element is not in the hash table.
Second,
the lookahead serves as temporary storage for elements undergoing displacement due to an insertion:
instead of completely erasing a lower-priority element $y$ and writing a higher-priority element $x$ in its place,
we temporarily shift $y$ from the \emph{value} slot to the \emph{lookahead} slot of the cell,
and
store $x$ in the \emph{value} slot.
Later on, we move $y$ into the \emph{value} slot of the next cell, possibly displacing a different element by shifting it into the \emph{lookahead} slot, and so on.
At any time, all elements that are in the table are physically stored in shared memory, either in the \emph{value} slot or the \emph{lookahead} slot of some cell --- ideally, in the \emph{value} slot of some cell and also in the \emph{lookahead} slot of the preceding cell.

\paragraph{The choice of Robin Hood hashing.}
Any priority-based linear-probing hash table satisfies the ordering invariant (Invariant~\ref{inv:order-invar}).
However, Robin Hood hashing also satisfies the following additional invariant:
If an element $v'$ stored in cell $i$ has higher priority than $v$, then for any cell $j$ between $h(v)$ and $i$, the element $v''$ stored in cell $j$ also has higher priority than $v$.
This invariant allows us to determine the following by reading only two consecutive cells,
cell $i$ which stores element $v'$ and cell $i+1$ which stores element $v''$:
\begin{enumerate}
    \item $v$ is in the table ($v = v'$ or $v = v''$),
    \item $v$ is not in the table and should be inserted to cell $i+1$ ($i+1 = h(v)$ and $v \prio{>}{i+1} v''$, or $v' \prio{>}{i} v \prio{>}{i+1} v''$),
    \item if $v$ is in the table, it can only be in a cell before cell $i$ ($v' \prio{<}{i} v$), or
    \item if $v$ is in the table, it can only be in a cell after cell $i+1$ ($v'' \prio{>}{i+1} v$).
\end{enumerate}
This allows the algorithm to make decisions based on a single cell using the \emph{lookahead} slot, without depending on the entire sequence of values read starting from the initial hash location, which may have changed since they were last read.
We show that it is not possible to determine that element $v$ is not in the hash table using any other priority mechanism by reading only two consecutive cells.
We do so by showing that every two consecutive cells in the canonical memory representation of a set that does not include $v$ is equal to the same cells in a different canonical memory representation of a set that does include $v$, where the priority mechanism determines the canonical memory representation
(see Section~\ref{app:linearprob}).

\paragraph{The role of the mark.}
    The \emph{mark} slot can be viewed as a \emph{lock},
    except that cells are locked by \emph{operations},
    not by processes;
    to release the lock, a process must
    \emph{help} the operation that placed the mark,
    completing whatever steps are necessary to propagate
    the operation one step forward, from the ``locked'' cell into the next cell (which is then ``locked'' by the operation).
    Operations move forward in a manner that resembles hand-over-hand locking~\cite{BayerS1977}; we explain this in detail below.


\paragraph{Life cycle of an operation.}
Each operation --- $\lookup(x), \insr(x)$ or $\del(x)$ ---
goes through some or all of the following stages:
\begin{enumerate}[nolistsep,noitemsep]
    \item Finding the ``correct position'' for element $x$:
    starting from position $h(x)-1$,%
    \footnote{To detect concurrent operations on element $x$, we need to start the scan in location $h(x)-1$, not $h(x)$; this will become clear below.} 
    we scan each cell, until we either find element $x$ in the hash table,
    or we find the cell where element $x$ \emph{should have been} if it were in the hash table,
    and conclude that $x$ is not in the hash table.
    As discussed above, the choice of Robin Hood hashing guarantees that if element $x$ is not in the table (and there are no ongoing operations), we can find such a cell.
    At this point, $\lookup(x)$ returns the appropriate answer,
    and $\insr(x), \del(x)$ may also return,
    if $x$ is already in the hash table (for $\insr(x)$)
    or is not in the hash table (for $\del(x)$).
    Note that if $x$ is found in a \emph{lookahead} slot of a cell with a delete operation in progress, we may not determine that $v$ is in the table. This is because an earlier operation may have already detected that the element is deleted, causing the delete operation to ``take effect'', and no later operation can be linearized before this delete operation.

    Along the way, operations \emph{help} advance any other operation that they encounter: if an operation encounters a non-stable cell, it performs whatever steps are necessary to make the cell stable,
    and only then is allowed to proceed.
    \item 
    Let $i-1$ be the cell reached,
    where cell $i$ is the ``correct position'' for element $x$.
    The \emph{lookahead} of cell $i-1$
    should reflect the absence (for $\insr$) or presence (for $\del$) of element $x$.%
    \footnote{This is the reason that scans for element $x$ begin in location $h(x)-1$: it is possible that $x$ is currently being inserted or deleted, but the change is reflected only in location $h(x)-1$ at this point.}
    Then the operation makes its initial write into the table:
        \begin{itemize}[nolistsep,noitemsep]
            \item $\insr(x)$  writes $x$ into the \emph{lookahead} slot of cell $i-1$:
            if $A[i-1] = \tup{ v, v', \Stable}$,
            then the $\insr$ writes $A[i-1] = \tup{ v, x, \I}$,
            ``locking'' the cell.
            At this point, the newly-inserted element $x$ is considered to be displaced,
            as it is stored only in the \emph{lookahead} slot of cell $i-1$,
            immediately preceding its correct location (cell $i$).
            \item $\del(x)$ logically erases $x$ from the \emph{lookahead} slot of cell $i-1$:
            if $A[i-1] = \tup{ v, x, \Stable}$,
            then the $\del$ writes
            $A[i-1] = \tup{ v, x, \D}$,
            ``locking'' the cell.
            The cell now contains a logical ``hole'' (the \emph{lookahead} slot containing the value $v$),
            which may need to be filled by moving backwards
            elements from cells $A[i+1],A[i+2],\ldots$
            until we reach either the end of the run or an element that is already in its hash position.
        \end{itemize}
        \item Following the initial write, operations continue
        moving forward until the end of the run, \emph{propagating}
        any operations that they encounter, to resolve the chain of displacements that may occur.
        In some sense, processes ``lose their identity'' in this stage:
        because we do not use process IDs, timestamps or sequence numbers,
        processes can no longer keep track of their own operation
        and distinguish it from other operations,
        and must instead help propagate all operations equally.
        By proceeding all the way to the end of the run and helping all operations along the way,
        a process ensures that by the time it returns, its own operation has been completed, either by itself
        or some other process, or a different process becomes responsible for ensuring the operation completes (see more details below).
        \end{enumerate}

\subsection{Propagating an Operation}
\label{sec:propagate_operation}

Consider consecutive cells $A[i] = \tup{ a, b, M}, A[i+1] = \tup{c, d, M'}$,
with the mark $M \in \set{\I,\D}$ in cell $A[i]$ indicating
that an operation in progress.
Operations are not allowed to overtake one another,
so if cell $A[i+1]$ is not stable ($M' \neq \Stable$),
we \emph{help} whatever operation is in progress there by propagating it forward one step, clearing the way for the current operation.
Note that helping an operation in cell $A[i+1]$ may require us to first clear the way by helping an operation in cell $A[i+2]$, which may in turn require helping an operation in cell $A[i+3]$, and so on.

Now suppose that cell $A[i+1]$ is stable.
Propagating the operation that is currently in cell $A[i]$ one step forward into cell $A[i+1]$ involves modifying both $A[i]$ and $A[i+1]$;
however, we cannot modify both cells in one atomic step,
as we are using single-word memory primitives.
Thus, the propagation is done by a careful interleaving of $\LL, \VL$ and $\SC$ steps,
in a manner that is similar to hand-over-hand locking~\cite{BayerS1977},
except that the ``locks'',
represented by the \emph{mark}
slot,
are owned by \emph{operations} rather than \emph{processes}:
this allows any process to take over the propagation of an operation that stands in its way, even if it did not invoke that operation.
To modify cells $A[i]$ and $A[i+1]$
when cell $A[i]$ is already ``locked'',
an operation first ``locks'' cell $A[i+1]$,
setting its \emph{mark} appropriately (and also changing its contents, that is, the \emph{value} and \emph{lookahead} slots);
then it releases the ``lock'' on cell $A[i]$,
setting the \emph{mark} to $\Stable$ (and possibly also changing its contents).

In the sequel, we make distinguish between \emph{operations} and \emph{processes}.
An \emph{operation} that is in mid-propagation 
(i.e., at any point following the initial write and before the operation is complete)
and is currently located in cell $A[i]$
can be propagated by any \emph{process} that reaches $A[i]$.
As noted above, following their initial write (and also before it, on their way to their target location),
processes are in some sense nameless workers that simply propagate all the operations they encounter, 
until they reach the end of the run.
However, prior to the initial write of an operation, only the process
that invoked the operation knows of it,
so at this point the operation is synonymous with the process that invoked it.

\paragraph{Propagating an insertion (see Fig.~\ref{fig:llsc_insert}).}
Suppose that $A[i] = \tup{a, b, \I}$ and $A[i+1] = \tup{c, d, \Stable}$,
with cell $A[i]$ ``locked'' by a propagating insertion.
Element $b$ has been displaced by the insertion:
\begin{wrapfigure}[44]{r}{0.45\textwidth}
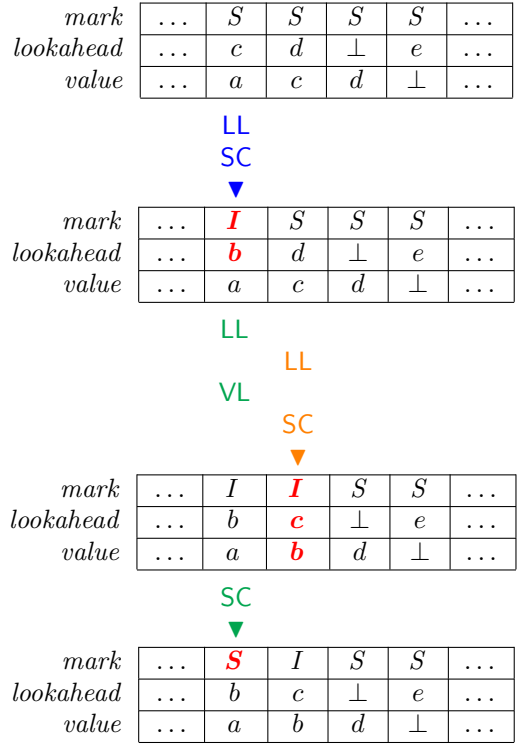
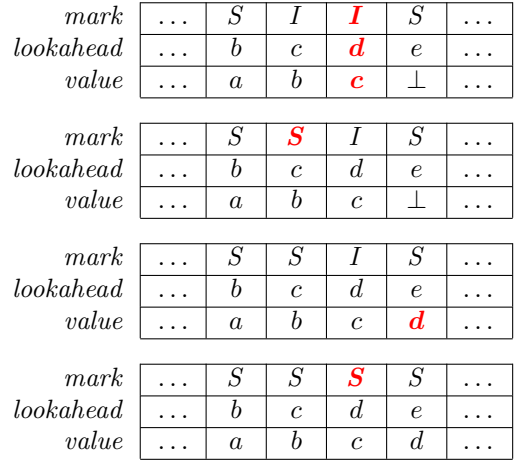

\centering
\begin{subfigure}[t]{0.45\textwidth}
\centering \small
\begin{tabular}{ r|p{3ex}|p{2.5ex}|p{2.5ex}|p{2.5ex}|p{2.5ex}|p{3ex}| } 
 \cline{2-7}
  \emph{mark} & \centering $\ldots$ & \centering $\Stable$ & \centering $\Stable$ & \centering $\Stable$ & \centering $\Stable$ & {\centering $\ldots$} \\
  \cline{2-7}
  \emph{lookahead} & \centering $\ldots$ & \centering $c$ & \centering $d$ & \centering $\bot$ & \centering $e$ & {\centering $\ldots$} \\
  \cline{2-7}
  \emph{value} & \centering $\ldots$ & \centering $a$ & \centering $c$ & \centering $d$ & \centering $\bot$ & {\centering $\ldots$} \\
 \cline{2-7}
\end{tabular}
\\
\vspace{1ex}
\begin{tabular}{ r|p{3ex}|p{2.5ex}|p{2.5ex}|p{2.5ex}|p{2.5ex}|p{3ex}| } 
 \multicolumn{1}{c}{} & \multicolumn{1}{c}{} & \multicolumn{1}{c}{\color{blue}$\LL$} & \multicolumn{1}{c}{} & \multicolumn{1}{c}{} & \multicolumn{1}{c}{} & \multicolumn{1}{c}{} \\
 \multicolumn{1}{c}{} & \multicolumn{1}{c}{} & \multicolumn{1}{c}{\color{blue}$\SC$} & \multicolumn{1}{c}{} & \multicolumn{1}{c}{} & \multicolumn{1}{c}{} & \multicolumn{1}{c}{} \\
 \multicolumn{1}{c}{} & \multicolumn{1}{c}{} & \multicolumn{1}{c}{\color{blue}$\blacktriangledown$} & \multicolumn{1}{c}{} & \multicolumn{1}{c}{} & \multicolumn{1}{c}{} & \multicolumn{1}{c}{} \\
 \cline{2-7}
 \emph{mark} & \centering $\ldots$ & \centering \textcolor{red}{\boldmath{$\I$}} & \centering $\Stable$ & \centering $\Stable$ & \centering $\Stable$ & {\centering $\ldots$} \\
 \cline{2-7}
 \emph{lookahead} & \centering $\ldots$ & \centering \textcolor{red}{\boldmath{$b$}} & \centering $d$ & \centering $\bot$ & \centering $e$ & {\centering $\ldots$} \\
 \cline{2-7}
 \emph{value} & \centering $\ldots$ & \centering $a$ & \centering $c$ & \centering $d$ & \centering $\bot$ & {\centering $\ldots$} \\
 \cline{2-7}
\end{tabular}
\\
\vspace{1ex}
\begin{tabular}{ r|p{3ex}|p{2.5ex}|p{2.5ex}|p{2.5ex}|p{2.5ex}|p{3ex}| } 
 \multicolumn{1}{c}{} & \multicolumn{1}{c}{} & \multicolumn{1}{c}{\color{Green}$\LL$} & \multicolumn{1}{c}{} & \multicolumn{1}{c}{} & \multicolumn{1}{c}{} & \multicolumn{1}{c}{} \\
 \multicolumn{1}{c}{} & \multicolumn{1}{c}{} & \multicolumn{1}{c}{} & \multicolumn{1}{c}{\color{orange}$\LL$} & \multicolumn{1}{c}{} & \multicolumn{1}{c}{} & \multicolumn{1}{c}{} \\
 \multicolumn{1}{c}{} & \multicolumn{1}{c}{} & \multicolumn{1}{c}{\color{Green}$\VL$} & \multicolumn{1}{c}{} & \multicolumn{1}{c}{} & \multicolumn{1}{c}{} & \multicolumn{1}{c}{} \\
 \multicolumn{1}{c}{} & \multicolumn{1}{c}{} & \multicolumn{1}{c}{} & \multicolumn{1}{c}{\color{orange}$\SC$} & \multicolumn{1}{c}{} & \multicolumn{1}{c}{} & \multicolumn{1}{c}{} \\
 \multicolumn{1}{c}{} & \multicolumn{1}{c}{} & \multicolumn{1}{c}{} & \multicolumn{1}{c}{\color{orange}$\blacktriangledown$} & \multicolumn{1}{c}{} & \multicolumn{1}{c}{} & \multicolumn{1}{c}{} \\
 \cline{2-7}
 \emph{mark} & \centering $\ldots$ & \centering $\I$ & \centering \textcolor{red}{\boldmath{$\I$}} & \centering $\Stable$ & \centering $\Stable$ & {\centering $\ldots$} \\
 \cline{2-7}
 \emph{lookahead} & \centering $\ldots$ & \centering $b$ & \centering \textcolor{red}{\boldmath{$c$}} & \centering $\bot$ & \centering $e$ & {\centering $\ldots$} \\
 \cline{2-7}
 \emph{value} & \centering $\ldots$ & \centering $a$ & \centering \textcolor{red}{\boldmath{$b$}} & \centering $d$ & \centering $\bot$ & {\centering $\ldots$} \\
 \cline{2-7}
\end{tabular}
\\
\vspace{1ex}
\begin{tabular}{ r|p{3ex}|p{2.5ex}|p{2.5ex}|p{2.5ex}|p{2.5ex}|p{3ex}| }
 \multicolumn{1}{c}{} & \multicolumn{1}{c}{} & \multicolumn{1}{c}{\color{Green}$\SC$} & \multicolumn{1}{c}{} & \multicolumn{1}{c}{} & \multicolumn{1}{c}{} & \multicolumn{1}{c}{} \\
 \multicolumn{1}{c}{} & \multicolumn{1}{c}{} & \multicolumn{1}{c}{\color{Green}$\blacktriangledown$} & \multicolumn{1}{c}{} & \multicolumn{1}{c}{} & \multicolumn{1}{c}{} & \multicolumn{1}{c}{} \\
 \cline{2-7}
 \emph{mark} & \centering $\ldots$ & \centering \textcolor{red}{\boldmath{$\Stable$}} & \centering $\I$ & \centering $\Stable$ & \centering $\Stable$ & {\centering $\ldots$} \\
 \cline{2-7}
 \emph{lookahead} & \centering $\ldots$ & \centering $b$ & \centering $c$ & \centering $\bot$ & \centering $e$ & {\centering $\ldots$} \\
 \cline{2-7}
 \emph{value} & \centering $\ldots$ & \centering $a$ & \centering $b$ & \centering $d$ & \centering $\bot$ & {\centering $\ldots$} \\
 \cline{2-7}
\end{tabular}

\caption{\small The initial write of $\insr(b)$ 
into a cell $A[i]$,
followed by its first propagation step. The sequence of $\LL,\VL$ and $\SC$ operations is depicted from top to bottom. 
}
\end{subfigure}
\begin{subfigure}[t]{0.45\textwidth}
\centering \small
\begin{tabular}{ r|p{3ex}|p{2.5ex}|p{2.5ex}|p{2.5ex}|p{2.5ex}|p{3ex}|  }
 \cline{2-7}
 \emph{mark} & \centering $\ldots$ & \centering $\Stable$ & \centering $\I$ & \centering \textcolor{red}{\boldmath{$\I$}} & \centering $\Stable$ & {\centering $\ldots$} \\
 \cline{2-7}
 \emph{lookahead} & \centering $\ldots$ & \centering $b$ & \centering $c$ & \centering \textcolor{red}{\boldmath{$d$}} & \centering $e$ & {\centering $\ldots$} \\
 \cline{2-7}
 \emph{value} & \centering $\ldots$ & \centering $a$ & \centering $b$ & \centering \textcolor{red}{\boldmath{$c$}} & \centering $\bot$ & {\centering $\ldots$} \\
 \cline{2-7}
\end{tabular}
\\
\vspace{2ex}
\begin{tabular}{ r|p{3ex}|p{2.5ex}|p{2.5ex}|p{2.5ex}|p{2.5ex}|p{3ex}|  }
 \cline{2-7}
 \emph{mark} & \centering $\ldots$ & \centering $\Stable$ & \centering \textcolor{red}{\boldmath{$\Stable$}} & \centering $\I$ & \centering $\Stable$ & {\centering $\ldots$} \\
 \cline{2-7}
 \emph{lookahead} & \centering $\ldots$ & \centering $b$ & \centering $c$ & \centering $d$ & \centering $e$ & {\centering $\ldots$} \\
 \cline{2-7}
 \emph{value} & \centering $\ldots$ & \centering $a$ & \centering $b$ & \centering $c$ & \centering $\bot$ & {\centering $\ldots$} \\
 \cline{2-7}
\end{tabular}
\\
\vspace{2ex}
\begin{tabular}{ r|p{3ex}|p{2.5ex}|p{2.5ex}|p{2.5ex}|p{2.5ex}|p{3ex}|  }
 \cline{2-7}
 \emph{mark} & \centering $\ldots$ & \centering $\Stable$ & \centering $\Stable$ & \centering $\I$ & \centering $\Stable$ & {\centering $\ldots$} \\
 \cline{2-7}
 \emph{lookahead} & \centering $\ldots$ & \centering $b$ & \centering $c$ & \centering $d$ & \centering $e$ & {\centering $\ldots$} \\
 \cline{2-7}
 \emph{value} & \centering $\ldots$ & \centering $a$ & \centering $b$ & \centering $c$ & \centering \textcolor{red}{\boldmath{$d$}} & {\centering $\ldots$} \\
 \cline{2-7}
\end{tabular}
\\
\vspace{2ex}
\begin{tabular}{ r|p{3ex}|p{2.5ex}|p{2.5ex}|p{2.5ex}|p{2.5ex}|p{3ex}|  }
 \cline{2-7}
  \emph{mark} & \centering $\ldots$ & \centering $\Stable$ & \centering $\Stable$ & \centering \textcolor{red}{\boldmath{$\Stable$}} & \centering $\Stable$ & {\centering $\ldots$} \\
 \cline{2-7}
 \emph{lookahead} & \centering $\ldots$ & \centering $b$ & \centering $c$ & \centering $d$ & \centering $e$ & {\centering $\ldots$} \\
 \cline{2-7}
 \emph{value} & \centering $\ldots$ & \centering $a$ & \centering $b$ & \centering $c$ & \centering $d$ & {\centering $\ldots$} \\
 \cline{2-7}
\end{tabular}

\caption{\small The rest of the propagation steps, omitting the $\LL,\VL$ and $\SC$ operations. 
}
\end{subfigure}

\caption{\small Depiction of the insertion process of element $b$, initiated by operation $\insr(b)$, where $b$ is inserted into a cell $A[i] = \tup{a, c, \Stable}$ such that $a \prio{>}{i} b \prio{>}{i+1} c$.
}
\label{fig:llsc_insert}
\end{wrapfigure}
it is temporarily stored in the \emph{lookahead} slot of $A[i]$, and we need to move it forward into the \emph{value} slot of $A[i+1]$, displacing element $c$ instead (unless $c = \bot$, in which case $A[i+1]$ is an empty cell, and no further elements need to be displaced).
The target state following the propagation step
is $A[i] = \tup{a, b, \Stable}, A[i+1] = \tup{b, c, \I}$,
if $c \neq \bot$ (element $c$ is now displaced), and 
$A[i] = \tup{a, b, \Stable}, A[i+1] = \tup{b, d, \Stable}$ if $c = \bot$ (no element is displaced, and the insertion is done propagating).
We describe here the case where $c \neq \bot$ (the other case is very similar).

We move from the current state, 
$A[i] = \tup{a, b, \I}, A[i+1] = \tup{c, d, \Stable}$,
to
the target state,
$A[i] = \tup{a, b, \Stable}, A[i+1] = \tup{b, c, \I}$,
in several steps:
we begin by performing an $\LL$ on cell $A[i]$,
an $\LL$ on cell $A[i+1]$,
and then a $\VL$ on cell $A[i]$,
to ensure that it has not been altered while we were reading $A[i+1]$.
Then we ``lock'' $A[i+1]$ and at the same time modify its contents, by 
performing an $\SC$ to set
$A[i+1] \leftarrow \tup{ b, c, \I}$.
This $\SC$ may fail, but since $A[i]$ is already ``locked'' (it is marked with $\I$),
no other operation can overtake the current insertion without helping it move forward.
Thus,
there are only two possible reasons for a failed $\SC$ on $A[i+1]$: either
some other process already performed the same $\SC$,
in which case we can proceed to the next step;
or a new operation performed its initial write into $A[i+1]$,
in the process ``locking'' it for itself.
In this case we must first help propagate the new operation to clear the way, and then we can resume trying to propagate the current insertion.

Suppose the $\SC$ on $A[i+1]$ succeeds (either the current process succeeded, or some other process did).
At this point we have $A[i] = \tup{a, b, \I}, A[i+1] = \tup{b, c, \I}$,
with both cells ``locked'' by the insertion. We now ``release the lock''
on $A[i]$
by performing an $\SC$ to set
$A[i] \leftarrow \tup{a, b, \Stable}$.
This $\SC$ may also fail, but in this case, since cell $A[i]$ itself is ``locked'' prior to the $\SC$, the only possible reason is that some other process performed the same $\SC$ successfully.
Thus,
there is no need to even check if the $\SC$ succeeded;
we simply move on to the next cell.
\paragraph{Propagating a deletion (see Fig.~\ref{fig:llsc_delete}).}
Recall that the initial write of a $\del(x)$ operation targets the cell {preceding} the location of $x$, 
``locking'' it by marking it with $\D$.
Deletions maintain the following invariant as they propagate:
if cell $A[i] = \tup{a, b, \D}$ is ``locked'' by a deletion (i.e., marked with $\D$),
\newlength{\strutheight}
\settoheight{\strutheight}{\strut}
\begin{adjustbox}{valign=T,raise=\strutheight,minipage={\textwidth}}
\begin{wrapfigure}[47]{r}{0.45\textwidth}
\vspace{-12pt}
\centering
\begin{subfigure}[t]{0.45\textwidth}
\centering \small
\begin{tabular}{ r|p{3ex}|p{2.5ex}|p{2.5ex}|p{2.5ex}|p{2.5ex}|p{3ex}| } 
 \cline{2-7}
  \emph{mark} & \centering $\ldots$ & \centering $\Stable$ & \centering $\Stable$ & \centering $\Stable$ & \centering $\Stable$ & {\centering $\ldots$} \\
  \cline{2-7}
  \emph{lookahead} & \centering $\ldots$ & \centering $b$ & \centering $c$ & \centering $d$ & \centering $e$ & {\centering $\ldots$} \\
  \cline{2-7}
  \emph{value} & \centering $\ldots$ & \centering $a$ & \centering $b$ & \centering $c$ & \centering $d$ & {\centering $\ldots$} \\
 \cline{2-7}
\end{tabular}
\\
\vspace{1ex}
\begin{tabular}{ r|p{3ex}|p{2.5ex}|p{2.5ex}|p{2.5ex}|p{2.5ex}|p{3ex}| } 
 \multicolumn{1}{c}{} & \multicolumn{1}{c}{} & \multicolumn{1}{c}{\color{blue}$\LL$} & \multicolumn{1}{c}{} & \multicolumn{1}{c}{} & \multicolumn{1}{c}{} & \multicolumn{1}{c}{} \\
 \multicolumn{1}{c}{} & \multicolumn{1}{c}{} & \multicolumn{1}{c}{\color{blue}$\SC$} & \multicolumn{1}{c}{} & \multicolumn{1}{c}{} & \multicolumn{1}{c}{} & \multicolumn{1}{c}{} \\
 \multicolumn{1}{c}{} & \multicolumn{1}{c}{} & \multicolumn{1}{c}{\color{blue}$\blacktriangledown$} & \multicolumn{1}{c}{} & \multicolumn{1}{c}{} & \multicolumn{1}{c}{} & \multicolumn{1}{c}{} \\
 \cline{2-7}
 \emph{mark} & \centering $\ldots$ & \centering \textcolor{red}{\boldmath{$\D$}} & \centering $\Stable$ & \centering $\Stable$ & \centering $\Stable$ & {\centering $\ldots$} \\
 \cline{2-7}
 \emph{lookahead} & \centering $\ldots$ & \centering $b$ & \centering $c$ & \centering $d$ & \centering $e$ & {\centering $\ldots$} \\
 \cline{2-7}
 \emph{value} & \centering $\ldots$ & \centering $a$ & \centering $b$ & \centering $c$ & \centering $d$ & {\centering $\ldots$} \\
 \cline{2-7}
\end{tabular}
\\
\vspace{1ex}
\begin{tabular}{ r|p{3ex}|p{2.5ex}|p{2.5ex}|p{2.5ex}|p{2.5ex}|p{3ex}| } 
 \multicolumn{1}{c}{} & \multicolumn{1}{c}{} & \multicolumn{1}{c}{\color{Green}$\LL$} & \multicolumn{1}{c}{} & \multicolumn{1}{c}{} & \multicolumn{1}{c}{} & \multicolumn{1}{c}{} \\
 \multicolumn{1}{c}{} & \multicolumn{1}{c}{} & \multicolumn{1}{c}{} & \multicolumn{1}{c}{\color{orange}$\LL$} & \multicolumn{1}{c}{} & \multicolumn{1}{c}{} & \multicolumn{1}{c}{} \\
 \multicolumn{1}{c}{} & \multicolumn{1}{c}{} & \multicolumn{1}{c}{\color{Green}$\VL$} & \multicolumn{1}{c}{} & \multicolumn{1}{c}{} & \multicolumn{1}{c}{} & \multicolumn{1}{c}{} \\
 \multicolumn{1}{c}{} & \multicolumn{1}{c}{} & \multicolumn{1}{c}{} & \multicolumn{1}{c}{\color{orange}$\SC$} & \multicolumn{1}{c}{} & \multicolumn{1}{c}{} & \multicolumn{1}{c}{} \\
 \multicolumn{1}{c}{} & \multicolumn{1}{c}{} & \multicolumn{1}{c}{} & \multicolumn{1}{c}{\color{orange}$\blacktriangledown$} & \multicolumn{1}{c}{} & \multicolumn{1}{c}{} & \multicolumn{1}{c}{} \\
 \cline{2-7}
 \emph{mark} & \centering $\ldots$ & \centering $\D$ & \centering \textcolor{red}{\boldmath{$\D$}} & \centering $\Stable$ & \centering $\Stable$ & {\centering $\ldots$} \\
 \cline{2-7}
 \emph{lookahead} & \centering $\ldots$ & \centering $b$ & \centering $c$ & \centering $d$ & \centering $e$ & {\centering $\ldots$} \\
 \cline{2-7}
 \emph{value} & \centering $\ldots$ & \centering $a$ & \centering \textcolor{red}{\boldmath{$c$}} & \centering $c$ & \centering $d$ & {\centering $\ldots$} \\
 \cline{2-7}
\end{tabular}
\\
\vspace{1ex}
\begin{tabular}{ r|p{3ex}|p{2.5ex}|p{2.5ex}|p{2.5ex}|p{2.5ex}|p{3ex}| }
 \multicolumn{1}{c}{} & \multicolumn{1}{c}{} & \multicolumn{1}{c}{\color{Green}$\SC$} & \multicolumn{1}{c}{} & \multicolumn{1}{c}{} & \multicolumn{1}{c}{} & \multicolumn{1}{c}{} \\
 \multicolumn{1}{c}{} & \multicolumn{1}{c}{} & \multicolumn{1}{c}{\color{Green}$\blacktriangledown$} & \multicolumn{1}{c}{} & \multicolumn{1}{c}{} & \multicolumn{1}{c}{} & \multicolumn{1}{c}{} \\
 \cline{2-7}
 \emph{mark} & \centering $\ldots$ & \centering \textcolor{red}{\boldmath{$\Stable$}} & \centering $\D$ & \centering $\Stable$ & \centering $\Stable$ & {\centering $\ldots$} \\
 \cline{2-7}
 \emph{lookahead} & \centering $\ldots$ & \centering \textcolor{red}{\boldmath{$c$}} & \centering $c$ & \centering $d$ & \centering $e$ & {\centering $\ldots$} \\
 \cline{2-7}
 \emph{value} & \centering $\ldots$ & \centering $a$ & \centering $c$ & \centering $c$ & \centering $d$ & {\centering $\ldots$} \\
 \cline{2-7}
\end{tabular}

\caption{\small The initial write of $\del(b)$ 
into a cell $A[i] = \tup{a, b, \Stable}$,
followed by one propagation step. The sequence of $\LL,\VL$ and $\SC$ operations is depicted from top to bottom.}
\end{subfigure}
\begin{subfigure}[t]{0.45\textwidth}
\centering \small
\begin{tabular}{ r|p{3ex}|p{2.5ex}|p{2.5ex}|p{2.5ex}|p{2.5ex}|p{3ex}|  }
 \cline{2-7}
 \emph{mark} & \centering $\ldots$ & \centering $\Stable$ & \centering $\D$ & \centering \textcolor{red}{\boldmath{$\D$}} & \centering $\Stable$ & {\centering $\ldots$} \\
 \cline{2-7}
 \emph{lookahead} & \centering $\ldots$ & \centering $c$ & \centering $c$ & \centering $d$ & \centering $e$ & {\centering $\ldots$} \\
 \cline{2-7}
 \emph{value} & \centering $\ldots$ & \centering $a$ & \centering $c$ & \centering \textcolor{red}{\boldmath{$d$}} & \centering $d$ & {\centering $\ldots$} \\
 \cline{2-7}
\end{tabular}
\\
\vspace{2ex}
\begin{tabular}{ r|p{3ex}|p{2.5ex}|p{2.5ex}|p{2.5ex}|p{2.5ex}|p{3ex}|  }
 \cline{2-7}
 \emph{mark} & \centering $\ldots$ & \centering $\Stable$ & \centering \textcolor{red}{\boldmath{$\Stable$}} & \centering $\D$ & \centering $\Stable$ & {\centering $\ldots$} \\
 \cline{2-7}
 \emph{lookahead} & \centering $\ldots$ & \centering $c$ & \centering \textcolor{red}{\boldmath{$d$}} & \centering $d$ & \centering $e$ & {\centering $\ldots$} \\
 \cline{2-7}
 \emph{value} & \centering $\ldots$ & \centering $a$ & \centering $c$ & \centering $d$ & \centering $d$ & {\centering $\ldots$} \\
 \cline{2-7}
\end{tabular}
\\
\vspace{2ex}
\begin{tabular}{ r|p{3ex}|p{2.5ex}|p{2.5ex}|p{2.5ex}|p{2.5ex}|p{3ex}|  }
 \cline{2-7}
 \emph{mark} & \centering $\ldots$ & \centering $\Stable$ & \centering $\Stable$ & \centering $\D$ & \centering $\Stable$ & {\centering $\ldots$} \\
 \cline{2-7}
 \emph{lookahead} & \centering $\ldots$ & \centering $c$ & \centering $d$ & \centering $d$ & \centering $e$ & {\centering $\ldots$} \\
 \cline{2-7}
 \emph{value} & \centering $\ldots$ & \centering $a$ & \centering $c$ & \centering $d$ & \centering \textcolor{red}{\boldmath{$\bot$}} & {\centering $\ldots$} \\
 \cline{2-7}
\end{tabular}
\\
\vspace{2ex}
\begin{tabular}{ r|p{3ex}|p{2.5ex}|p{2.5ex}|p{2.5ex}|p{2.5ex}|p{3ex}|  }
 \cline{2-7}
  \emph{mark} & \centering $\ldots$ & \centering $\Stable$ & \centering $\Stable$ & \centering \textcolor{red}{\boldmath{$\Stable$}} & \centering $\Stable$ & {\centering $\ldots$} \\
 \cline{2-7}
 \emph{lookahead} & \centering $\ldots$ & \centering $c$ & \centering $d$ & \centering \textcolor{red}{\boldmath{$\bot$}} & \centering $e$ & {\centering $\ldots$} \\
 \cline{2-7}
 \emph{value} & \centering $\ldots$ & \centering $a$ & \centering $c$ & \centering $d$ & \centering $\bot$ & {\centering $\ldots$} \\
 \cline{2-7}
\end{tabular}

\caption{\small The rest of the propagation steps, omitting the $\LL,\VL$ and $\SC$ operations. 
The propagation of the delete operation completes upon encountering element $e$ where $h(e) = i+4$, that is, element $e$ is already in its hash location.
}
\end{subfigure}

\caption{\small Depiction of the deletion process of element $b$, initiated by operation $\del(b)$.}
\label{fig:llsc_delete}
\end{wrapfigure}
\strut{}
and the subsequent cell is stable, $A[i+1] = \tup{c, d, \Stable}$,
then the \emph{value} slot of $A[i+1]$
(i.e., the value $c$) is either $\bot$ or redundant:
in the latter case, it is either the target of the deletion (immediately following the initial write),
or we have already copied
it into cell $A[i]$, so that $a = c$ (following subsequent propagation steps).
Since cell $A[i+1]$ is stable,
its \emph{lookahead} matches the \emph{value} slot of the next cell,
$A[i+2] = \tup{d, e, \ast}$.
Now there are three cases:
\begin{itemize}[nolistsep,noitemsep]
    \item
    If $d = \bot$, then we have reached the end of the run,
    and we need to set $A[i] = \tup{ a, \bot, \Stable}, A[i+1] = \tup{\bot, \bot, \Stable}$.
    The delete operation is then done propagating.
    \item
    If $d \neq \bot$ and $h(d) = i+2$,
    then element $d$ is already in its hash location, $A[i+2]$,
    and does not need to be shifted backwards.
    In this case we need to set $A[i] = \tup{a, \bot, \Stable}, A[i+1] = \tup{\bot, d, \Stable}$,
    ``puncturing'' the run
    and completing the propagation of the delete operation.
    However, the process that punctured the run is not allowed to return immediately, for reasons that we explain below.
    \item
    Finally, if $d \neq \bot$ and $h(d) \neq i+2$,
    then element $d$ should be shifted one step back (closer to $h(d)$), 
    from $A[i+2]$ to $A[i+1]$.
    In this case we need to set
    $A[i] = \tup{a, d, \Stable}, A[i+1] = \tup{d, d, \D}$, duplicating element $d$.
    Notice that the invariant is maintained:
    following the update,
    the \emph{value} slot of cell $A[i+2] = \tup{d, e, \ast}$,
    which remains untouched,%
    \footnote{By the current operation. Other operations may make their initial write into cell $i+2$,
    but if they do, they update the \emph{lookahead} slot, not the \emph{value} slot.}
    is redundant,
    because we  copied $d$ backwards into cell $A[i+1]$.
    We duplicate element $d$ instead of deleting it from the \emph{lookahead} slot to preserve the next invariant: The element in the \emph{lookahead} slot is either consistent with the next cell, or it directly precedes the element in the \emph{value} slot of the next cell in a table constructed according to Robin Hood hashing, containing all the elements in $A$.
    
    In this final case the deletion is not done propagating:
    we still need to delete $d$ from cell $A[i+2]$,
    and we may also need to shift subsequent elements one step back, closer to their hash locations.
    \end{itemize}
The changes to cells $A[i]$ and $A[i+1]$ are done in a manner similar to the way insertions are propagated,
\end{adjustbox}
by interleaving $\LL$, $\VL$ and $\SC$ operations
so that we
first ``lock'' cell $A[i+1]$ and update its contents at
the same time,
then ``release'' cell $A[i]$ and update its contents at the same time.


\paragraph{Puncturing a run.}
 One delicate point is that if a process $p$ punctures
 a run $R$ by splitting it into two runs $R_1, R_2$
 with an empty cell between them,
 and then $p$ returns (upon ``reaching the end of the run'' $R_1$), then
some operations may become stranded in the second part, $R_2$,
with no process working to propagate them.
This can happen, for example, if prior to the puncture
a $\lookup$ operation helped another operation $o$
and pushed it into $R_2$, but then found its target and returned prior to completing operation $o$.
The process $q$ that originally invoked $o$ may lag behind, so that it reaches the end of $R_1$ only after the run is punctured. In this case, if $q$ returns, operation $o$ will be stranded with no process working in $R_2$.

 To avoid this scenario, 
 it suffices to have any process $p$ that punctures a run continue into the second part and propagate all operations it finds there.
 However, this is also problematic, because we cannot allow a situation where some operations are incomplete, and a process $p$ that invoked a $\lookup$ operation is the \emph{only} process propagating them: this violates SQHI.
 Therefore we distinguish between two cases:
 if a process $q$ whose original operation is an $\insr$ or $\del$
 punctures a run,
 it continues into the second part.
 On the other hand, if a process $q$ whose original operation is a $\lookup$ encounters a location where it needs to puncture a run in order to help propagate a delete operation,
 it does not do so.
 This can cause the indication that an element is not in the table to be ``split'' across two cell, rather than just one, by an ongoing $\insr$ operation that cannot overtake an ongoing $\del$ operation that needs to puncture a run to complete. 
 To address this, We design a
 mechanism that allows $\lookup$ operations to identify that an element is not in the table based on the values of two cells instead of just one
 (Lines~\ref{lin:lookup-split}--\ref{lin:lookup-elementnotpresent3}).


\paragraph{Restarting an operation.}
There are two scenarios that cause an operation $o(x)$ to restart, and both occur prior to a successful initial write into the table.
The first is if  $o$ is an insertion or deletion that has found the ``correct'' location for element $x$,
but the initial write into the table, which is done using $\SC$, fails.
This indicates that the operation may need to re-position itself, as the ``correct'' location for element $x$ may have changed, and we do this by restarting the operation.
The second scenario occurs during the lookup stage,
when the process searches for element $x$:
if the process reads a cell $A[i] = \tup{a, b, \ast}$
such that $b \prio{>}{i} x$, but the next cell is $A[i+1] = \tup{c, d, \ast}$ such that $c \prio{<}{{i+1}} x$, then there is a chance that element $x$ was ahead of location $i$ when $A[i]$ was read,
but due to concurrent deletions, it was moved behind cell $i$ by the time $A[i+1]$ was read.
To avoid false negatives, the process restarts its lookup from position $h(x) - 1$.

We note that restarting from scratch is  only for simplicity;
instead, we can move backwards cell-by-cell until reaching a cell $A[j] = \tup{a, b, \ast}$ such that either $j = h(x) - 1$ or $a \prio{>}{j} x$, and then resume forward movement. Each step back can be blamed on a concurrent deletion that moved element $x$ one step backwards, so this is more efficient in workloads with low contention of deletes.



%% file: impl-llsc.tex
\section{Code and Proof for the History-Independent Hash Table}
\label{sec:impl}

The pseudocode for the $\insr$, $\del$ and $\lookup$ operations is presented in Algorithm~\ref{alg:insert}, Algorithm~\ref{alg:delete}, and Algorithm~\ref{alg:lookup}, respectively.
Procedure $\helpop$, presented in Algorithm~\ref{alg:help}, describes the code for helping propagate an operation to the next cell by some process. 
Procedure $\prop$, presented in Algorithm~\ref{alg:prop}, describes the code for ensuring that an operation completes its propagation.

\begin{algorithm}[!tb]
    \begin{algorithmic}[1]
         \Statex $\insr(v)$:
         \State $i \gets h(v) - 1$
         \label{lin:start-insert}
         \State $\tup{a, b, M} \gets \LL(A[i])$
         \label{lin:insr-readcurr1}
         \State $\textit{first} \gets \textit{true}$
         \Repeat
         \label{lin:insert-startloop}
            \If{$a = v$ or ($b = v$ and ($M \neq \D$ or $h(b) \neq i + 1$))}
            \label{lin:insr-element-present}
                \State \algorithmicreturn{} \textit{false}
            \Comment{$v$ already in the table}
            \EndIf
            \If{$M \neq \Stable$}
            \Comment{$A[i]$ is unstable}
            \label{lin:insr-helpop}
                \State \helpop$(i)$
                \State \textbf{goto} Line~\ref{lin:insert-readnextcell}
            \EndIf
            \If{$b = \bot$ or $v \prio{>}{i+1} b$}
            \Comment{Insert position found}
            \label{lin:insert-loc}
                \If{$\SC\parens*{A[i], \tup{a, v, \I}}$}
                \label{lin:sc-insert}
                 \State \prop$\parens*{h(v),i, \set{\I}}$
                    \label{lin:propinsr}
                    \State \algorithmicreturn{} \textit{true}
                \EndIf
                \State \textbf{goto} Line~\ref{lin:start-insert}
                \label{lin:insr-restartloop1}
            \EndIf
         \State $i \gets i + 1$
         \label{lin:next-cell}
         \If{$\neg \textit{first}$ and $i = h(v)$} \abort()
         \Comment{No empty cell to insert $v$}
         \label{lin:nospace}
         \EndIf
         \State $\textit{first} \gets \textit{false}$
         \State $\tup{a, b, M} \gets \LL(A[i])$
         \label{lin:insert-readnextcell}
        \Until{$v \prio{>}{i} a$}
        \label{lin:insert-until}
        \State \textbf{goto} Line~\ref{lin:start-insert}
        \label{lin:insr-restartloop2}
        \algstore{insert}
    \end{algorithmic}
    \caption{Pseudocode for $\insr$}
    \label{alg:insert}
\end{algorithm}

\begin{algorithm}[!tb]
    \begin{algorithmic}[1]
        \algrestore{insert}
         \Statex $\del(v)$:
         \State $i \gets h(v)-1$
         \label{lin:start-delete}
         \State $\tup{a, b, M} \gets \LL(A[i])$
         \State $\textit{first} \gets \textit{true}$
         \label{lin:del-readcurr1}
         \Repeat
         \label{lin:del-startloop}
            \If{($i = h(v)$ and $v \prio{>}{i} a$) or ({$a \prio{>}{i} v \prio{>}{i+1} b$ and ($M = \Stable$ or $h(b)\neq i + 1$}))}
            \State \algorithmicreturn{} \textit{false}
            \Comment{$v$ is not in the table}
            \label{lin:del-elementnotpresent1}
         \EndIf
         \If{$M \neq \Stable$}
            \label{lin:del-helpop}
            \Comment{$A[i]$ is unstable}
                \State \helpop$(i)$
                    \State \textbf{goto} Line~\ref{lin:del-readnextcell}
            \EndIf
            \If{$a = v$}
            \label{lin:del-vinval}
            \Comment{Step back one cell to locate $v$ in the \emph{lookahead} slot}
                \State $i \gets i - 1$
                \State \textbf{goto} Line~\ref{lin:del-readnextcell}
            \EndIf
            \If{$b = v$} 
            \label{lin:del-loc}
            \Comment{Delete position found}
                \If{$\SC(A[i], \tup{a, v, \D})$}
                \label{lin:sc-del}
                    \State \prop$\parens*{h(v), i, \set{\D}}$
                    \label{lin:propdel}
                    \State \algorithmicreturn{} \textit{true}
            \EndIf
            \State \textbf{goto} Line~\ref{lin:start-delete}
            \label{lin:del-restartloop1}
         \EndIf
         \State $i \gets i + 1$
         \label{lin:del-next-cell}
         \If{$\neg\textit{first}$ and $i = h(v)$} \algorithmicreturn{} \textit{false}
         \Comment{The loop iterated through all cells in the table}
         \label{lin:del-elementnotpresent2}
         \EndIf
         \State $\textit{first} \gets \textit{false}$
         \State $\tup{a, b, M} \gets \LL(A[i])$
         \label{lin:del-readnextcell}
        \Until{$i \neq h(v)$ and $v \prio{>}{i} a$}
        \label{lin:until-del}
        \State \textbf{goto} Line~\ref{lin:start-delete}
        \label{lin:del-restartloop2}
        \algstore{delete}
    \end{algorithmic}
    \caption{Pseudocode for $\del$}
    \label{alg:delete}
\end{algorithm}

\begin{algorithm}[!tb]
    \begin{algorithmic}[1]
        \algrestore{delete}
         \Statex $\lookup(v)$:
         \State $i \gets h(v) - 1$
         \label{lin:start-lookup}
         \State $\tup{a, b, M} \gets \LL(A[i])$
         \label{lin:lookup-readcurr1}
         \State $\textit{first} \gets \textit{true}$
         \Repeat
         \label{lin:lookup-startloop}
            \If{$a = v$ or ($b = v$ and ($M \neq \D$ or $h(b) \neq i + 1$))} 
            \label{lin:lookup-elementfound}
            \State \algorithmicreturn{} \textit{true}
            \Comment{$v$ is in the table}
         \EndIf
         \If{($i = h(v)$ and $v \prio{>}{i} a$) or ({$a \prio{>}{i} v \prio{>}{i+1} b$ and 
         ($M = \Stable$ or $h(b)\neq i + 1$}))}
            \State \algorithmicreturn{} \textit{false}
            \Comment{$v$ is not in the table}
            \label{lin:lookup-elementnotpresent1}
         \EndIf
         \If{$M = \I$}
         \label{lin:lookup-split}
         \Comment{Check if the indication that $v$ is not in the table is split across two cells}
            \State $\tup{c, \ast, \ast} \gets \LL(A[i+1])$
            \label{lin:lookupllAi+1}
            \If{$b \prio{>}{i} v \prio{>}{i+1} c$ and $h(b) \neq i + 1$ and $\VL(A[i])$} 
            \State \algorithmicreturn{} \textit{false}
            \label{lin:lookup-elementnotpresent3}
            \EndIf
         \EndIf
         \If{$M \neq \Stable$}
         \helpop$(i)$
            \label{lin:lookup-helpop}
        \EndIf
         \State $i \gets i + 1$
         \If{$\neg \textit{first}$ and $i = h(v)$} \algorithmicreturn{} \textit{false}
         \label{lin:lookup-elementnotpresent2}
         \Comment{The loop iterated through all cells in the table}
         \EndIf
         \State $\textit{first} \gets \textit{false}$
         \State $\tup{a, b, M} \gets \LL(A[i])$
        \Until{$i \neq h(v)$ and $v \prio{>}{i} a$}
        \label{lin:until-lookup}
        \State \textbf{goto} Line~\ref{lin:start-lookup}
        \algstore{lookup}
    \end{algorithmic}
    \caption{Pseudocode for $\lookup$}
    \label{alg:lookup}
\end{algorithm}

\begin{algorithm}[!tb]
    \begin{algorithmic}[1]
        \algrestore{lookup}
        \Statex \helpop$(i)$:
        \State $\tup{a, b, M_1} \gets \LL(A[i])$
        \label{lin:llAi}
        \State $\tup{c, d, M_2} \gets \LL(A[i+1])$
        \label{lin:llAi+1}
        \If{$M_1 \neq \Stable$}
            \While{$M_2 \neq \Stable$ 
            and ($M_1 \neq \D$ or $M_2 \neq \D$ or $b = c$)  
            and \\[]
            \quad\quad\quad\quad ($M_1 \neq \I$ or $b \neq c$)
            and ($M_1 \neq \D$ or $c \neq \bot$)}
            \label{lin:helpop-while}
            \State $\tup{a, b, M_1} \gets \tup{c, d, M_2}$
            \State $i \gets i + 1$
            \State $\tup{c, d, M_2} \gets \LL(A[i+1])$
            \label{lin:llAi+12}
        \EndWhile
        \Comment{($M_1 \neq \Stable$ and $M_2 = \Stable$) or ($M_1 = M_2 = \D$ and $b \neq c$) 
        \Statex \hfill 
        or ($M_1 = \I$ and $b = c$) or ($M_1 =\D$ and $c = \bot$)}
        \If{$M_1 = \I$ and $\VL(A[i])$}
        \label{lin:helpop-I-cond}
            \State $\tup{x,y,M_3} \gets \LL(A[i-1])$
            \label{lin:llAi-1I}
            \If{$M_3 = \I$ and $y = a$ and $\VL(A[i])$}
            \label{lin:sc-I1-cond}
            \Comment{Ensure $A[i-1]$ is unlocked}
                \State $\SC(A[i-1], \tup{x,y,\Stable})$
                \label{lin:sc-I1}
            \EndIf
            \If{$c \prio{>}{i+1} b$ and not a $\lookup$ operation} \abort()
            \label{lin:stuckinsr}
            \Comment{Table is full}
            \ElsIf{$b = c$}
                $\SC(A[i], \tup{a,b,\Stable})$
                \Comment{$A[i+1]$ is already locked or propagation ended}
                \label{lin:sc-I2}
            \ElsIf{$c = \bot$}
                $\DSC(i+1, \tup{b,d,\Stable}, i, \tup{a,b,\Stable})$
                \label{lin:dsc-I1}
                \Comment{End of the propagation}
            \Else{}
                $\DSC(i+1, \tup{b,c,\I}, i, \tup{a,b,\Stable})$
                \label{lin:dsc-I2}
                \Comment{Propagate insert to $A[i+1]$ and unlock $A[i]$}
            \EndIf
        \ElsIf{$s_1 = \D$ and $\VL(A[i])$}
            \State $\tup{x,y,M_3} \gets \LL(A[i-1])$
            \label{lin:llAi-1D} 
            \If{$M_3 = \D$ and $y \neq a$ and $\VL(A[i])$}
            \label{lin:sc-D1-cond}
            \Comment{Ensure $A[i-1]$ is unlocked}
                \State $\SC(A[i-1], \tup{x,a,\Stable})$
                \label{lin:sc-D1}  
            \EndIf
            \If{$c = \bot$ or $M_2 = \D$}
            \label{lin:sc-D2-cond}
                \State $\SC(A[i], \tup{a,c,\Stable})$
                \label{lin:sc-D2}
                \label{lin:sc-D3}
                \Comment{$A[i+1]$ is already locked or propagation ended}
            \ElsIf{$M_2 = \Stable$} 
            \Comment{Lock $A[i+1]$ and update $A[i]$}
            \label{lin:helpop-stableD}
                \If{$d\neq \bot$ and $h(d)\neq i+2$}
                    \State $\DSC(i+1, \tup{d,d,\D}, i, \tup{a,d,\Stable})$
                    \label{lin:dsc-D1}
                \ElsIf{not a $\lookup$ operation}
                \label{lin:no-lookup}
                \If{$\DSC(i+1, \tup{\bot,d,\Stable}, i, \tup{a,\bot,\Stable})$}
                \Comment{End of the propagation}
                \label{lin:dsc-D2}
                        \If{$d \neq \bot$}
                        \prop$(i+1,i+2,\set{\I,\D})$
                        \Comment{Punctured a run}
                    \EndIf
                \EndIf
            \EndIf
            \EndIf
        \EndIf
        \EndIf
        \Statex
        \Statex \DSC$(i, \mathit{tup}_1, j, \mathit{tup}_2)$:
        \Comment{Write $\mathit{tup}_1$ to $A[i]$ and if successful, write $\mathit{tup}_2$ to $A[j]$}
            \If{$\SC(A[i],\mathit{tup}_1)$}
            \label{lin:dsc-sc}
                \State $\SC(A[j], \mathit{tup}_2)$
                \State \algorithmicreturn{} \textit{true}
            \Else{}
                \State $\mathit{tup} \gets \LL(A[i])$
                \If{$\mathit{tup}.\mem = \mathit{tup}_1.\mem$}
                \label{lin:dsc-checksucc}
                    $\SC(A[j], \mathit{tup}_2)$
                \EndIf
            \EndIf
        \algstore{helper}
    \end{algorithmic}
    \caption{Pseudocode for $\helpop$}
    \label{alg:help}
\end{algorithm}

\begin{algorithm}[!tb]
    \begin{algorithmic}[1]
        \algrestore{helper}
        \Statex \prop$(i, j, A)$:
        \Comment{$A\subseteq \set{\I,\D}$}
        \Repeat
        \State $\tup{a,b,M}, prev \gets A[j]$
        \While{$\tup{a,b,M} = prev$ and $M \in A$}
        \label{lin:prop-helpcond}
        \Statex \Comment{Continue while $A[j]$ does not change and contains a 
        delete or insert operation}
            \State \helpop$(j)$
            \State $prev \gets \tup{a,b,M}$
            \State $\tup{a,b,M} \gets A[j]$
        \EndWhile
        \State $j \gets j + 1$
        \Until{($a = \bot$ or ($M = \Stable$ and $b = \bot$)) or $j = i$}
        \label{lin:prop-stopcond}
        \Statex \Comment{Stop upon reaching an empty cell or
        back to the start}
    \end{algorithmic}
    \caption{Pseudocode for $\prop$}
    \label{alg:prop}
\end{algorithm}

The hash table $A$ is of size $\ncell$, and is initialized to $A[i] = \tup{\bot,\bot,\Stable}$, for every $0\leq i < \ncell$.
For a cell $A[i]$, $0\leq i < \ncell$, that contains the value $\tup{a,b,M}$, denote $A[i].\mem = a$, $A[i].\lookahead = b$ and $A[i].\flag = M$.



It is easy to see that Robin Hood hashing maintains the ordering invariant.
Furthermore, it has been shown that
the ordering invariant guarantees the unique memory representation for any given set of elements~\cite{ShunBlellochSPAA14}.
For a set $P$ of elements from $\CalU$ of size at most $\ncell$, let $\can(P)$ be the canonical memory representation induced by Robin Hood hashing in a table of size $\ncell$.
This also holds for \emph{multisets}, where elements can appear multiple times.
Specifically, in Robin Hood hashing, duplicates of the same element are stored in consecutive cells.
For a multiset $P$ of elements from $\CalU$ of size at most $\ncell$, including multiplicities, let $\canmult(P)$ be the canonical memory representation induced by Robin Hood hashing in a table of size $\ncell$. Note that if set $P$ does not include multiplicities of the same element, $\can(P) = \canmult(P)$.
Let $A|_\mem$ denote the projected array from $A$, which includes only the 
$\mem$ slot at each cell.

\begin{lemma}[\cite{BlellochGo07}]
\label{lem:order-invariant}
    Let $P$ be the multiset of elements in $A|_\mem$, if $A|_\mem$ satisfies the ordering invariant, then $A|_\mem = \canmult(P)$.
\end{lemma}

The proof relies on the following key invariant that extends the ordering invariant to also account for the \emph{lookahead} slot and element repetition.

\begin{invariant}[Extended Ordering Invariant]
\label{inv:ext-order-invar}
For every $0 \leq i < m$:
\begin{enumerate}[label=(\alph*)]
    \item If $A[i].\mem = v$, then for any $h(v) \leq j < i$, $A[j].\mem \prio{\geq}{j} v$.
    \label{inv:ext-order-invar1}
    \item $A[i].\lookahead \prio{\geq}{i+1} A[i+1].\mem$ or $h(A[i].\lookahead) = i+1$.
    \label{inv:ext-order-invar2}

    \item Either $A[i].\mem \prio{\geq}{i} A[i].\lookahead$ or $h(A[i].\lookahead) = i+1$ and $A[i].\lookahead \prio{\geq}{i+1} A[i+1].\mem$.
    \label{inv:ext-order-invar3}
\end{enumerate}
\end{invariant}

\invenumref{inv:ext-order-invar}{1} is the ordering invariant for the \emph{value} slots in a table that can store multiple copies of the same element (recall that deletions create duplications).
\invenumref{inv:ext-order-invar}{2} and \invenumref{inv:ext-order-invar}{3} ensure that the ``correct position'' of the element in the \emph{lookahead} slot is the next cell. 
\invenumref{inv:ext-order-invar}{3} ensures that an element with higher priority displaced the element in the \emph{lookahead} slot, while \invenumref{inv:ext-order-invar}{2} verifies that the \emph{lookahead} can displace the value of the next cell. The only exception to this is when all the cells in the table are occupied, in which case it may be that $h(A[i].\lookahead) = i+1$.
\invref{inv:ext-order-invar} allows us to use the properties of Robin Hood described earlier.

Cell $i$ is \emph{stable} if $A[i].\flag = \Stable$, and \emph{unstable} otherwise.
An unstable cell is associated with exactly one operation, where a cell $A[i] = \tup{\ast,\ast,\I}$ is associated with an insert operation, and a cell $A[i] = \tup{\ast,\ast,\D}$ is associated with a delete operation.
A cell that becomes unstable after a successful 
$\SC$ in Line~\ref{lin:sc-insert} is associated with the insert operation that performs this successful 
$\SC$.
A cell that becomes unstable after a successful 
$\SC$ in Line~\ref{lin:sc-del} is associated with the delete operation that performs this successful 
$\SC$.
An operation propagates from cell $A[i]$ to cell $A[i+1]$ because of a successful $\SC$ in $\DSC$ called in Line~\ref{lin:dsc-I1}, \ref{lin:dsc-I2}, \ref{lin:dsc-D1} or~\ref{lin:dsc-D2}, and it either becomes unstable with the operation or remains stable.
Stable cell $A[i+1]$ becomes unstable after a successful 
$\SC$ in a $\DSC$ called in Line~\ref{lin:dsc-I2} or Line~\ref{lin:dsc-D1}.
In both cases, 
$A[i+1]$ becomes unstable based on the value of an unstable cell $A[i]$, hence, $A[i+1]$ becomes associated with the same operation as $A[i]$.
For simplicity, we assume that if a cell becomes unstable with the same operation again, it becomes unstable with a new different operation.
Later (Lemma~\ref{lem:finishprop}), we show that this cannot happen and a cell becomes unstable with a particular operation only once.

The next invariant describes the structure of stable and unstable cells.

\begin{invariant}
\label{inv:stable-unstable}
For every $0 \leq i < m$:
\begin{enumerate}[label=(\alph*)]
    \item If $A[i].\flag \neq \D$, $A[i].\mem \neq A[i].\lookahead$ or $A[i].\mem = \bot$.
    \label{inv:stable-unstable1}

    \item If $A[i].\flag = \Stable$, then $A[i].\lookahead = A[i+1].\mem$.
    \label{inv:stable-unstable2}

    \item If $A[i].\flag \in \set{\I,\D}$, then $A[i].\lookahead \neq \bot$, and either $A[i].\mem \neq \bot$ or $h(A[i].\lookahead) = i+1$.
    \label{inv:stable-unstable2.5}
    
    \item If $A[i].\flag = \D$, and if the operation propagated to cell $A[i+1]$, then $A[i].\lookahead \prio{>}{i+1} A[i+1].\mem$, 
    where:
    \begin{itemize}
        \item If $A[i+1].\flag \in \set{\Stable, \I}$, or if $A[i+1].\flag = \D$ but $A[i+1]$ is unstable with a different operation, then $A[i+1].\mem = \bot$.

        \item If $A[i+1].\flag = \D$ and $A[i+1]$ is unstable with the same operation, then $A[i+1].\mem \neq \bot$.
    \end{itemize}
    Otherwise, $A[i].\lookahead = A[i+1].\mem$.
    \label{inv:stable-unstable3}

    \item If $A[i].\flag = \I$, and if the operation propagated to cell $A[i+1]$, then $A[i].\lookahead = A[i+1].\mem$, otherwise, $A[i].\lookahead \neq A[i+1].\mem$.
    \label{inv:stable-unstable4}
\end{enumerate}
\end{invariant}

An operation may abort if the table is full and there is no vacant cell to insert a new element.
This can happen when a process initiating an $\insr$ operation does not find a vacant cell for the new element (Line~\ref{lin:nospace}), or when an insert operation cannot be propagated to the next cell without violating the ordering invariant, indicating there is no vacant cell for the displaced element (Line~\ref{lin:stuckinsr}).
We prove correctness assuming there is always a stable empty cell available to insert a new element. 
This guarantees that no operation aborts, and if there is an unstable cell, some operation can propagate to the next cell.

\begin{assumption}
\label{asm:neverfull}
    There exists an empty cell in $A$, i.e., there exists a cell $i$ such that $A[i] = \tup{\bot,\ast,\Stable}$, where $0\leq i < \ncell$.
\end{assumption}

In Section~\ref{sec:implproof-prop} we state and prove the structure of propagating operations in the table, and this allows us to prove in Section~\ref{sec:implproof-val&order} that $A$ satisfies both Invariant~\ref{inv:ext-order-invar} and Invariant~\ref{inv:stable-unstable} throughout the execution of the algorithm.
This allows us to prove the algorithm's key properties;
In Section~\ref{sec:implproof-linearizability} we prove the hash table implementation is linearizable, in Section~\ref{sec:implproof-hi} that it is history independent, and in Section~\ref{sec:implproof-lockfree} that it is lock-free.

\subsection{Proving the Propagation Structure of Operations}
\label{sec:implproof-prop}

Next, we show that the propagation of operations follows a pattern similar to hand-over-hand locking. Assuming that $A[i]$ is locked (i.e., it contains an ongoing $\insr$ or $\del$ operation), we then lock $A[i+1]$, release $A[i]$, lock $A[i+2]$,  release $A[i+1]$, and so on.
Throughout the next section, we assume that Invariant~\ref{inv:stable-unstable} holds at any point in the execution, where \invenumref{inv:stable-unstable}{3} and \invenumref{inv:stable-unstable}{4} ensure that when a process helps propagate an operation to the next cell, it can distinguish whether the operation propagated to the next cell or not.

The next lemma follows immediately from the definition of the $\LL$ and $\VL$ operations.

\begin{lemma}
\label{lem:llllvl}
    If a process performs $\LL(A[i])$, followed by $\LL(A[j])$, then performs a $\VL(A[i])$ that returns \textit{true}, then during the execution of the $\LL(A[j])$ operation,
    the value of $A[i]$ remained unchanged since the last corresponding $\LL$ operation.
\end{lemma}

\begin{lemma}
\label{lem:propstructure}
    Assume Invariant~\ref{inv:stable-unstable} holds and cell $A[i]$ becomes unstable with operation $o$, then the only possible order of events is:
    \begin{enumerate}
        \item $A[i]$ becomes unstable with operation $o$.
        \item Operation $o$ propagates to cell $A[i+1]$.
        \item $A[i]$ stabilizes.
        \item If operation $o$ also propagates to $A[i+2]$, then the operation propagates to $A[i+2]$.
        \item If $A[i+1]$ becomes unstable with the operation $o$, then 
        $A[i+1]$ stabilizes.
    \end{enumerate}
\end{lemma}

\begin{proof}
    By the code, an operation propagates to cell 
    $A[i+1]$ after a process sees that cell $A[i]$ is unstable with the same operation.
    Assume the lemma statement does not hold and consider the first two events that happen in reverse order.
    \begin{itemize}
        \item \textbf{$A[i]$ stabilizes before the operation propagates to $A[i+1]$:}
        Assume the operation is an $\insr$ operation and consider all the possible scenarios in which a successful $\SC$ stabilizes $A[i]$.
        If the successful $\SC(A[i],\tup{\ast,\ast,\Stable})$ happens in Line~\ref{lin:sc-I1} or~\ref{lin:sc-I2}, then by Lemma~\ref{lem:llllvl},
        during the $\LL(A[i])$ or  $\LL(A[i+1])$, respectively, that precedes the successful $\SC$, we have
        $A[i] = \tup{\ast,b,\I}$ and $A[i+1] = \tup{b,\ast,\ast}$. By 
        \invenumref{inv:stable-unstable}{4},
        the operation has already propagated to $A[i+1]$. 
        If the successful $\SC$ happens in
        a $\DSC$ called in Line~\ref{lin:dsc-I1} or~\ref{lin:dsc-I2}, then the value of $A[i]$ does not change between the $\LL(A[i])$ that precedes the successful $\SC$ and the successful $\SC$.
        Thus,
        the conditions in Line~\ref{lin:dsc-sc} and Line~\ref{lin:dsc-checksucc} in the $\DSC$ procedure ensure that
        $A[i+1] = \tup{b,\ast,\ast}$, while $A[i] = \tup{\ast,b,\I}$ and unstable with the $\insr$ operation. By 
        \invenumref{inv:stable-unstable}{4},
        the operation has already propagated to $A[i+1]$.

         Assume the operation is a $\del$ operation and consider all the possible scenarios in which a successful $\SC$ stabilizes $A[i]$.
        If the successful $\SC(A[i],\tup{\ast,\ast,\Stable})$ happens in Line~\ref{lin:sc-D1} or~\ref{lin:sc-D2}, then by Lemma~\ref{lem:llllvl},
        during the $\LL(A[i])$ or  $\LL(A[i+1])$, respectively, that precedes the successful $\SC$, we have
        $A[i] = \tup{\ast,b,\D}$ and $A[i+1] = \tup{c,\ast,M}$, such that $b\neq c$. 
        This is because the conditions in Line~\ref{lin:helpop-while}, \ref{lin:sc-D1-cond} and~\ref{lin:sc-D2} ensure that either $M = \D$ and $b\neq c$, or $c = \bot$. In the latter case, by 
        \invenumref{inv:stable-unstable}{2.5}, $b\neq \bot$.
        By~\invenumref{inv:stable-unstable}{3}, the operation has already propagated to $A[i+1]$. 
        If the successful $\SC$ happens in
        a $\DSC$ called in Line~\ref{lin:dsc-D1} or~\ref{lin:dsc-D2}, then the value of $A[i]$ does not change between the $\LL(A[i])$ that precedes the successful $\SC$ and the successful $\SC$.
        Thus,
        the conditions in Line~\ref{lin:dsc-sc} and Line~\ref{lin:dsc-checksucc} in the $\DSC$ procedure ensure that
        $A[i+1] = \tup{d,\ast,\ast}$, while $A[i] = \tup{\ast,b,\D}$, such that $b\neq d$.
        This is because either $d = \bot$ and by \invenumref{inv:stable-unstable}{2.5}, $b\neq \bot$, or 
        by the condition in Line~\ref{lin:helpop-stableD} and Lemma~\ref{lem:llllvl}, $A[i+1] = \tup{c,d,\Stable}$ while $A[i] = \tup{a,b,\D}$.
        Since the condition in Line~\ref{lin:sc-D2-cond} does not hold, $c\neq \bot$, and by~\invenumref{inv:stable-unstable}{3}, $b = c$. In addition, by~\invenumref{inv:stable-unstable}{1}, $c\neq d$.
        However, by~\invenumref{inv:stable-unstable}{3},
        the operation has already propagated to $A[i+1]$.

        \item \textbf{
         If the operation also propagates to $A[i+2]$, 
         the operation propagates to $A[i+2]$ before $A[i]$ stabilizes:}
        If the operation is an $\insr$ operation, then the operation propagates to $A[i+2]$ after a successful $\SC$ in a $\DSC$ called in Line~\ref{lin:dsc-I1} or~\ref{lin:dsc-I2}.
        In the $\LL(A[i+1])$ that precedes the successful $\SC$, $A[i+1] = \tup{c,\ast,\I}$ and unstable with the $\insr$ operation.
        We already showed that $A[i+1]$ remains unstable with the operation up until the operation propagates to $A[i+2]$, thus the $\VL(A[i+1])$ in Line~\ref{lin:sc-I1-cond} returns \textit{true}.
        As $A[i]$ is still unstable, the $\LL(A[i])$ in Line~\ref{lin:llAi-1I} returns the value $\tup{\ast,b,\I}$, and by Lemma~\ref{lem:llllvl}, as the $\VL(A[i+1])$ is successful, during the $\LL(A[i])$, $A[i] = \tup{\ast,b,\I}$ and $A[i+1] = \tup{c,\ast,\I}$. As both cells are unstable with the same operation, by~\invenumref{inv:stable-unstable}{4}, $b = c$, and the condition in Line~\ref{lin:sc-I1-cond} holds.
        Thus, there is an $\SC(A[i],\tup{\ast,\ast,\Stable})$ before the successful $\SC$ that propagates the operation to $A[i+2]$.
        Since all $\SC$s on unstable cells change their value to be stable, both outcomes of the $\SC(A[i],\tup{\ast,\ast,\Stable})$ operation, either successful or unsuccessful, indicate that cell $A[i]$ stabilizes.
        
        If the operation is a $\del$ operation, then the operation propagates to $A[i+2]$ after a successful $\SC$ in a $\DSC$ called in Line~\ref{lin:dsc-D1} or~\ref{lin:dsc-D2}.
        In the $\LL(A[i+1])$ that precedes the successful $\SC$, $A[i+1] = \tup{c,\ast,\D}$ and unstable with the $\del$ operation.
        We already showed that $A[i+1]$ remains unstable with the operation up until the operation propagates to $A[i+2]$, thus the $\VL(A[i+1])$ in Line~\ref{lin:sc-D1-cond} returns \textit{true}.
        As $A[i]$ is still unstable, the $\LL(A[i])$ in Line~\ref{lin:llAi-1I} returns the value $\tup{\ast,b,\D}$, and by Lemma~\ref{lem:llllvl}, as the $\VL(A[i+1])$ is successful, during the $\LL(A[i])$, $A[i] = \tup{\ast,b,\D}$ and $A[i+1] = \tup{c,\ast,\D}$. As both cells are unstable with the same operation, by~\invenumref{inv:stable-unstable}{3}, $b \neq c$, and the condition in Line~\ref{lin:sc-D1-cond} holds.
        Thus, there is an $\SC(A[i],\tup{\ast,\ast,\Stable})$ before the successful $\SC$ that propagates the operation to $A[i+2]$.
        Since all $\SC$s on unstable cells change their value to be stable, both outcomes of the $\SC(A[i],\tup{\ast,\ast,\Stable})$ operation, either successful or unsuccessful, indicate that cell $A[i]$ stabilizes.

         \item \textbf{If $A[i+1]$ becomes unstable with the operation, $A[i+1]$ stabilizes before $A[i]$ stabilizes:}
         We showed that if $A[i+1]$ becomes unstable, then the operation also propagates to cell $A[i+2]$ before $A[i+1]$ stabilizes.
         Since we also showed that $A[i]$ stabilizes before the operation propagates to $A[i+2]$, this proves this case. \qedhere
    \end{itemize}
\end{proof}

The next lemma verifies that an operation only propagates once to a cell. If the cell becomes unstable after the propagation, this follows directly from Lemma~\ref{lem:propstructure}.
\begin{lemma}
\label{lem:proponce}
    An operation can only propagate once to any cell in the table.
\end{lemma}

\begin{proof}
    If cell $A[i+1]$ becomes unstable after the propagation of the operation, then by Lemma~\ref{lem:propstructure}, $A[i]$ stabilizes before $A[i+1]$ stabilizes. Since an unstable cell can only change to be stable, and cell $A[i]$ must be unstable for the operation to propagate to cell $A[i+1]$, the operation cannot propagate again to $A[i+1]$.
    If cell $A[i+1]$ remains stable after the propagation, then by Lemma~\ref{lem:llllvl} \invenumref{inv:stable-unstable}{3} and \invenumref{inv:stable-unstable}{4}, after the operation propagates to $A[i+1]$ and while $A[i]$ is unstable, the conditions in Line~\ref{lin:sc-I1-cond}, \ref{lin:sc-I2}, \ref{lin:sc-D1-cond} and \ref{lin:sc-D2-cond}, 
    prevent the operation to propagate again to cell $A[i+1]$.
\end{proof}

\subsection{Proving Invariant~\ref{inv:ext-order-invar} and Invariant~\ref{inv:stable-unstable}}
\label{sec:implproof-val&order}

The next property of Robin Hood hashing is the key to establishing the correctness of the
lookup stage.
It shows that if an element in location $i$ in the table has higher priority than element $v$, then all elements in locations between $h(v)$ and $i$ in the table also have higher priority than element $v$. According to the ordering invariant, this shows that $v$'s location in the table must be after $i$. 
For element $v\in \CalU$ and cell $0\leq i < \ncell$, let $\rank(v,i) = i - h(v)$ be the distance of cell $i$ from $v$'s initial hash location. 
Since we use linear probing, this is also $i$’s rank in $v$’s probing sequence.

\begin{lemma}
\label{lem:monotone-priority}
    Consider a multiset $P$ and element $v$, if
    $\canmult(P)[i] >_{p_i} v$, $0\leq i < \ncell$, then  for all $h(v)\leq j \leq i$, $\canmult(P)[j] >_{p_j} v$.
\end{lemma}

\begin{proof}
    Denote $u = \canmult(P)[i]$.
    Since $u >_{p_i} v$, either $\rank(u, i) > \rank(v, i)$ or  $\rank(u, i) = \rank(v, i)$ and $u > v$. Either way, for all $h(v) \leq j \leq i$, $\rank(u, j) \geq \rank(v, j)$.
    By the ordering invariant, for every $h(u)\leq j < i$ we have $\canmult(P)[j] \geq_{p_j} u$, which implies that $\rank(\canmult(P)[j], j) \geq \rank(u, j)$.
    Since $\rank(u, i) \geq \rank(v, i)$, $h(u) \leq h(v) \leq i$, and for every $h(v)\leq j \leq i$, then $\rank(\canmult(P)[j], j) \geq \rank(u, j) \geq \rank(v, j)$.
    If $\rank(\canmult(P)[j], j) > \rank(v, j)$, then $\canmult(P)[j] >_{p_j} v$. Otherwise, if $\rank(\canmult(S)[j], j) = \rank(v, j)$, it holds that $\rank(\canmult(P)[j], j) = \rank(u, j)$ and $\canmult(P)[j] \geq u > v$, implying that $\canmult(P)[j] >_{p_j} v$.
\end{proof}

The next lemma shows that if an element is pushed into the \emph{lookahead} slot during an insert, its position in the table should be in the next cell.

\begin{lemma}
\label{lem:consistentprio}
    If $a \prio{>}{i} b$, $0\leq i < \ncell$, and $h(a) \neq i + 1$, then $a \prio{>}{i+1} b$.
\end{lemma}

\begin{proof}
    As $a \prio{>}{i} b$, either $\rank(a,i) > \rank(b,i)$ or $\rank(a,i) = \rank(b,i)$ and $a > b$. If $\rank(a,i) > \rank(b,i)$, since $h(a) \neq i + 1$, $\rank(a,i) \neq m-1$, and
    we also get that $\rank(a,i+1) > \rank(b,i+1)$.
    If $\rank(a, i) = \rank(b, i)$, we have that $h(b) = h(a)$, and it also holds that $\rank(a, i+1) = \rank(b, i+1)$.
\end{proof}

The next lemma shows that if an element is shifted backwards during a delete, its position in the table should be in the previous cell.
\begin{lemma}
\label{lem:consistentpriosmaller}
    If $a \prio{>}{i} b$, $0\leq i < \ncell$, and $h(b) \neq i$, then $a \prio{>}{i-1} b$.
\end{lemma}

\begin{proof}
    As $a \prio{>}{i} b$, either $\rank(a,i) > \rank(b,i)$ or $\rank(a,i) = \rank(b,i)$ and $a > b$. If $\rank(a,i) > \rank(b,i)$, since $h(b) \neq i$, $\rank(b,i)\neq 0$, and we also get that $\rank(a,i-1) > \rank(b,i-1)$.
    If $\rank(a, i) = \rank(b, i)$, we have that $h(b) = h(a)$, and it also holds that $\rank(a, i-1) = \rank(b, i-1)$.
\end{proof}

The next lemma is proved by case analysis of all the possible ways to change a cell in $A$.

\begin{lemma}
\label{lem:validorder}
    At any algorithm step, $A$ satisfies Invariant~\ref{inv:ext-order-invar} and Invariant~\ref{inv:stable-unstable}.
\end{lemma}

\begin{proof}
    At the beginning $A$ is empty and both invariants hold trivially.
    Assume table $A$ satisfies the lemma statement and consider the next step that changes the table.
    This happens because of a successful $\SC$ operation:
    \begin{itemize}
        \item \textbf{Successful $\SC$ in Line~\ref{lin:sc-insert}:}
        Let $\tup{a,b,M}$ be the value returned from the $\LL(A[i])$ that precedes the successful $\SC$ operation.
        By the condition in Line~\ref{lin:insr-helpop}, $M = \Stable$.
        The successful $\SC$ changes the value of $A[i]$ to be $\tup{a,v,\I}$.
        \invenumref{inv:ext-order-invar}{1}  is trivially maintained, as no changes are made to any $\mem$ slot.
        By the condition in Line~\ref{lin:insert-loc}, $v \prio{>}{i+1} b$.
        By \invenumref{inv:stable-unstable}{2}, during the successful $\SC$, $A[i+1] = \tup{b,\ast,\ast}$.
        Either $h(v) = i+1$, or since this is not the first iteration in the loop, by the conditions in Line~\ref{lin:insert-until} and Line~\ref{lin:insr-element-present}, $v \prio{<}{i} a$.
        This proves that Invariant~\ref{inv:ext-order-invar} continues to hold.
         Since $A[i]$ is the first cell to become unstable with this insert operation and the $\mem$ slot does not change, this also proves Invariant~\ref{inv:stable-unstable}.

        \item \textbf{Successful $\SC$ in Line~\ref{lin:sc-del}:}
        Let $\tup{a,b,M}$ be the value returned from the $\LL(A[i])$ that precedes the successful $\SC$ operation.
        By the condition in Line~\ref{lin:del-helpop}, $M = \Stable$.
        The successful $\SC$ changes the value of $A[i]$ to be $\tup{a,b,\D}$. 
        Since no $\mem$ or $\lookahead$ slot is changed, the table continues to satisfy Invariant~\ref{inv:ext-order-invar}.
         By \invenumref{inv:stable-unstable}{2}, just after the successful $\SC$, $A[i].\mem = A[i+1].\lookahead$, and since $A[i]$ is the first cell to become unstable with this delete operation and the $\mem$ slot does not change, the table also
         satisfies Invariant~\ref{inv:stable-unstable}.

        \item \textbf{Successful $\SC$ in $\helpop$:}
         Let $\tup{a,b,M_1}$ and $\tup{c,d,M_2}$ be the values returned from the $\LL(A[i])$ and $\LL(A[i+1])$, respectively, in Line~\ref{lin:llAi}, ~\ref{lin:llAi+1}, or~\ref{lin:llAi+12},
         which precede the successful $\SC$ operation.
         If a successful $\SC(A[i],\ast)$ or $\SC(A[i+1],\ast)$ happens, by Lemma~\ref{lem:llllvl}, during the $\LL(A[i])$ that precedes the successful $\SC$, $A[i] = \tup{a,b,M_1}$ and $A[i+1] = \tup{c,d,M_2}$.
         Let $\tup{x,y,M_3}$ be the value returned from the $\LL(A[i-1])$ in Line~\ref{lin:llAi-1I} or~\ref{lin:llAi-1D} before the successful $\SC$.
         If a successful $\SC(A[i-1],\ast)$ happens, by Lemma~\ref{lem:llllvl}, during the $\LL(A[i])$ that precedes the successful $\SC$, $A[i-1] = \tup{x,y,M_3}$ and $A[i] = \tup{a,b,M_1}$. 
         \begin{itemize}

         \item \textbf{Successful $\SC(A[i-1],\ast)$ in Line~\ref{lin:sc-I1}:}
         The value of $A[i-1]$ is replaced from $\tup{x,y,\I}$ to $\tup{x,y,\Stable}$.
         Since $M_1 = \I$ and $y=a$, 
         by \invenumref{inv:stable-unstable}{4}, $A[i]$ is unstable with the same $\insr$ operation as $A[i-1]$.
         By Lemma~\ref{lem:propstructure}, during the successful $\SC$, $A[i] = \tup{y, b, \I}$ and Invariant~\ref{inv:stable-unstable} holds.
         Invariant~\ref{inv:ext-order-invar} trivially holds as no $\mem$ or $\lookahead$ slot is changed.

         \item \textbf{Successful $\SC(A[i],\ast)$ in Line~\ref{lin:sc-I2}:}
         The value of $A[i]$ is replaced from $\tup{a,b,\I}$ to $\tup{a,b,\Stable}$.
        By \invenumref{inv:stable-unstable}{4}, as $b=c$, the operation in $A[i]$ propagated to $A[i+1]$, and just after  the successful $\SC$, $A[i+1] = \tup{b, \ast, \ast}$, which implies Invariant~\ref{inv:stable-unstable}.
         Invariant~\ref{inv:ext-order-invar} trivially holds as no $\mem$ or $\lookahead$ slot is changed.
        
        \item \textbf{Successful $\SC(A[i+1],\ast)$ in $\DSC$ called in Line~\ref{lin:dsc-I1} or~\ref{lin:dsc-I2}:}
        By the code, $M_1 = \I$, $M_2 = \Stable$ and $b \prio{>}{i+1} c$.
        The successful $\SC$ changes the value of $A[i+1]$ to $\tup{b,d,\Stable}$ if $c = \bot$, and to $\tup{b,c,\I}$ otherwise. 
        Since by Lemma~\ref{lem:proponce} this is the only propagation of the operation to cell $A[i+1]$,
        by Lemma~\ref{lem:propstructure}, during the successful $\SC$, $A[i] = \tup{a, b, \I}$ and 
        \invref{inv:stable-unstable} holds for this cell.
        Since $c$ is replaced with a higher priority value, all elements where cell $i+1$ is between their initial hash location and their actual location in $A$ continue to satisfy the ordering invariant.
        By \invenumref{inv:stable-unstable}{1} and \invenumref{inv:stable-unstable}{2.5}, $a\neq b$. By \invenumref{inv:ext-order-invar}{3}, either $h(b) = i+1$ or $a \prio{>}{i} b$, and by Lemma~\ref{lem:monotone-priority}, the table satisfies \invenumref{inv:ext-order-invar}{1}

        If $c = \bot$, $A[i+1].\lookahead$ does not change.
        Additionally, since either $d = \bot$, or $\bot \prio{<}{i+1} d$ and $h(d) = i + 2$, both Invariant~\ref{inv:ext-order-invar} and Invariant~\ref{inv:stable-unstable} hold.
        Otherwise, by Invariant~\ref{inv:stable-unstable}, $A[i+2].\mem = d$ and $c\neq d$.
        By Lemma~\ref{lem:propstructure}, $A[i]$ is still unstable with the operation, and the operation has not yet propagated to the next cell, which implies Invariant~\ref{inv:stable-unstable}.
        Either $h(d) = i+2$, or $c \prio{>}{i+1} d$.
        If $c \prio{>}{i+1} d$ and $h(c) \neq i + 2$, by Lemma~\ref{lem:consistentprio}, $c \prio{>}{i+2} d$.
        Otherwise, if $h(d) = i+2$, either $h(c) = i+2$ or $\rank(c,i+2) > \rank(d, i+2)$ and $c \prio{>}{i+2} d$, and Invariant~\ref{inv:ext-order-invar} holds.

        \item \textbf{Successful $\SC(A[i],\ast)$ in $\DSC$ called in Line~\ref{lin:dsc-I1} or~\ref{lin:dsc-I2}:}
        The successful $\SC$ changes the value of $A[i]$ from $\tup{a,b,\I}$ to $\tup{a,b,\Stable}$.
        The conditions in Line~\ref{lin:dsc-sc} and Line~\ref{lin:dsc-checksucc} ensure that the operation propagated to cell $A[i+1]$, and by \invenumref{inv:stable-unstable}{4}, just after the successful $\SC$, $A[i+1] = \tup{b,\ast,\ast}$, 
        Since the $\mem$ slot does not change, this
        implies Invariant~\ref{inv:stable-unstable}.
        Invariant~\ref{inv:ext-order-invar} trivially holds as no $\mem$ or $\lookahead$ slot is changed.
        \end{itemize}

        For the next scenarios, which deal with propagating a $\del$ operation, we use the claim below to prove that \invref{inv:ext-order-invar} continues to hold when an element is shifted backwards.

        \begin{claim}
        \label{clm:shiftback}
            If $A[i] = \tup{a,b,\D}$ and $A[i+1] = \tup{c,\ast,\ast}$ such that $b\neq c$, then either $a \prio{>}{i} c$ or $h(c) = i+1$ and $a\neq c$.
        \end{claim}

        \begin{proof}
             By \invenumref{inv:stable-unstable}{3}, $b \prio{>}{i+1} c$.
             If $h(c) = i+1$, 
             we show that $h(a) = i+1$ implies that $a\neq c$, and this shows that $a\neq c$. If $h(a) = i+1$, by \invenumref{inv:ext-order-invar}{1}, $c \prio{>}{i+1} b$, in contradiction.
             Next, assume that $h(c) \neq i+1$.
             By Lemma~\ref{lem:consistentpriosmaller}, $b \prio{>}{i} c$.
             By \invenumref{inv:ext-order-invar}{3}, either $a \prio{\geq}{i} b$ or $h(b) = i+1$.
             By \invenumref{inv:ext-order-invar}{3}, either $a \prio{\geq}{i} b$ or $h(b) = i+1$.
             If $a \prio{\geq}{i} b$, since $b \prio{>}{i} c$, by transitivity of $\prio{>}{i}$, we also have $a \prio{>}{i} c$.
             Otherwise, if $h(b) = i+1$, as $b \prio{>}{i+1} c$, it must also hold that $h(c) = i+1$.
        \end{proof}

        \begin{itemize}
         \item \textbf{Successful $\SC(A[i-1],\ast)$ in Line~\ref{lin:sc-D1}:}
         The successful $\SC$ changes the value of $A[i-1]$ to $\tup{x,a,\Stable}$.
         Since $M_1 = \D$ and $y\neq a$, 
         by \invenumref{inv:stable-unstable}{3}, the operation propagated to $A[i]$.
         If $a\neq \bot$, $A[i]$
         is unstable with the same delete operation as $A[i-1]$, and by Lemma~\ref{lem:propstructure}, just after the successful $\SC$,
         $A[i] = \tup{a,b,\D}$.
         Otherwise, if $a = \bot$, by Lemma~\ref{lem:proponce} $A[i]$ cannot become unstable with the operation. By \invenumref{inv:stable-unstable}{3}, just after the successful $\SC$, $A[i] = \tup{\bot,\ast,\ast}$.
         Either way, 
         as the $mem$ slot does not change, Invariant~\ref{inv:stable-unstable} holds.
         Claim~\ref{clm:shiftback} proves that Invariant~\ref{inv:ext-order-invar} holds.

         \item \textbf{Successful $\SC(A[i],\ast)$ in Line~\ref{lin:sc-D2}:}
         The successful $\SC$ changes the value of $A[i]$ from $\tup{a,b,\D}$ to $\tup{a,c,\Stable}$, where either $c = \bot$ or $M_2 = \D$. 
         If $c = \bot$, by \invenumref{inv:stable-unstable}{2.5}, $b\neq \bot$.
         By \invenumref{inv:stable-unstable}{3} and Lemma~\ref{lem:proponce} $A[i+1]$ cannot become unstable with the operation. By \invenumref{inv:stable-unstable}{3}, just after the successful $\SC$, $A[i+1] = \tup{\bot,\ast,\ast}$.
         Otherwise, if $c \neq \bot$, then $M_2 = \D$ and the condition in Line~\ref{lin:helpop-while} ensures that $b \neq c$.
         By \invenumref{inv:stable-unstable}{3}, $A[i+1]$
         is unstable with the same delete operation as $A[i]$, and by Lemma~\ref{lem:propstructure}, just after the successful $\SC$,
         $A[i+1] = \tup{c,d,\D}$.
         Either way, 
         as the $mem$ slot does not change, Invariant~\ref{inv:stable-unstable} holds.
         Claim~\ref{clm:shiftback} proves that Invariant~\ref{inv:ext-order-invar} holds

        \item \textbf{Successful $\SC(A[i+1],\ast)$ in $\DSC$ called in Line~\ref{lin:dsc-D1} or~\ref{lin:dsc-D2}:}
        The successful $\SC$ changes the value of $A[i+1]$ to $\tup{d,d,\D}$ if $d \neq \bot$ and $h(d)\neq i+2$, and to $\tup{\bot,d,\Stable}$ otherwise.
        By Lemma~\ref{lem:propstructure},
        during the successful $\SC$, $A[i] = \tup{a,b,\D}$.
        By the code, $M_2  = \Stable$ and $c\neq \bot$, and by \invenumref{inv:stable-unstable}{3}, the operation did not propagate to cell $A[i+1]$ yet. Thus, by Invariant~\ref{inv:stable-unstable},
        $b=c$ and $c\neq d$.
        In addition, just after the successful $\SC$, $A[i+2] = \tup{d,\ast,\ast}$.
        By \invenumref{inv:ext-order-invar}{3}, $b = c \prio{>}{i+1} d$ or $h(d) = i+2$.
        This implies that Invariant~\ref{inv:stable-unstable},
        \invenumref{inv:ext-order-invar}{2} and \invenumref{inv:ext-order-invar}{3} hold for both cases.
        For $d\neq \bot$ and $h(d) \neq i+2$, $c \prio{>}{i+1} d$, and by Lemma~\ref{lem:monotone-priority} and \invenumref{inv:ext-order-invar}{1}, just before the successful $\SC$, for all $h(d) \leq j\leq i+1$, $A[j].\mem \prio{>}{j} d$.
        Hence, \invenumref{inv:ext-order-invar}{1}
        continues to hold when $d$ is shifted to the $\mem$ slot of $A[i+1]$.
        For $d = \bot$ or $h(d) = i+2$, \invenumref{inv:ext-order-invar}{1} continues to hold trivially.

         \item \textbf{Successful $\SC(A[i],\ast)$ in $\DSC$ called in Line~\ref{lin:dsc-D1} or \ref{lin:dsc-D2}:}
        The successful $\SC$ changes the value of $A[i]$ from $\tup{a,b,\D}$ to $\tup{a,d,\Stable}$ if $d \neq \bot$ and $h(d)\neq i+2$, and to $\tup{a,\bot,\Stable}$ otherwise.
        The conditions in Line~\ref{lin:dsc-sc} and Line~\ref{lin:dsc-checksucc} ensure that the operation propagated to cell $A[i+1]$.
        If $d \neq \bot$ and $h(d)\neq i+2$, by Lemma~\ref{lem:propstructure}, it must be that $A[i+1] = \tup{d,d,\D}$ before $A[i]$ stabilizes. Otherwise, by~\invenumref{inv:stable-unstable}{3} and Lemma~\ref{lem:proponce}, it must be that $A[i+1] = \tup{\bot,\ast,\ast}$ before $A[i]$ stabilizes. 
        Either way, 
        as the $mem$ slot does not change, Invariant~\ref{inv:stable-unstable} holds.
         Claim~\ref{clm:shiftback} proves that Invariant~\ref{inv:ext-order-invar} holds. \qedhere
    \end{itemize}
    \end{itemize}
\end{proof}

For the rest of the section, we assume that $A$ always satisfies Invariant~\ref{inv:ext-order-invar} and Invariant~\ref{inv:stable-unstable}, omitting the reference to Lemma~\ref{lem:validorder} for brevity.

Next, we show that relying on Assumption~\ref{asm:neverfull}, no operation aborts.

\begin{lemma}
    If Assumption~\ref{asm:neverfull} holds, no operation aborts.
\end{lemma}

\begin{proof}
    An operation aborts in Line~\ref{lin:nospace} or Line~\ref{lin:stuckinsr}.
    If an $\insr(v)$ operation aborts in Line~\ref{lin:nospace}, by the code, $A[h(v)-1] = \tup{a,b,\ast}$, such that $v \prio{<}{h(v)-1} a$ (by Line~\ref{lin:insert-until}) and $v \prio{<}{h(v)} b$ (by Line~\ref{lin:insert-loc}). By Lemma~\ref{lem:monotone-priority} and \invenumref{inv:ext-order-invar}{1}, for every $0\leq i < \ncell$, $A[i].\mem \prio{>}{i} v$.
    Thus, all $\mem$ slots must be nonempty, that is, not equal to $\bot$.

    If an operation aborts in Line~\ref{lin:stuckinsr}, then
    by Lemma~\ref{lem:llllvl}, there is a point before the abort where for some $0 \leq i < m$,  $A[i] = \tup{a,b,I}$ and $A[i+1] = \tup{c,\ast,\ast}$ such that $c \prio{>}{i+1} b$.
    By Invariant~\ref{inv:ext-order-invar} 
    and \invref{inv:stable-unstable},
    this implies that $h(b) = i+1$ and $a \prio{>}{i} b$.
    Thus, by Lemma~\ref{lem:monotone-priority} and \invenumref{inv:ext-order-invar}{1}, for every $0\leq i < \ncell$, $A[i].\mem \prio{>}{i} b$.
    Similarly to the previous case, this implies that Assumption~\ref{asm:neverfull} does not hold.
\end{proof}

\subsection{Linearizability}
\label{sec:implproof-linearizability}

We order the $\insr$ and $\del$ operations that actually insert or delete elements from the table and change the state of the dictionary according to the order of their initial writes.
In most cases, we can show that a $\lookup(v)$ operation returns \textit{true} if and only if $v$ is present in the table $A$, either in a \emph{lookahead} or \emph{value} slot, during its last read of $A$, and $v$ is present in the table $A$
if and only if $v$ is in the set by the order of operations defined above. (This also applies for $\insr$ and $\del$ operations that return \textit{false}.)
The only exception is when the initial write by an $\insr(v)$ or $\del(v)$ operation has occurred, but the first propagation of the operation has not yet completed.
An operation on element $v$ which starts before the initial write may be unaware of such insertion or deletion.
As a result, it can conclude that $v$ is not in the set (resp., in the set) when it is present (resp., not present) in the table, before the first propagation completes.
However, since this operation must be concurrent with the uncompleted $\insr$ or $\del$ operation, we can simply place this operation before the uncompleted operation in the order, ensuring that the operation returns the correct response according to the order of operations.

The next two lemmas show that the \emph{lookahead} slots also satisfy the ordering invariant of Robin Hood hashing, similar to the \emph{value} slots.

\begin{lemma}
\label{lem:biggerprio}
    Let $v\in \CalU$, if for some $0 \leq i < m$, $A[i].\mem \prio{>}{i} v$, then for all $h(v) \leq j < i$, $A[j].\mem \prio{>}{j} v$, and $A[j].\lookahead \prio{>}{j+1} v$ or $h(A[j].\lookahead) = j + 1$.
\end{lemma}

\begin{proof}
    As \invenumref{inv:ext-order-invar}{1} holds, by Lemma~\ref{lem:monotone-priority}, for all $h(v) \leq j \leq i$, $A[j].\mem \prio{>}{j} v$.
    Let $h(v) \leq j < i$, then by \invenumref{inv:ext-order-invar}{2},
    either $A[j].\lookahead 
    \prio{\geq}{j+1} A[j+1].\mem$ or $h(A[j].\lookahead) = j+1$.
    If $A[j].\lookahead 
    \prio{\geq}{j+1} A[j+1].\mem$, since $A[j+1].\mem 
    \prio{>}{j+1}v$ and as $\prio{>}{i+1}$ is transitive, we get that
    $A[j].\lookahead \prio{>}{j+1} v$.
\end{proof}

\begin{lemma}
\label{lem:smallerprio}
    Let $v\in \CalU$, if for some $0 \leq i < m$, $A[i].\mem \prio{<}{i} v$, then for all $i+1 \leq j \leq h(v) - 1$, $A[j].\mem \prio{<}{j} v$ and all $i \leq j \leq h(v)-1$, $A[j].\lookahead \prio{<}{j} v$ or $h(A[j].\lookahead) = j + 1$.
\end{lemma}

\begin{proof}
    By \invenumref{inv:ext-order-invar}{1}, 
    for all $i+1 \leq j \leq h(v) - 1$, $A[j].\mem \neq v$. Otherwise, there is a lower priority element between $v$'s location and $h(v)$. This together with Lemma~\ref{lem:monotone-priority} implies that for all $i+1 \leq j \leq h(v) - 1$, $A[j].\mem \prio{<}{j} v$.
    Let $i \leq j \leq h(v)-1$.
    By~\invenumref{inv:ext-order-invar}{3},
    either $A[j].\lookahead \prio{\leq}{j} A[j].\mem \prio{<}{j} v$ or $h(A[j].\lookahead) = j+1$.
\end{proof}

We say the \emph{$v$ is physically in $A$} if there is a cell $0\leq i < \ncell$ such that $A[i].\mem = v$ or $A[i].\lookahead = v$.
Even if $v$ is physically present in $A$, an operation might conclude that $v$ is not in the set. For this reason, we also have the following definition;
We say that \emph{$v$ is logically in $A$}, if there is a cell $0\leq i < \ncell$ such that $A[i].\mem = v$ or $A[i].\lookahead = v$ and $A[i].\mem \prio{\geq}{i} v$.
By Invariant~\ref{inv:ext-order-invar}, if $A[i].\mem \prio{<}{i} A[i].\lookahead$, then $h(A[i].\lookahead) = i + 1$.
Hence, if $v$ is physically in $A$ but not logically in $A$, $A[h(v)-1].\lookahead = v$, and no other cells contain the value $v$.

Next, we prove that each of the conditions for determining that $v$ is not in the table ensures this.

\begin{lemma}
    \label{lem:elemnotpresentcond}
    Let $v\in \CalU$. If one of the following statements holds, then $v$ is not logically in $A$:
    \begin{enumerate}
        \item $A[h(v)].\mem \prio{<}{h(v)} v$.
        \item $A[h(v)-1].\mem \prio{>}{h(v)-1} v$.
    \end{enumerate}
    Furthermore, if one of the following statements holds, then $v$ is not physically in $A$: 
    \begin{enumerate}
        \item $A[h(v)].\mem \prio{<}{h(v)} v$ and $A[h(v)-1].\lookahead \neq v$.

        \item $A[h(v)-1].\mem \prio{>}{h(v)-1} v$ and $A[h(v)-1].\lookahead \neq v$.

        \item There is cell $i$, $0\leq i < \ncell$, such that $A[i].\mem \prio{>}{i} v \prio{>}{i+1} A[i].\lookahead$ and either $A[i].\flag = \Stable$ or $h(A[i].\lookahead) \neq i+1$.

        \item There is cell $i$, $0\leq i < \ncell$ and $i\neq h(v)-1$, such that $A[i].\lookahead \prio{>}{i} v \prio{>}{i+1} A[i+1].\mem$ and $h(A[i].\lookahead) \neq i+1$.
    \end{enumerate}
\end{lemma}

\begin{proof}
    We consider the different cases:
    \begin{enumerate}
        \item If $A[h(v)].\mem \prio{<}{h(v)} v$, by Lemma~\ref{lem:smallerprio}, for all $h(v) \leq j \leq h(v) - 1$, $A[j].\mem \prio{<}{j} v$ and all $h(v) \leq j \leq h(v)-1$, $A[j].\lookahead \prio{<}{j} v$ or $h(A[j].\lookahead) = j + 1$. 
        If $A[h(v)-1].\lookahead = v$, $v$ is not logically in $A$, and if $A[h(v)-1].\lookahead \neq v$, it is not even physically in $A$.

        \item Assume $A[h(v)-1].\mem \prio{>}{h(v)-1} v$ and $A[h(v)-1].\lookahead \neq v$.
        By Lemma~\ref{lem:biggerprio}, for all $h(v) \leq j < h(v)-1$, $A[j].\mem \prio{>}{j} v$, and $A[j].\lookahead \prio{>}{j+1} v$ or  $h(A[j].\lookahead) = j+1 \neq h(v)$.
        If $A[h(v)-1].\lookahead = v$, $v$ is not logically in $A$, and if $A[h(v)-1].\lookahead \neq v$, it is not even physically in $A$.
    \end{enumerate}

    For the next two cases, we use the following claim:
    \begin{claim}
    \label{clm:biggersmallerprio}
        If there is $0\leq i < \ncell$ such that $A[i].\mem \prio{>}{i} v \prio{>}{i+1} A[i+1].\mem$ and $A[i].\lookahead \neq v$, then $v$ is not physically in $A$.
    \end{claim}

    \begin{proof}
    By Lemma~\ref{lem:biggerprio}, for all $h(v) \leq j < i$, $A[j].\mem \prio{>}{j} v$, and $A[j].\lookahead \prio{>}{j+1} v$ or $h(A[j].\lookahead) = j + 1 \neq h(v)$.
    By Lemma~\ref{lem:smallerprio},
    for all $i+2 \leq j \leq h(v) - 1$, $A[j].\mem \prio{<}{j} v$ and all $i+1 \leq j \leq h(v)-1$, $A[j].\lookahead \prio{<}{j} v$ or $h(A[j].\lookahead) = j + 1$.
    Let $i+1 \leq j \leq h(v)-1$ and assume that $h(A[j].\lookahead) = j + 1$. If $j\neq h(v)-1$, then $j+1 \neq h(v)$ and $A[j].\lookahead \neq v$.
    Otherwise, by \invenumref{inv:ext-order-invar}{3}, either $A[h(v)-1].\lookahead \prio{\leq}{h(v)-1} A[h(v)-1].\mem \prio{<}{h(v)-1} v$ or $A[h(v)-1].\lookahead \prio{\geq}{h(v)} A[h(v)].\mem \prio{>}{h(v)} v$.
    \end{proof}

    \begin{enumerate}
    \setcounter{enumi}{2}
        \item Assume that $A[i].\mem \prio{>}{i} v \prio{>}{i+1} A[i].\lookahead$ and $A[i].\flag = \Stable$ or $h(A[i].\lookahead) \neq i+1$.
        If $i = h(v) - 1$, by the second case of the lemma, $v$ is not physically in $A$.
        By \invenumref{inv:stable-unstable}{2}, if $A[i]$ is stable, $A[i+1].\mem = A[i].\lookahead \prio{<}{i+1} v$.
        By \invenumref{inv:ext-order-invar}{3},
        if $h(A[i].\lookahead) \neq i+1$, then $A[i+1].\mem \prio{\leq}{i+1} A[i].\lookahead \prio{<}{i+1} v$.
        Either way, as $\prio{<}{i+1}$ is transitive, 
        $A[i+1].\mem \prio{<}{i+1} v$.
        By Claim~\ref{clm:biggersmallerprio}, $v$ is not physically in the table.

    \item Assume that $A[i].\lookahead \prio{>}{i} v \prio{>}{i+1} A[i+1].\mem$ and $h(A[i].\lookahead) \neq i + 1$.
    Since $h(A[i].\lookahead) \neq i + 1$, by \invenumref{inv:ext-order-invar}{3}, we have $A[i].\mem \prio{\geq}{i} A[i].\lookahead$. By the transitivity of $\prio{>}{i}$, $A[i].\mem \prio{>}{i} v$.
    By Claim~\ref{clm:biggersmallerprio}, $v$ is not physically in the table.
    \qedhere
    \end{enumerate}
\end{proof}

In the next lemma, we verify that if a $\del(v)$ or $\lookup(v)$ operation concludes that $v$ is not in the table, then $v$ is not logically in the table.

\begin{lemma}
\label{lem:vnotinA}
    Consider a $\del(v)$ or $\lookup(v)$ operation that returns \textit{false}, then during the last read of $A$ in the operation, $v$ is not logically in $A$.
\end{lemma}

\begin{proof}
     We show that during the last read read of $A$ in the operation one of the conditions in Lemma~\ref{lem:elemnotpresentcond} holds, which proves the lemma statement.
     Let $\tup{a,b,M}$ be the last value read from $A$ in the $\del(v)$ or $\lookup(v)$ operation. Let $i$ be the cell from which the value is read.

    Consider a $\del(v)$ or $\lookup(v)$ operation that returns in Line~\ref{lin:del-elementnotpresent2} or~\ref{lin:lookup-elementnotpresent2} respectively.
     This happens if the loop goes through all the cells in the table and returns to the starting point, which is back to $h(v)$, so, $i = h(v) - 1$.
     By the condition in Line~\ref{lin:until-del} or Line~\ref{lin:until-lookup}, $a \prio{\geq}{i} v$, otherwise the loop restarts.
     In addition, by Line~\ref{lin:del-vinval} or Line~\ref{lin:lookup-elementfound}, $a\neq v$.
     So, we must have $a \prio{>}{i} v$, and by Lemma~\ref{lem:elemnotpresentcond}, $v$ is not logically in the table.

     Consider a $\del(v)$ or $\lookup(v)$ operation that returns in Line~\ref{lin:del-elementnotpresent1} or~\ref{lin:lookup-elementnotpresent1}, respectively.
     Either $i = h(v)$ and $v \prio{>}{i} a$, or $a \prio{>}{i} v \prio{>}{i+1} b$ and $M = \Stable$ or $h(b) \neq i+1$.
     By Lemma~\ref{lem:elemnotpresentcond}, in the first case $v$ is not logically in the table, and in the second case $v$ is not physically in the table.

     Consider a $\lookup(v)$ operation that returns in Line~\ref{lin:lookup-elementnotpresent3}.
     Then, as the $\VL(A[i])$ in Line~\ref{lin:lookup-elementnotpresent3} returns \textit{true}, during the $\LL(A[i+1])$ in Line~\ref{lin:lookupllAi+1}, $A[i] = \tup{a,b,\I}$ and $A[i+1] = \tup{c,\ast,\ast}$ such that $b \prio{>}{i+1} v \prio{>}{i+1} c$ and $h(b) = i+1$. By Lemma~\ref{lem:elemnotpresentcond}, $v$ is not physically in the table.
\end{proof}

A simple inspection of the code reveals that the only place where a new element is inserted into the table is by an $\SC$ in Line~\ref{lin:sc-insert}.
The next lemma verifies that the element is not physically present in the table just prior to such $\SC$.

\begin{lemma}
\label{lem:elemnotpresent}
    Consider an $\insr(v)$ operation that
    performs a successful $\SC$ in Line~\ref{lin:sc-insert}, then
    right before the successful $\SC$, $v$ is not physically in $A$.
\end{lemma}

\begin{proof}
     We show that right before the successful $\SC$ in Line~\ref{lin:sc-insert}, one of the conditions in Lemma~\ref{lem:elemnotpresentcond} that imply $v$ is not physically in the table holds, which proves the lemma statement.
     Let $\tup{a,b,M}$ be the last value read from $A$ in Line~\ref{lin:insr-readcurr1} or \ref{lin:insert-readnextcell} in the operation.
     Let $i$ be the cell from which the value is read.
     
     By the code, the successful $\SC$ is performed on cell $A[i]$.
     By Line~\ref{lin:insert-loc}, either $b = \bot$ or $v \prio{>}{i+1} b$.
     If $i = h(v)-i$, since by the code $M = \Stable$, $A[h(v)] = b$.
     If $i \neq h(v)-1$, this is not the first iteration in the loop, and by the conditions in Line~\ref{lin:insert-until} and \ref{lin:insr-element-present}, $a \prio{>}{i} v$.
\end{proof}

The only place where an element is deleted from the table is by an $\SC$ in a $\DSC$ in Line~\ref{lin:dsc-D1} or \ref{lin:dsc-D2}, during the first propagation of a $\del$ operation.
The next lemma verifies that right after the first propagation ends the element is deleted entirely from the table.

\begin{lemma}
\label{lem:nondangdel}
    Consider a $\del(v)$ operation that performs a successful $\SC(A[i],\ast)$ operation in Line~\ref{lin:sc-del}.
    Then right after a successful $\SC$ in Line~\ref{lin:dsc-D1} or \ref{lin:dsc-D2} that stabilizes $A[i]$ with the delete operation, $v$ is not physically in $A$.
\end{lemma}

\begin{proof}
    By the code, when $A[i]$ is unstable with the delete operation, $A[i] = \tup{a,v,\D}$, $a\neq v$.
    By Lemma~\ref{lem:propstructure}, $A[i]$ stabilizes after the operation propagates to $A[i+1]$.
    Just before the successful $\SC$ that
    propagates the operation to $A[i+1]$, cell $A[i+1]$ is stable. Thus, just before this successful $\SC$, $A[i+1] = \tup{v,b,\Stable}$, such that $v\neq b$.
    By the code, \invenumref{inv:stable-unstable}{3} and Lemma~\ref{lem:proponce},
    after the propagation, $A[i+1] = \tup{b,b,\D}$ or $A[i+1] = \tup{\bot, b, \D}$.

    We show that $v$ is not present in cell $j$, $j\neq i$, while $A[i]$ is unstable and after the operation propagated to cell $A[i+1]$.
    By Invariant~\ref{inv:ext-order-invar}, either $a \prio{>}{i} v$, or $h(v) = i+1$ and $v \prio{>}{i+1} b$.
    If $h(v) = i+1$, we get that $A[h(v)-1].\mem \prio{>}{h(v)-1} v$ or $A[h(v)].\mem \prio{<}{h(v)} v$, and by Lemma~\ref{lem:elemnotpresentcond}, $v$ is only present in the $A[i].\lookahead$ slot.
    Otherwise, if $h(v) \neq i+1$, $a \prio{>}{i} v$.
    By Lemma~\ref{lem:biggerprio}, for all $h(v) \leq j \leq i$, $A[j].\mem \prio{>}{j} v$ and $A[j].\lookahead \prio{>}{j+1} v$.
    Since $b\neq v$, either $v \prio{>}{i+1} b$ or $v \prio{<}{i+1} b$.
    However, it cannot be the case that $v \prio{<}{i+1} b$, since by Lemma~\ref{lem:biggerprio}, this implies that $A[i].\lookahead \neq v$.
    Hence, $v \prio{>}{i+1} b$, and by Lemma~\ref{lem:smallerprio}, for all $i+2 \leq j \leq h(v)-1$, $A[j].\mem \prio{<}{j} v$ and for all $i+1 \leq j \leq h(v)-1$, $A[j].\lookahead \prio{<}{j} v$ or $h(A[j].\lookahead) = j+1$.
    For $j = h(v)-1$, if $h(A[j].\lookahead) = h(v)$, then $A[j].\lookahead \prio{\geq}{j+1} A[j+1].\mem \prio{>}{j+1} v$.

    When $A[i]$ stabilizes, $A[i] = \tup{a,b,\Stable}$. Since $a,b \neq v$ and $v$ is not present in any other cell in the table, $v$ is completely deleted from $A$.
\end{proof}

We say that $\insr$ and $\del$ operations that perform a successful $\SC$ in Line~\ref{lin:sc-insert} and Line~\ref{lin:sc-del}, respectively, are 
\emph{set-changing operations}.
We say that a set-changing operation is \emph{dangling} if it performs a successful $\SC(A[i],\ast)$ in Line~\ref{lin:sc-insert} or Line~\ref{lin:sc-del}, and $A[i]$ has not stabilized yet.
Note that by the code, these operations eventually return \textit{true}.

\begin{lemma}
\label{lem:onedang}
    There is only one dangling $\insr(v)$ or $\del(v)$ at any point in the execution.
\end{lemma}

\begin{proof}
    Assume there are two dangling operations $o_1$ and $o_2$ at the same point in the execution, such that operation $o_1$ performs a successful $\SC(A[i],\ast)$ in Line~\ref{lin:sc-insert} or Line~\ref{lin:sc-del} before operation $o_2$ performs a successful $\SC(A[j],\ast)$ in Line~\ref{lin:sc-insert} or Line~\ref{lin:sc-del}.
    By the code, after the successful $\SC$ in $o_1$, $A[i] = \tup{a,v,\I}$ or $A[i] = \tup{a,v,\D}$, $a\neq v$.
    Since $o_2$ in dangling at the same time as $o_1$, the value of $A[i]$ stays the same when the successful $\SC$ in $o_2$ happens.
    As only a stable cell can become unstable, this implies that $i\neq j$.
    If $o_2$ is an $\insr(v)$ operation, by Lemma~\ref{lem:elemnotpresent}, just before the successful $\SC$ in $o_2$, $v$ is not physically in the table. However, this contradicts that $v$ is in $A[i]$ at that point.
    
    If $o_2$ is a $\del(v)$ operation, just before the successful $\SC$ in $o_2$, $A[j] = \tup{c,v,\Stable}$, such that $c\neq v$.
    By~\invenumref{inv:stable-unstable}{2}, $A[j+1] = \tup{v,\ast,\ast}$.
    If $h(v)-1\leq j < i$, by~ \invenumref{inv:ext-order-invar}{3}, as $h(v) \neq i+1$, $a \prio{>}{i} v$. However, this contradicts Lemma~\ref{lem:monotone-priority}.
    If $h(v)-1\leq i < j < h(v)-1$,
    then $h(v) \neq j + 1$. By~\invenumref{inv:ext-order-invar}{3}, $c \prio{>}{j} v$. By Lemma~\ref{lem:biggerprio}, this implies that $A[i+1].\mem  \prio{>}{i+1} v$.
    We show that while $A[i]$ is unstable, $A[i+1] = \tup{b,\ast,\ast}$, such that $v \prio{\geq}{i+1} b$, leading to a contradiction.
    If $o_2$ is a $\del$ operation, this directly follows from \invenumref{inv:ext-order-invar}{3}.
    If $o_1$ is an $\insr$ operation, just before the successful $\SC$ that makes $A[i]$ unstable, by the code, $A[i] = \tup{a,b,\Stable}$ such that $v \prio{>}{i+1} b$. By \invenumref{inv:stable-unstable}{2}, at the same point, $A[i+1] = \tup{b,\ast,\ast}$.
    By Lemma~\ref{lem:propstructure}, no other operation other than $o_1$ can propagate to cell $A[i+1]$. As only successful $\SC$s that propagate operations change $\mem$ slots, and by \invenumref{inv:stable-unstable}{4}, the claim follows. 
\end{proof}

\begin{lemma}
\label{lem:nodangnoambig}
    If there is no dangling operation, then every element $v$ is either logically in $A$ or not physically in $A$.
\end{lemma}

\begin{proof}
    By the code, after a successful $\SC(A[i],\ast)$ in Line~\ref{lin:sc-insert} in an $\insr(v)$ operation, $A[i] = \tup{\ast,v,\I}$.
    By Lemma~\ref{lem:propstructure} and \invenumref{inv:stable-unstable}{4}, just before $A[i]$ stabilizes, $A[i+1] = \tup{v,\ast,\ast}$.
    This along with Lemma~\ref{lem:nondangdel} implies that just after an operation is no longer dangling, $v$ is either physically in $A$ or not physically in $A$.
    Element $v$ can be ``pushed'' from a $\mem$ slot only during a propagation of some insert operation; if $A[i] = \tup{v,d,\Stable}$ it becomes $A[i] = \tup{b,v,\Stable}$, where by the code, $b \prio{>}{i} v$, and $v$ stays logically in $A$.
\end{proof}

\begin{lemma}
\label{lem:dangpending}
    If an operation is dangling, then it is still pending.
\end{lemma}

\begin{proof}
    After a process performs a successful $\SC$ in Line~\ref{lin:sc-insert} or \ref{lin:sc-del} in a delete or insert operation, it calls $\prop$.
    By the code, the process does not return until $A[i]$ stabilizes.
\end{proof}

In the next lemma, we verify that if there is a dangling $\del(v)$ operation, an $\insr(v)$ or $\lookup(v)$ operation concludes that $v$ is in the table only if $v$ is logically in the table.
If such an operation returns that $v$ is in the table when $v$ is only physically in the table, then it ignores the concurrent deletion, meaning the deletion does not take effect yet.
However, a $\del(v)$ operation that precedes the $\insr(v)$ or $\lookup(v)$ operation may already determine that $v$ is not in the table, thereby causing the deletion to take effect.

\begin{lemma}
\label{lem:vinA}
    Consider an $\insr(v)$ operation that returns \textit{false} or a $\lookup(v)$ operation that returns \textit{true}, if during the last read of $A$ in the operation there is a dangling $\del(v)$ operation, then $v$ is logically in $A$.
\end{lemma}

\begin{proof}
    Since there is a dangling $\del(v)$ operation, by the code, there is a cell $A[i] = \tup{a,v,\D}$ during the last read from $A$ in the $\insr(v)$ or $\lookup(v)$ operation.
    Let $\tup{a,b,M}$ be the last value read from $A$ in the $\insr(v)$ or $\lookup(v)$ operation. Let $j$ be the cell from which the value is read.
    By Line~\ref{lin:insr-element-present} in an $\insr$ operation, and Line~\ref{lin:lookup-elementfound} in a $\lookup$ operation, $a = v$ or $b = v$ and either $M \neq \D$ or $h(v) \neq i +1 $.
    If $a = v$ or $b = v$ and $h(v) \neq i+1$, then $v$ is logically in $A$.
    If $b = v$ and $M \neq \D$, then $i \neq j$, and at least one of them is different than $h(v)-1$, which implies that $v$ is logically in $A$.
\end{proof}

Let element $v\in \CalU$, we construct the permutation $\pi_v$, that consists only of the $\insr(v)$, $\del(v)$ and $\lookup(v)$ operations in the execution, as follows;
First, we order the set-changing operations in the same order that the successful $\SC$s in Line~\ref{lin:sc-insert} and Line~\ref{lin:sc-del} happen in the execution.
Next, we go through all the
non-set changing operations in the order they are invoked in the execution, which are insert or delete operations that return \textit{false}, and lookup operations that return either \textit{true} or \textit{false}.
Consider non-set changing $\insr(v)$, $\del(v)$ or $\lookup(v)$ operation $o$;
If there is a dangling $\del(v)$ operation $o_{\del}$ during the last read of $A$ in $o$, if $o$ is an insert operation or a lookup operation that returns \textit{true},
we place $o$ immediately before operation $o_{\del}$.
If $o$ is a delete operation or a lookup operation that returns \textit{false},
we place $o$ immediately before the next set-changing operation that follows $o_{\del}$ in the ordering.
If there is a dangling $\insr(v)$ operation $o_{\insr}$ during the last read of $A$ in $o$, 
if $o$ is an insert operation or a lookup operation that returns \textit{true}, we place $o$ immediately before the next set-changing operation that follows $o_{\insr}$ in the ordering.
If $o$ is a delete operation or a lookup operation that returns \textit{false}, we place $o$ immediately before operation $o_{\insr}$.
Otherwise, if there is no dangling operation during the last read of $A$ in $o$, let $o'$ be the set-changing operation that performs the next successful $\SC$ in Line~\ref{lin:sc-insert} that follows the last read of $A$ in $o$, then we place $o$ immediately before operation $o'$. If no such operation $o'$ exists, we place $o$ at the end. 

We build the permutation $\pi$ by interleaving all permutations $\pi_v$ for every $v\in \CalU$. We do this in a way that respects the real-time order of operations with different input elements. Since the ordering of the operations in each permutation is independent, such interleaving is possible.

\begin{lemma}
    The hash table implementation is linearizable.
\end{lemma}

\begin{proof}
Let $v\in \CalU$, we prove the claim for permutation $\pi_v$. This is enough to prove the claim, as the correctness of the sequential specification is independent of operations with a different input element, and $\pi$ respects the real-time order of operations with different input elements.
 
\textbf{Correctness:}
First, we show the ordering respects the sequential specification of a set.
Consider an operation $o$ included in the permutation $\pi_v$, and consider all the different options:
\begin{itemize}
    \item \textbf{$\insr(v)$ operation $o$ that performs a successful $\SC$ in Line~\ref{lin:sc-insert}:}
    By Lemma~\ref{lem:elemnotpresent}, $v$ is not physically in the table just before this $\SC$.
    This implies that if there is an $\insr(v)$ operation $o'$ that also performs a successful $\SC$ in Line~\ref{lin:sc-insert} and precedes $o$ is the order, then there must be a $\del(v)$ operation $o''$ that performs a successful $\SC$ in Line~\ref{lin:sc-del} and is placed between $o$ and $o'$.

    \item \textbf{$\del(v)$ operation $o$ that performs a successful $\SC$ in Line~\ref{lin:sc-del}:}
    Let $o'$ be the latest $\insr(v)$ operation that performs a successful $\SC$ in Line~\ref{lin:sc-insert} and precedes $o$ in the order. 
    There is such an insert operation, as by the code, during the successful $\SC$ in Line~\ref{lin:sc-del}, $v$ is physically in the table.
    Assume there is a $\del(v)$ operation $o''$ that performs a successful $\SC$ in Line~\ref{lin:sc-del} and is placed between $o'$ and $o$. Then, $o''$ must be dangling during the successful $\SC$ in $o$, otherwise, by Lemma~\ref{lem:nondangdel} and since there is no set-changing $\insr(v)$ operation placed between $o''$ and $o$, $v$ is physically not in the table.
    However, this contradicts Lemma~\ref{lem:onedang}, which states there is only one dangling $\del(v)$ operation at a given point in the execution.

    \item 
    \textbf{$\insr(v)$ operation $o$ that returns \textit{false} or a $\lookup(v)$ operation $o$ that returns \textit{true}:}
    If there is a dangling $\insr(v)$ during the last read of $A$, then $o$ is placed after this operation and there is no delete operation that returns \textit{true} between them.
    Otherwise, let $o'$ be the latest $\insr(v)$ operation that performs a successful $\SC$ in Line~\ref{lin:sc-insert} before the last read of $A$ in $o$.
    There is such an insert operation, as by the code, during the last read of $A$ element $v$ is physically in the table.
    If there is a dangling $\del(v)$ operation $o''$ during the last read of $A$ in $o$, then $o$ is placed before this operation.
    By Lemma~\ref{lem:onedang}, 
    $o'$ precedes $o''$ in the order.
    Thus, $o'$ precedes $o$ in $\pi_v$.
    Assume there is a non-dangling $\del(v)$ operation $o''$ that performs a successful $\SC$ in Line~\ref{lin:sc-del} and is placed between $o'$ and $o$.
    By the construction of $\pi_v$, the successful $\SC$ in $o''$ happens before the last read of $A$ in $o$ and after the successful $\SC$ in $o'$.
    However, $o''$ must be dangling during the last read of $A$ is $o$, otherwise, by Lemma~\ref{lem:nondangdel} and since there is no set-changing $\insr(v)$ operation between $o''$ and $o$, $v$ is physically not in the table, in contradiction to assuming $o''$ is not dangling.
    This proves no set-changing $\del(v)$ operation is placed between $o'$ and $o$.

    \item
    \textbf{$\del(v)$ or $\lookup(v)$ operation $o$ that returns \textit{false}:}
    By Lemma~\ref{lem:vnotinA}, during the last read of $A$ in $o$, $v$ is not logically in $A$.
    Assume that $v$ is not physically in the table.
    This implies that if there is an $\insr(v)$ operation $o'$ that performs a successful $\SC$ in Line~\ref{lin:sc-insert} and precedes the last read of $A$ in $o$, then there must be a $\del(v)$ operation $o''$ that performs a successful $\SC$ in Line~\ref{lin:sc-del} 
    after the successful $\SC$ in $o'$ and before the last read of $A$ in $o$.
    By the construction of $\pi_v$, $o$ is placed after $o''$ in the order.
    If $v$ is physically in the table, by Lemma~\ref{lem:nodangnoambig},
    there must be a dangling $\insr(v)$ or $\del(v)$ operation $o'$.
    If the dangling operation is an $\insr(v)$ operation, $o$ is placed before $o'$, and there are no set-changing operations placed between them.
    If there is a set-changing $\insr(v)$ operation $o''$ that precedes $o'$ in the order it also precedes $o$, and 
    as shown in the previous case, there is a set-changing $\del(v)$ operation $o'''$ placed between $o''$ and $o$ in the order.
    If the dangling operation is $\del(v)$, $o$ is placed after $o'$, and there are no set-changing operations placed between them. 
\end{itemize}

\textbf{Real-time order:}
We need to show that if operation $o_1$ returns before operation $o_2$ invokes, then $o_1$ also precedes $o_2$ in $\pi_v$:
\begin{itemize}
    \item \textbf{Both $o_1$ and $o_2$ are set-changing operations:}
    Since set-changing operations are pending during the successful $\SC$ in Line~\ref{lin:sc-insert} or Line~\ref{lin:sc-del}, the claim follows directly.

    \item \textbf{$o_1$ is a set changing operation and $o_2$ is a non-set changing operation:}
    Operation $o_1$ performs a successful $\SC$ in Line~\ref{lin:sc-insert} or \ref{lin:sc-del} before the last read of $A$ in $o_2$.  
    By the construction, $o_2$ is placed before $o_1$ only if $o_1$ is dangling, but by Lemma~\ref{lem:dangpending}, this implies $o_1$ is pending, in contradiction to the initial assumption.

    \item \textbf{$o_1$ is a non-set changing operation and $o_2$ is a set changing operation:}
    Then the last read of $A$ in $o_1$ happens before the successful $\SC$ in Line~\ref{lin:sc-insert} or \ref{lin:sc-del} in $o_2$. 
    By Lemma~\ref{lem:dangpending}, $o_2$ is not dangling during the last read of $A$ in $o_1$,
    and according to the construction of $\pi_v$, $o_1$ is placed before $o_2$ in $\pi_v$.

    \item \textbf{Both $o_1$ and $o_2$ 
    are an $\insr(v)$ operation that returns \textit{false} or a $\lookup(v)$ operation that returns \textit{true}, or a $\del(v)$ or $\lookup(v)$ operation that returns \textit{false}:}
    If there is no dangling operation during the last read of $A$ in both $o_1$ and $o_2$, the operations are
    placed after the last set-changing that performs a successful $\SC$ before the last read of $A$
    and 
    before the next set-changing operation that performs a successful $\SC$ after the last read of $A$.
    If there is a dangling operation $o$ during the last read of $A$ in $o_1$,
    then $o_1$ is placed just before $o$ or just after $o$.
    If $o_1$ is placed just before $o$, $o_2$ is either placed before $o$, or before another set-changing operation that follows $o$ in $\pi_v$.
    If $o_1$ is placed just after $o$, $o_2$ is also placed after $o$, or before another set-changing operation that follows $o$ in the order.
    Assume there is no dangling operation during the last read of $A$ in $o_1$ and 
    a dangling operation $o$ during the last read of $A$ in $o_2$.
    By construction, since $o$ does not perform a successful $\SC$ yet during the last read of $A$ in $o_1$, $o_1$ is placed before $o$ in $\pi$.
    Operation $o_2$ is placed just before $o$, or after $o$.
    To conclude, $o_1$ is placed before an operation that does not follow the operation $o_2$ is placed before.
    If $o_1$ and $o_2$ are placed before the same operation, the construction places $o_1$ before $o_2$.

    \item \textbf{$o_1$ is
    an $\insr(v)$ operation that returns \textit{false} or a $\lookup(v)$ operation that returns \textit{true} and $o_2$ is a $\del(v)$ or $\lookup(v)$ operation that returns \textit{false}, or vice-versa:}
    Like in the previous case, if there is no dangling operation during the last read of $A$ in $o_1$, $o_1$ is placed before $o_2$.
    If there is a dangling operation $o$ during the last read of $A$ in $o_1$,
    then $o_1$ is placed just before $o$ or just after $o$.
    If $o$ is no longer dangling during the last read of $A$ in $o_2$, $o_2$ is placed after $o$.
    Lastly, if $o$ is also dangling during the last read of $A$ in $o_2$,
    if $o_1$ is placed just before $o$ then $o_2$ is placed just after $o$, and if $o_1$ is placed just after $o$ then $o_2$ is placed just before $o$.
    We show that the latter case cannot happen.

    Assume that $o$ is dangling during the last read of $A$ in both $o_1$ and $o_2$, and $o_1$ is placed just after $o$ and $o_2$ is placed just before $o$.
    If
    $o_1$ is
    an $\insr(v)$ operation that returns \textit{false} or a $\lookup(v)$ operation that returns \textit{true} and $o_2$ is a $\del(v)$ or $\lookup(v)$ operation that returns \textit{false}, then $o$ is an $\insr(v)$ operation.
    As $o$ is a dangling $\insr(v)$ operation and by Lemma~\ref{lem:onedang}, there is only one dangling operation, $v$ is physically in $A$ during the last read of $A$ in $o_2$.
    By the code, $o_2$ identifies one of the conditions in Lemma~\ref{lem:elemnotpresentcond} based on the last read of $A$.
    Since $v$ is physically in $A$ during the last read of $A$ in $o_2$, the last read is of $A[h(v)]$ and $A[h(v)].\mem \prio{<}{h(v)} v$. By the code, $A[h(v)-1]$ is read before $A[h(v)]$ in $o_2$. 
    Since $v$ is not logically in $A$, the dangling insert operation performs an $\SC(A[h(v)-1], \tup{\ast,v,\I})$ in Line~\ref{lin:sc-insert}.
    In addition, this $\SC$ happens before $o_2$ invokes, which implies that $o_2$ reads $v$ from $A[h(v)-1]$, in contradiction.

    If
    $o_1$ is a $\del(v)$ or $\lookup(v)$ operation that returns \textit{false} and $o_2$ is
    an $\insr(v)$ operation that returns \textit{false} or a $\lookup(v)$ operation that returns \textit{true}, then $o$ is a $\del(v)$ operation.
    Right after a $\del(v)$ operation performs a successful $\SC(A[i],\ast)$ in Line~\ref{lin:sc-del}, $A[i] = \tup{a,v,\D}$,
    $a\neq v$ and $A[i+1] = \tup{v,\ast,\ast}$.
    Therefore, $v$ is physically in $A$ while $A[i]$ is unstable.
    If $i = h(v)-1$,
    following the code, when $A[i+1]$ becomes unstable, $v$ is no longer logically in $A$, and it stays this way until $A[i]$ stabilizes.
    Otherwise, $v$ is physically in $A$ until $A[i]$ stabilizes.
    By Lemma~\ref{lem:vnotinA}, during the last read of $A$ in $o_1$, $v$ is not logically in $A$, and it stays this way during the last read of $A$ in $o_2$.
    This contradicts Lemma~\ref{lem:vinA}, which states that 
    $v$ is logically in $A$ in this last read. \qedhere
\end{itemize}
\end{proof}

\subsection{History Independence}
\label{sec:implproof-hi}

When all cells are stable, there are no duplicate elements in the table, and the \emph{lookahead} slot in each cell is equal to the \emph{value} slot of the next cell. Together with Invariant~\ref{inv:ext-order-invar}, this implies that if no operation is propagating, the table is in the canonical memory representation of the set of elements stored in the table, imposed by Robin Hood hashing.
As discussed in Section~\ref{sec:propagate_operation}, we show that when an operation returns, it either completes its propagating, or a different ongoing $\insr$ or $\del$ operation takes over responsibility for finishing the propagation.
This ensures that when there are no ongoing $\insr$ or $\del$ operations, the table is in its canonical memory representation.

We say that an insert or delete operation is \emph{active} if there is a cell in the table unstable with this operation.
An \emph{empty cell} is a cell $A[i] = \tup{\bot,\ast,\ast}$ or $A[i] = \tup{\ast,\bot,\Stable}$.
Note that if $A[i] = \tup{\ast,\bot,\Stable}$, cell $A[i]$ itself is not empty in the sense that the value it contains is $\bot$, but it indicates that cell $A[i+1]$ is empty.
The propagation of an operation ends upon reaching an empty cell $A[i] = \tup{\bot,\ast,\Stable}$ for an insert operation (Line~\ref{lin:dsc-I1}), and empty cell $A[i] = \tup{\ast,\bot,\Stable}$ for a delete operation (Line~\ref{lin:dsc-D1}).
For a delete operation, the propagation also ends upon reaching an element in its hash location (Line~\ref{lin:dsc-D2}).
When the propagation of the operation ends, if $A[i]$ is unstable and reaches such cell $A[i+1]$, cell $A[i+1]$ does not become unstable with the operation.
By Lemma~\ref{lem:propstructure}, at most two cells are unstable with the same operation, and these two cells are consecutive.
In addition, once no cells are unstable with the operation, no other cell becomes unstable with the same operation.

It is easy to verify that an unstable cell is associated with exactly one active $\insr$ or $\del$ operation:

\begin{lemma}
\label{lem:activeop}
    Any unstable cell $i$, $0\leq i < \ncell$, contains an active operation:
    \begin{enumerate}
        \item If $A[i] = \tup{\ast,\ast,\I}$, then an $\insr$ operation is unstable in cell $i$.
        \item If $A[i] = \tup{\ast,\ast,\D}$, then a $\del$ operation is unstable in cell $i$.
    \end{enumerate}
\end{lemma}

Lemma~\ref{lem:activeop} implies that any cell without an active operation is stable.

The next lemma shows that the propagation of an operation that inserts or deletes element $v$ must end before reaching cell $h(v)$ again.

\begin{lemma}
\label{lem:finishprop}
    An active $\insr(v)$ or $\del(v)$ operation cannot be propagated beyond cell $h(v)-1$ after cell $h(v)-1$ becomes unstable with the operation.
\end{lemma}

\begin{proof}
    Assume an active $\insr(v)$ or $\del(v)$ operation propagated beyond cell $h(v)-1$ after cell $h(v)-1$ becomes unstable with the operation.
    That is, the operation propagates to $A[h(v)]$ after cell $A[h(v)-1]$ becomes unstable with the operation.
    Consider an $\insr(v)$ operation initially unstable in cell $i$.
    By Line~\ref{lin:sc-insert},
    when the operation is unstable in cell $i$, it is of the form $\tup{\ast, a_{i+1}, \I}$, where $a_{i+1} = v$.
    By Line~\ref{lin:dsc-I2},
    when the operation is unstable in cell $j$, $i< j\leq h(v)-1$ it is of the form $\tup{a_{j}, a_{j+1}, \I}$.
    By Line~\ref{lin:dsc-I1}, just after the operation propagates to cell $h(v)$ it is equal to $\tup{a_{h(v)},\ast,\ast}$.
    By Line~\ref{lin:insert-loc}, $v \prio{>}{i+1} a_{i+2}$, and by Line~\ref{lin:stuckinsr}, $a_{j} \prio{>}{j} a_{j+1}$ for $i < j\leq h(v)-1$.

    Consider a $\del(v)$ operation initially unstable in cell $i$.
    By Line~\ref{lin:sc-del},
    when the operation is unstable in cell $i$, it is of the form $\tup{\ast, a_{i+1}, \D}$, where $a_{i+1} = v$.
    By Line~\ref{lin:dsc-D1},
    Just before the successful $\SC$ that propagates the operation to cell $j$, $i < j\leq h(v)-1$, it is of the form $\tup{a_{j}, a_{j+1}, \Stable}$, and immediately after the successful $\SC$ it is of the form $\tup{a_{j+1}, a_{j+1}, \D}$.
    By the condition in Line~\ref{lin:dsc-D1}, 
    $h(a_{j+1}) \neq j+1$, for $i < j \leq h(v)-1$,
    and by \invref{inv:stable-unstable}, $a_j \prio{>}{j} a_{j+1}$.

    In both cases we have that for $i< j \leq h(v)-1$, $a_j \prio{>}{j} a_{j+1}$. This implies the next claim:

    \begin{claim}
    \label{clm:insredge}
        $h(a_{h(v)-1}) = h(a_{h(v)}) = h(v)$.
    \end{claim}

    \begin{proof}
        If $i = h(v)-2$, then $a_{h(v)-1} = v \prio{>}{h(v)-1} a_{h(v)}$, and this also implies that $h(a_{h(v)}) = h(v)$.
        Otherwise, we inductively show that $v \prio{>}{j+1} a_{j+1}$, $i < j < h(v)-1$.

        Base case: its given that $v \prio{>}{i+1} a_{i+2}$.
        As $h(v) \neq i+1$, by Lemma~\ref{lem:consistentprio}, $v \prio{>}{i+2} a_{i+2}$.
        
        Induction step:
        Assume the claim holds for $j$ and prove it for $j+1$, $i < j+1 < h(v)-1$.
        By the induction hypothesis, $v \prio{>}{j+1} a_{j+1}$.
        Since $a_{j+1} \prio{>}{j+1} a_{j+2}$, by transitivity of $\prio{>}{j+1}$, it also holds $v \prio{>}{j+1} a_{j+2}$.
        As  $h(v) \neq j+1$, by Lemma~\ref{lem:consistentprio}, $v \prio{>}{j+2} a_{j+2}$.

        By setting $j = h(v)-2$, we showed that $v \prio{>}{h(v)-1} a_{h(v)-1}$, which implies that $h(a_{h(v)-1}) = h(v)$. 
        Since $a_{h(v)-1} \prio{>}{h(v)-1} a_{h(v)}$, this also implies that $h(a_{h(v)}) = h(v)$.
    \end{proof}

    For a $\del(v)$ operation, Claim~\ref{clm:insredge} contradicts that $h(a_{h(v)}) \neq h(v)$.
    Consider an $\insr(v)$ operation;
    by Lemma~\ref{lem:propstructure}, just before 
    the successful $\SC$ that propagates the operation to $A[h(v)]$, $A[h(v)-1] = \tup{a_{h(v)-1}, a_{h(v)}, \I}$ and $A[h(v)] = \tup{a_{h(v)+1}, \ast, \Stable}$.
    As $a_{h(v)-1} \prio{>}{h(v)-1} a_{h(v)}$ and by Claim~\ref{clm:insredge}, it also holds
    $a_{h(v)-1} \prio{>}{h(v)} a_{h(v)}$.
    In addition, $a_{h(v)} \prio{>}{h(v)} a_{h(v)+1}$, and by the transitivity of $\prio{>}{h(v)}$, $a_{h(v)-1} \prio{>}{h(v)} a_{h(v)+1}$.
    However, since $A[h(v)-1].\mem = a_{h(v)-1}$, $A[h(v)].\mem = a_{h(v)+1}$ and $h(a_{h(v)-1}) = h(v)$, this contradicts \invenumref{inv:ext-order-invar}{1}.
\end{proof}

For each active operation, we assign a \emph{responsible} operation to ensure its complete propagation through the table.
During its execution, an active operation may have several operations assigned to it, but at any particular point, only one operation is responsible for it.
In addition, the same operation can be responsible for several active operations.
Initially, an active operation is responsible for itself.
Next, we explain how to assign responsible operations, in case the active operation returns before being propagated.

Fix an active operation $o_a$ and let $\resp(k)$, $k\geq 0$, be the $k$-th operation responsible to $o_a$. 
The responsible operation is either an insert or delete operation, but not a lookup operation.
We also define a point in the execution for each responsible operation in which the operation becomes responsible for $o_a$.
Define $\resp(0) = o_a$ and $\pt(0)$ to be the step in the execution which includes $o_a$ invocation.
The next lemma shows how to inductively find a responsible operation to $o_a$, ensuring that if a responsible operation returns before $o_a$ becomes inactive, then another operation becomes responsible to $o_a$.

\begin{lemma}
\label{lem:respactiveopind}
    Let $o_r = \resp(k)$, $k\geq 0$, be the $k$-th operation responsible to $o_a$.
    If $o_r$ returns while $o_a$ is still active we can find the next responsible operation $\resp(k+1)$ as follows;
    there is an operation overlapping $o_r$ that
    empties cell $j$ in Line~\ref{lin:dsc-D2}, with $d \neq \bot$, while $o_a$ is unstable in cell $i$ and previously unstable in cell $j$, and all cells between $j$ and $i$ are nonempty. 
    The successful $\SC$ in Line~\ref{lin:dsc-D2} is defined to be $\pt(k+1)$, and it holds that $\pt(k)$ precedes $\pt(k+1)$.
\end{lemma}

\begin{proof}
    \underline{$k = 0$}:
    In this case, $o_r = o_a$.
    Since $o_a$ is active, by Lemma~\ref{lem:finishprop}, it returns after calling $\prop$, and seeing an empty cell $j$.
    Assume the operation is unstable in cell $i$ during its last read of cell $A[j]$.
    As operation $o_a$ checks that the cell value changes or there is no unstable $\insr$ or $\del$ operation up to cell $j$, the active operation propagated to cell $j$.
    It must be that $i\neq j$, as by the propagation structure, only the first cell the operation is unstable in can have the $\mem$ slot equal to $\bot$, and by \invenumref{inv:stable-unstable}{2.5}, the $\lookahead$ is also not equal to $\bot$. This also implies that 
    cell $A[j]$ becomes unstable with the operation and then stabilizes.
    Consider the propagation of the operation from cell $j$ down to cell $i$.
    By the code, when the operation is propagated to some cell it must be nonempty,
    otherwise, this contradicts the fact that the operation is active.
    Therefore, cell $j$ is emptied after the the cell becomes unstable with the operation and before it is last read in the operation.
    
    Consider the closest cell to $i$, between cell $j$ and cell $i$, that is emptied after the active operation becomes unstable in this cell and not after cell $j$ is emptied. Let $k$ be this closest cell, $j\leq k < i$.
    While cell $k$ is emptied, the operation is unstable in some cell $k< \ell\leq i$.
    When cell $k$ is emptied, as we assume the operation is still active, the operation is propagated to at least cell $k+1$. From the choice of $k$, this implies that cell $k+1$ up to cell $\ell$ must be non-empty, and the cell must be emptied in Line~\ref{lin:dsc-D2} with $d\neq \bot$.
    The operation that empties cell $k$ is the operation we are looking for.

    \underline{$k \geq 1$}:
    By the assumption, operation $o_r$ empties cell $j$ in Line~\ref{lin:dsc-D2}, while $o_a$ is active in cell $i$ and all cells between $j$ and $i$ are nonempty.
    By the code, before $o_r$ returns and after emptying cell $j$, it calls $\prop$, starting from cell $j+1$.
    If $o_r$ ends the $\prop$ call upon reaching cell $j$, as $o_a$ was already active in cell $j$, by Lemma~\ref{lem:finishprop}, $o_a$ is no longer active.
    Therefore, the $\prop$ call ends after encountering an empty cell $j'$.
    Assume the operation is unstable in cell $i'$ during the last read of cell $A[j']$ in the $\prop$ call.
    Just Before calling $\prop$, there are no empty cells between cell $j+1$ and the cell $o_a$ is unstable in. Additionally, by Lemma~\ref{lem:dangpending},
    $o_a$ can only be unstable in nonempty cells, hence, cell $j'$ is emptied after cell $j$ is emptied.

    Consider the closest cell to $i'$, between cell $j'$ and cell $i'$, emptied after the active operation becomes unstable in this cell and not after cell $j'$ is emptied. Let $k$ be this closest cell, $j'\leq k < i'$.
    While cell $k$ is emptied, the operation is unstable in some cell $k< \ell\leq i'$.
    When cell $k$ is emptied, the operation is propagated to at least cell $k+1$. From the choice of $k$, this implies that cell $k+1$ up to cell $\ell$ must be non-empty, and the cell must be emptied in Line~\ref{lin:dsc-D2} with $d\neq \bot$
    after $\pt(k)$. 
\end{proof}

Lemma~\ref{lem:respactiveopind} implies that if an execution of the algorithm ends in a configuration with no pending
operations, then there is no active operation in this configuration.
The next lemma can be proven using Lemma~\ref{lem:nodangnoambig}.

\begin{lemma}
\label{lem:linquiescentstate}
    Consider a finite execution that ends in a configuration with no pending $\insr$ or $\del$ operation,
    if $v$ is present (resp., not present) in the table, then the state implied from $\pi$ (defined in Section~\ref{sec:implproof-linearizability}) includes (resp., does not include) $v$.
\end{lemma}

\begin{lemma}
    The hash table implementation is SQHI.
\end{lemma}

\begin{proof}
    Consider a configuration with no pending $\insr$ or $\del$ operations.
    Assume there is an active operation $o_a$ in this configuration.    
    By Lemma~\ref{lem:respactiveopind}, there is a sequence of operations, $o_a = \resp(0), \ldots \resp(k)$, $k\geq 0$, such that $\resp(j+1)$ overlaps $\resp(j)$, $0\leq j < k$. In addition, the sequence induce execution points, $\pt(0), \ldots, \pt(k)$, such that $\resp(j)$ is pending at $\pt(j)$, $0\leq j\leq k$ and $\pt(j)$ precedes $\pt(j+1)$, $0\leq j < k$. 
    This implies that from $o_a$ invocation and up until operation $\resp(k)$ returns, there is a pending operation in the execution.
    Since the execution reaches a configuration with no pending insert or delete operations after $o_a$'s invocation, $k$ is finite and $o_a$ must be inactive before $\resp(k)$ returns, in contradiction to the assumption that $o_a$ is active.
    We showed there is no active operation in this configuration, and by Lemma~\ref{lem:activeop}, all cells in $A$ are stable.
    This implies that for all $0\leq i < \ncell$, $A[i].\lookahead = A[i+1].\mem$ and
    all elements appear once in a $\mem$ slot in $A$.
    Let $P$ be the set of all elements in $A$, where every element appears once.
    By Lemma~\ref{lem:linquiescentstate},
    this is exactly the state implied from a linearization of an execution leading to this configuration.
    By \invenumref{inv:ext-order-invar}{1} and Lemma~\ref{lem:order-invariant}, $A|_{\mem} = \can(P)$, and for all $0\leq i < \ncell$, $A[i].\lookahead = \can(P)[i+1]$.
\end{proof}


\subsection{Lock-Freedom}
\label{sec:implproof-lockfree}

A failed $\SC$ operation indicates that there was either a successful initial write or a successful propagation of some operation. 
As we showed an operation propagates a finite number of times before it completes, and there is a finite number of processes, there can only be a finite number of initial writes before an operation finishes propagating.
Finally, we show that eventually, some process must return; otherwise, after some time, there will be no new initial writes, and all operations must complete propagating. In this case, the table becomes constant with only stable cells, and some process must eventually identify this and return.
The only exception is when only $\lookup$ operations take steps. In this case, we show that even if operations cannot finish propagating, since $\lookup$ operations do not puncture runs, some operation can still find the element it is looking for or detect its absence from the table.

The next lemma shows the condition in Line~\ref{lin:stuckinsr} only holds if all cells in the table are non-empty.

\begin{lemma}
\label{lem:consec-cells}
   Consider a multiset $P$ such that $|P|\leq \ncell-1$,  
   if $\canmult(P)[i] \neq \bot$ and $\canmult(P)[i+1] \neq \bot$,
   $0\leq i < \ncell$, then $\canmult(P)[i] \prio{\geq}{i+1} \can(P)[i+1]$.
\end{lemma}

\begin{proof}
    Denote $\canmult(P)[i] = a$ and $\canmult(P)[i+1] = b$.
    If $\canmult(P)[i] = \canmult(P)[i+1]$ we are done, so assume that $\canmult(P)[i] \neq \canmult(P)[i+1]$.
    If $\rank(b,i) < \rank(b,i+1)$, then it must also hold that $\rank(a,i) > \rank(b,i)$. Since there must be an empty cell in $\canmult(P)$, $h(a)\neq i +1$, and it also holds that $\rank(a,i+1) > \rank(b,i+1)$. 
    Otherwise, $h(b) = i+1$ and $\rank(b,i+1) = 0 \leq \rank(a,i) < \rank(a,i+1)$.
\end{proof}


\begin{lemma}
    The hash table implementation is lock-free.
\end{lemma}

\begin{proof}
    Consider execution prefix $\alpha$ with a pending operation. Assume there is an infinite extension of $\alpha$ where no operation returns.
    Since there is a finite number of processes, eventually, there are no new active operations, that is, no more successful $\SC$s in Line~\ref{lin:sc-insert} or \ref{lin:sc-del}.
    Consider this infinite suffix of $\alpha$'s extension.

    \begin{claim}
    \label{clm:returnorsc}
        Assume the table $A$ does not change in the suffix, if a process takes infinitely many steps, it must either perform an $\SC$ or return.
    \end{claim}

    \begin{proof}
        First, we show that if a process in an $\insr$ or $\del$ operation that takes infinitely many steps reaches $\helpop$, it performs an $\SC$.
        By Assumption~\ref{asm:neverfull} there is a stable cell in the cell, and therefore, the process eventually exists the while loop in Line~\ref{lin:helpop-while}.
        If the condition in Line~\ref{lin:stuckinsr} does not hold, the process performs an $\SC$. 
        By Lemma~\ref{lem:consec-cells} and \invenumref{inv:ext-order-invar}{1}, the condition in Line~\ref{lin:stuckinsr} can only hold when there is no empty cell in the table, and Assumption~\ref{asm:neverfull} ensures this does not happen.
        A $\lookup$ operation may not perform an $\SC$ inside $\helpop$ because of the condition in Line~\ref{lin:no-lookup}. However, Similar to $\insr$ and $\del$ operations, a $\lookup$ operation that takes steps inside $\helpop$ eventually exits it.
        
        Next, we show that a process taking infinitely many steps must either return, or, if it is an $\insr$ or $\del$ operation, call $\helpop$, and if it is a $\lookup$ operation, execute an $\SC$ inside $\helpop$.
        Consider a process that takes infinitely many steps inside an $\insr(v)$, $\del(v)$, or $\lookup(v)$ operation in the loop between Lines~\ref{lin:insert-startloop}--\ref{lin:insert-until}, Lines~\ref{lin:del-startloop}--\ref{lin:until-del} or Lines~\ref{lin:lookup-startloop}--\ref{lin:until-lookup} respectively.
        By the conditions in Lines~\ref{lin:insr-helpop} and \ref{lin:del-helpop}, a process in an $\insr$ or $\del$ operation calls
        $\helpop$, or only reads stable cells.
        
        First, assume $v$ is physically in the table. 
        If $v$ appears in a $\mem$ slot, by \invenumref{inv:ext-order-invar}{1}, for all $j$ between $h(v)$ and the first appearance of $v$, $A[j].\mem \prio{>}{j} v$. Otherwise, if $v$ appears only in a $\lookahead$ slot
        of cell $i$,  by \invenumref{inv:ext-order-invar}{3}, either $A[i].\mem \prio{>}{i} v$ or $i = h(v) - 1$.
        If $A[i].\mem \prio{>}{i} v$,
        by Lemma~\ref{lem:monotone-priority}, for all $j$ between $h(v)$ and $i$, $A[j].\mem \prio{>}{j} v$.
        Thus, since the loop starts at $A[h(v)-1]$ and the condition in the \textbf{until} statement does not hold before reaching this cell (Line~\ref{lin:insert-until}, Line~\ref{lin:until-del} or Line~\ref{lin:until-lookup}), eventually the loop reaches a cell where $v$ is present in its $\lookahead$ or $\lookahead$ slot.
        Since we assume all cells reached in the loop are stable, if an $\insr$ operation reaches a cell with $v$ in its $\lookahead$ slot, it is stable. Similarly, for a $\del$ operation, the loop eventually reaches a cell with $v$ in its $\lookahead$ slot.
        Hence, by the condition in Line~\ref{lin:insr-element-present}, an $\insr$ operation returns, and by the condition in Line~\ref{lin:del-loc}, a $\del$ operation executes an $\SC$.
        For a $\lookup$ operation, the only scenario the condition in Line~\ref{lin:lookup-elementfound} does not hold when the loop reaches cell $i$ that contains $v$ is when $A[i] = \tup{\ast,v,\D}$ and $h(v) = i + 1$.
        By \invenumref{inv:stable-unstable}{3}, either $A[i+1] = \tup{v,\ast,\ast}$, $A[i+1] = \tup{c,\ast,\D}$ where $c\neq v$, or $\tup{\bot,\ast,\ast}$. If $A[i+1] = \tup{v,\ast,\ast}$, then the condition in Line~\ref{lin:lookup-elementfound} holds for the next cell.
        Otherwise, $\helpop(i)$ is called and the condition in Line~\ref{lin:helpop-while} does not hold upon entering the while loop. Then, the condition in Line~\ref{lin:sc-D2-cond} holds and the operation executes an $\SC$.
        This shows that for a $\lookup$ operation, either the condition in Line~\ref{lin:insr-element-present} holds and the operation returns, or it executes an $\SC$.
        
        Assume $v$ is not physically in the table.
        By Lemma~\ref{lem:monotone-priority}, there are three options:
        \begin{enumerate*}
            \item for all cells $i$, $v \prio{>}{i} A[i].\mem$, 
            \item for all cells $i$, $A[i].\mem \prio{>}{i} v$, or
            \item there is an index $i$ such that for all cells $h(v) \leq j \leq i$, $A[j].\mem \prio{>}{j} v$, and for all cells $i < j < h(v)$,  $v \prio{>}{j} A[j].\mem$.
        \end{enumerate*}
        In the first case, 
        the condition in Line~\ref{lin:insert-loc} in an $\insr$ operation holds after reading the first cell in the loop $h(v)-1$, and
        the condition in Line~\ref{lin:del-elementnotpresent1} in a $\del$ operation or the condition on Line~\ref{lin:lookup-elementnotpresent1} in a $\lookup$ operation holds after reading the second cell in the loop $h(v)$.
        In the second case, the condition in the \textbf{until} statement never holds (Lines~\ref{lin:insert-until}, \ref{lin:until-del} or \ref{lin:until-lookup}), and eventually, the condition in Line~\ref{lin:del-elementnotpresent2} in a $\del$ operation or the condition in Line~\ref{lin:lookup-elementnotpresent2} in a $\lookup$ operation holds.
        By Assumption~\ref{asm:neverfull}, there is an empty cell, and a process in an insert operation eventually reaches it, and the condition in Line~\ref{lin:insert-loc} holds.
        In the third case, the loop reaches cell $i$ as the condition in the \textbf{until} statement does not hold up until this cell. 
        Since every cell reached in the loop is assumed to be stable for an $\insr$ or $\del$ operation, we have $A[i].\lookahead = A[i+1].\mem \prio{<}{i+1} v$ and $A[i].\mem \prio{>}{i} v$, and the condition in Line~\ref{lin:insert-loc} in an $\insr$ operation and the condition in Line~\ref{lin:del-elementnotpresent1} in a $\del$ operation holds.
        Consider a $\lookup$ operation. If $A[i]$ is stable, then similar to the previous case, the condition in Line~\ref{lin:lookup-elementnotpresent1} holds.
        If $A[i].\flag = \D$ and $A[i+1].\mem \neq A[i].\lookahead$, then, following the code, when the $\lookup$ operation reaches this cell it calls $\helpop$, the while loop stops at cell $A[i+1]$, and 
        the condition in Line~\ref{lin:sc-D2-cond} holds and the operation executes an $\SC$.
        If $A[j].\flag = \I$, either $A[j].\lookahead \prio{<}{j+1} v$ and the condition in Line~\ref{lin:lookup-elementnotpresent1} holds, or $A[j].\lookahead \prio{>}{j+1} v$, and the condition in Line~\ref{lin:lookup-elementnotpresent3} holds.
        In all cases, we showed the process returns, performs an $\SC$, or is an $\insr$ or $\del$ operation and it calls $\helpop$, which also implies the process performs an $\SC$.

        Consider a process that takes infinitely many steps after calling $\prop$,
        then the process takes infinitely many steps in a specific call of $\prop$.
        Assume the process never calls $\helpop$, then the condition in Line~\ref{lin:prop-helpcond} never holds. However, the condition in Line~\ref{lin:prop-stopcond} does not hold after a finite number of iterations, and the process exists the $\prop$ call.
    \end{proof}
    
    The next claim shows that if there is an active operation, after a finite number of steps, a propagation of an active operation down the table must occur.
    
    \begin{claim}
    \label{clm:prop}
        If there is an active operation at the ends of $\alpha$, after a finite number of steps, an unstable cell becomes stable in the infinite extension of $\alpha$.
    \end{claim}

    \begin{proof}
        Assume no active operation becomes inactive or propagated down the table in the infinite extension of $\alpha$.
        By the assumption, eventually, any $\SC$ in the execution happens in $\helpop$, and there are no new active operations.
        As there are a finite number of active operations, by Lemma~\ref{lem:propstructure}, there are a finite number of propagations of active operations before a cell stabilizes.
        Since an $\SC$ in the extension only stabilizes cells or propagates active operations, eventually, the values in the table do not change. Before any $\SC$, the executing process reads the updated values of the table, hence, eventually, there are no attempts to perform an $\SC$, otherwise, since the table is constant, it must be successful and propagate or deactivate an operation.
        However, by Claim~\ref{clm:returnorsc}, a process either returns or performs an $\SC$, 
        contradicting the initial assumption of the claim.
    \end{proof}

    By Lemma~\ref{lem:finishprop}, an active operation can be propagated a finite number of times. Thus, by repeatedly applying Claim~\ref{clm:prop} and since the number of active operations in the extension is bounded, eventually, all operations become inactive and the table becomes constant.
    However, this contradicts Claim~\ref{clm:returnorsc}, as no process can perform an $\SC$ or return.
\end{proof}

%% file: amortized_sketch.tex
\subsection{Amortized Step Complexity}
\label{sec:amortized_sketch}

This section proves that the amortized step complexity of our hash table is $O(c)$ per operation, where $c$ is the number of potentially-concurrent operations working on the same element (whether they be lookups, inserts or deletes).
Consider a batch of $n$ operations, each invoked by a different process.%
\footnote{This assumption is made to simplify the modeling and analysis, but it is not essential.}
Assume that the ratio between the number of insertions in the batch and the size of the hash table is bounded by some constant $\alpha \in (0,1)$.%
\footnote{Here, too, a more fine-grained analysis is possible, taking into consideration only elements that are in the hash table at the same time,
but for simplicity we take the total number of insertions as an upper bound.}
We prove that in expectation over the hash function,
even under a \emph{worst-case scheduler} that knows the hash function,
the total number of steps required to complete all operations is $O(n \cdot c)$.
First we sketch the proof and then provide the full details.

\subsubsection{Overview of the Analysis}

Fix a process $p$ that invokes an operation $o$ on element $x$ (i.e., $o \in \set{\insr(x), \del(x), \lookup(x)}$).
Let $N$ be a random variable
representing the length of the run that contains $h(x)$,
if we schedule all insertions from the batch prior to scheduling operation $o$.
This is an upper bound on the distance from $h(x)$ to the ``correct location'' for element $x$ (the cell where $x$ is stored, or if $x$ is not in the set,
the cell where $x$ \emph{would} be stored if it were in the set).
Also, let $P \leq N \cdot c$ be the number of processes working on elements in the same run as $x$.
We design a \emph{charging scheme} that charges each step taken by process $p$ to some operation --- either its own operation $o$, or an operation working in the same run as element $x$ --- in such a way that the total charge for any operation is at most $O(N^3 \cdot c)$.
We then prove that if the load on the hash 
table is bounded away from 1,
then $\E[N^3] = O(1)$,
where the expectation is over the choice of the hash function.
By linearity of expectation, the expected number of steps required for all operations to complete is $O(n \cdot c)$.

The charging scheme assigns to each operation all \emph{successful} steps performed ``on its behalf'', regardless of which process performed them. Propagating an operation one position forward requires only a constant number of successful steps, so the total charge here is $O(N)$. In contrast, \emph{failed} steps are charged to the operation whose successful $\SC$ caused the failure. For example, if a process $p$ attempts to propagate an operation $o$ from position $i$ to position $i+1$, but some other process $q$ succeeds in doing so faster than $p$, then the $O(1)$ steps that process $p$ ``wasted''  are charged to operation $o$;
but if $q$'s failure is instead caused by some operation $o' \neq o$ making its initial write into position $i+1$,
then the wasted steps are charged to operation $o'$.
The largest charge is for causing a process to restart its operation prior to its initial write into the table; the process may have wasted up to $N$ steps,
and all are charged to the operation whose successful $\SC$ caused the restart.
Each operation is propagated at most once across each position in the run (at most $N$ positions), and each propagation step involves a constant number of successful $\SC$s, which each fail up to $P$ processes, possibly causing them to restart and costing $N$ wasted steps. Thus,
the total charge to an operation is bounded by $O(N \cdot N \cdot P) = O(N^3 \cdot c)$.

%% file: amortized.tex
\subsubsection{Amortized Runtime Analysis}
\label{sec:amortized}

We begin by proving that a process that restarts its lookup can ``blame'' some other process for this:
\begin{lemma}
    Let $p$ be a process that restarts its operation
    $o$ on element $x \in \calU$
    during its lookup stage,
    after reading cell $i$.
    Then between the time when 
    $p$ last read cell $i - 1$
    and the time it read cell $i$ and restarted,
    some insert or delete operation was propagated across cell $i-1$.
     \label{lemma:restart_blame}
\end{lemma}
\begin{proof}
    Let $t_{i-1}$ be the time when cell $i-1$ was last read prior to restarting,
    and let $t_{i}$ be the time cell $i$ was read,
    causing the operation to restart.
    At time $t_{i-1}$
    we had $A[i-1] = \tup{a, b, M}$
    for some values $a,b \in \calU$ and $M\in \set{\I,\D,\Stable}$.
    If $a \neq x$ and $b \neq x$, 
    since the process then continued to read cell $i$, we know that $b \prio{>}{i} x$ (otherwise, the operation would either restart at cell $i - 1$, return, or make its initial write into cell $i - 1$).
    First, assume that $M = \Stable$ and the cell is stable. This must hold for an insert or delete operation, otherwise $p$ would have helped
    the operation working on cell $i-1$
    and then re-read the cell.
    If $a = x$ or $b = x$, the operation would return if it is $\lookup(x)$ or $\insr(x)$,
    and a $\del(x)$ would either make its initial write into cell $i-1$ if $b = x$,
    or it would step back to cell $i - 2$ if $a = x$
    and then either make its initial write there,
    continue moving backwards, resume moving forwards, or restart at cell $i - 2$.
    In all of these cases, time $t_{i-1}$ is not the last time that cell $i-1$ is read prior to restarting at cell $i$.
    Thus, we may assume that $a \neq x$ and $b \neq x$.
    All together, since cell $i - 1$ is stable at time $t_{i-1}$
    and its \emph{lookahead} slot is $b$,
    the \emph{value} slot of cell $i$ is also $b$.
    If $M \in \set{\I,\D}$, then this must be a lookup operation, as explained above.
    If $M = \I$, then $a\neq v$ and $b\neq v$, as otherwise, the operation returns.
    Thus, $b \prio{>}{i} v$. 
    Cell $i$ is read by $p$ between $t_{i-1}$ and $t_i$, and at that time, either the value of cell $i$ changes, which implies there was a propagation across cell $i - 1$, or
    $A[i] = \tup{c, \ast, \ast}$ such that $c \prio{\geq}{i} v$. 
    Since after the operation propagates to cell $i$ the \emph{value} slot of cell $i$ is equal to the \emph{lookahead} slot of cell $i-1$, at time $t_{i-1}$, either $A[i] = \tup{c, \ast, \ast}$ or $A[i] = \tup{b, \ast, \ast}$.
    If $M = \D$, then at time $t_{i-1}$, either $A[i] = \tup{b, \ast, \ast}$ or $A[i] = \tup{c, \ast, \ast}$ such that $b \prio{>}{i} c$. In the latter case, as $p$ helps the operation working on cell $i-1$ after it propagated to cell $i$, it must complete its propagation to cell $i$ by stabilizing cell $i$. Otherwise, either $a\neq v$ and $b\neq v$, which implies that $b \prio{>}{i} v$, or $b = v$. Either way, $b \prio{\geq}{i} v$.
    We showed that in all cases, during time $t_i$, $A[i] = \tup{c',\ast,\ast}$, where $c' \prio{\geq}{i} v$.
    
    Since we know that the process restarted after reading $A[i]$ at time $t_i$,
    at that time we had $A[i] = \tup{c, \ast, \ast}$ 
and $c \prio{<}{i} x$.
It follows that some operation modified the \emph{value} slot of cell $i$ from its previous value, $c' \prio{\geq}{i} x$,
to its current value, $c \prio{<}{i} x$.
This can only happen if some operation propagated across cell $i - 1$ and into cell $i$.
\end{proof}

For our analysis we also need to bound the third moment of the length of the run to which an arbitrary element belongs.
We do this by adapting the analysis from~\cite{Thorup}, which bounds the expected length of the run (i.e., the first moment).
\begin{lemma}
Let $\Op$ be a batch of $n$ sequential operations.
For any constants $k \in \mathbb{N}^+$ and $\alpha \in (0,1)$,
if the load on the table is at most $\alpha$,
then for any element $x \in \calU$,
if $N$ denotes the length of the run containing position $h(x)$,
we have $\E\left[ N^k \right] = O_k(1)$.
    \label{lemma:run_length}
\end{lemma}
\begin{proof}
    As we said, we adapt the analysis of~\cite{Thorup},
    which partitions the table into dyadic intervals, and shows that a long run can be ``blamed'' on some interval having a higher than expected load --- a low-probability event.
    
    Fix the set $S \subseteq \calU$ of elements that are currently in the hash table,
    and recall that $\alpha = |S|/\ncell$ is the load on the hash table.
    For $i = 0,\ldots,\log \ncell$,
    an \emph{$\ell$-block} is a consecutive interval of $2^{\ell}$ cells starting from some cell of the form $i \cdot 2^{\ell}$ for $i \geq 0$, that is,
    $B = [i \cdot 2^{\ell}, (i+1) \cdot 2^{\ell} - 1]$.
    The \emph{load} on block $B$
    is given by $L_h(B) = \left| \set{ y \in S \setminus \set{x} : h(y) \in B } \right|$.
    A block is called \emph{nearly full} if 
    $L_h(B) \geq (1+\delta) \alpha |B|$,
    where $\delta \in (0,1)$
    is a parameter whose value will be fixed later,
    subject to the constraint that $(1+\delta)\alpha < 1$.

    Let $R$ be the run containing $h(x)$.
    As shown in~\cite{Thorup} (Lemma 3),
    for every $\ell \geq 0$,
    if $|R| \in [2^{\ell+2}, 2^{\ell+3})$,
    then 
    one of the $12$ consecutive $\ell$-blocks 
    starting $8$ blocks before the $\ell$-block containing $h(x)$ and ending $3$ blocks, is nearly full.
    Thus, if $P_{\ell}$ is the probability that an $\ell$-block is full,
    we have
    \begin{equation*}
        \Pr\left[ |R| \in [2^{\ell+2}, 2^{\ell + 3} ) \right]
        \leq
        12P_{\ell}.
    \end{equation*}

    The \emph{expected} load on an $\ell$-block $B$ is given by
    \begin{equation*}
        \E[ L_h(B) ] = \sum_{y \in S \setminus \set{x}} \Pr[h(y) \in B] = 
        |S| \cdot \frac{2^{\ell}}{\ncell} 
        = 2^{\ell} \cdot \alpha,
    \end{equation*}
    where $\alpha = |S|/\ncell$ is the load on the table.
    Thus, if $B$ is nearly full, we have
    $L_h(B) \geq (1+\delta) \E[L_h(B)]$.
    In~\cite{Thorup}, the hash function family is assumed to be 5-wise independent, but here we consider a fully independent hash function family,%
    \footnote{In fact, as in~\cite{Thorup}, a $t$-wise independent hash family for a sufficiently large constant $t$ which depends on $k$ should suffice.}
    so by Chernoff,
    \begin{equation*}
        P_{\ell} \leq \Pr\left[ L_h(B) - \E[L_h(B)] \geq (1+\delta) \E[L_h(B)] \right] \leq e^{- \delta^2 \E[L_h(B)] / 3} 
        = e^{-\delta^2 \alpha 2^{\ell} / 3}.
    \end{equation*}

    To bound the expected $k$-th moment of $N$,
    we can write
    \begin{align*}
        \E[ N^k ]
        &\leq
        1^k + 2^k + 
        \sum_{\ell = 0}^{\log \ncell} \left( 2^{\ell + 3} \right)^k \cdot P_{\ell}
        \\
        &
        \leq
        1 + 2^k
        \sum_{\ell = 0}^{\log \ncell} e^{(\ell+3)k \log e } \cdot e^{-\delta^2 \alpha 2^{\ell} / 3}
        \\
        &
        = 1 + 2^k + \sum_{\ell = 0}^{\log \ncell} e^{ (\ell+3)k \log e - \delta^2 \alpha 2^{\ell} / 3}.
    \end{align*}
    For sufficiently large $\ell$, the exponent is bounded from above by $- \delta^2 \alpha 2^{\ell} / 6$,
    and therefore, for some constant $\beta > 0$,
    the sum can be written as $O(1) + \sum_{\ell = \ell_0}^{\log \ncell} e^{- \beta \cdot 2^{\ell}} = O(1)$.
    All together, we have $\E[ N^k ] = O(1)$,
    where the $O(\cdot)$-notation hides constants depending on $k$ and $\alpha$.    
\end{proof}

For a hash function $h : \calU \rightarrow [\ncell]$,
a set of $n$ operations $\Op$, each performed by a different process,
and a schedule $\sigma \in [n]^{\omega}$,
let $T_h(\Op, \sigma)$ denote the total running time of all operations in $\Op$ when $h$ is the hash function and $\sigma$ is the schedule (not counting steps of processes that have already completed their operation).
Let $T_h(\Op) = \max_{\sigma \in [n]^{\omega}} T_h(\Op,\sigma)$ be the worst-case running time over all schedulers.
The maximum is well-defined: 
let $\mathcal{T}_h(\Op)$ be a tree representing all possible executions of the operations in $\Op$ with the hash function $h$:
every node of $\mathcal{T}_h(\Op)$ represents a prefix of an execution,
and has at most $n$ children, one for each process that may take the next step, if it has not completed its operation yet.
Nodes representing execution prefixes where all operations have completed have no children.
By lock-freedom, every path in $\mathcal{T}_h(\Op)$ is finite, as under every schedule, all operations eventually complete.
Also, the tree has finite branching factor (at most $n$).
Thus, by K\H{o}nig's Lemma, the tree must be finite,
and $T_h(\Op)$, which is the depth of the tree, is well-defined.

Given the contents $A : [\ncell] \rightarrow \calU \times \CalU \times \set{ \Stable, \I, \D}$
of the hash table,
we define the \emph{target position} $\posit(x)$
of an element $x \in \calU$ as follows:
\begin{itemize}
    \item If there is a cell $i$ such that $A[i] = \tup{a, b, \ast}$ where either $a = x$ or $b = x$, then $\posit(x)$ is the closest index of such a cell to $h(x)$ (in the sense of the directed distance going forward from $h(x)$).
    \item Otherwise, $\posit(x)$ is the index of the closest cell $i$ to $h(x)$ such that 
    $A[i] = \tup{a, b, \ast}$,
    with $a \prio{>}{i} x \prio{>}{i+1} b$.
\end{itemize}

\begin{lemma}
    Fix a set $\Op$ of $n$ operations,
    and let $c$ be the maximum number of operations 
    working on the same element (whether they be insertions, deletions or lookups).
    If $0 < \alpha < 1$ for some constant $\alpha$,
    Then
    \begin{equation*}
        \E_h \left[ T_h(\Op)  \right] = O( n \cdot c).
    \end{equation*}
\label{lemma:runtime}
\end{lemma}

\begin{proof}

For a fixed hash function $h : \CalU \rightarrow [\ncell]$,
and an operation $o_i \in \Op$ on element $x \in \calU$,
define:
\begin{itemize}
        \item $N_i = N_i(h)$: the length of the run containing $h(x) - 1$, if we scheduled all the insertions from the batch $\Op$ and then $o_i$.
    This is an \emph{upper bound} on the distance from $h(x)$ to $\posit(x)$ at any point in any execution prefix of $\Op$.
    \item $P_i = P_i(h) \leq N_i \cdot c$: the number of processes whose elements belong to the same run as operation $o_i$'s element.
\end{itemize}

    We account for the steps of an operation $o \in \Op$ 
    on element $x \in \calU$ as follows.
    \begin{itemize}
    \item In the lookup phase, steps where the process moves from $h(x)-1$ to 
    $\posit(x)$
    and does not help another operation:
    \begin{itemize}
        \item If the process ultimately restarts, then all steps taken in the current pass from $h(x)-1$ to
        the location where the lookup restarts are charged to some insert or delete operation that caused the discrepancy:
        by Lemma~\ref{lemma:restart_blame},
        some insert or delete must have propagated across the position the process read just prior to restarting,
        and we charge this operation (or one of them, if there are multiple).
        Since an operation only propagates at most once over each position in the table,
        it can only be blamed once by each process for each position in the run.
        Therefore the total charge to any operation incurred by all restarts is $O( N_i \cdot N_i \cdot P_i) = O(N_i^3 \cdot c)$.
        \item If the process does not restart, this pass is charged to operation $o$ itself.
        The total charge is $O(N_i)$,
        as this is the maximum possible distance from  $h(x) - 1$ to the target position $\posit(x)$
        where the process either returns or attempts to make its initial write.
    \end{itemize}
    \item Failed initial writes, along with the pass from $h(x)-1$ to the target location that preceded (at most $N_i+1$ steps), are charged to the operation whose successful $\SC$ caused the failure.
    Since operations propagate at most once across any location,
    and propagating one step forward involves a constant number of successful $\SC$s,
    the total charge to an operation is $O( N_i \cdot N_i \cdot P_i ) = O( N_i^3 \cdot c)$.
    \item The successful initial write of an insert/delete operation $o$ is charged to operation $o$ itself.
    \item Following the initial write,
    steps that help an operation $o'$ are charged as follows.
    Propagating $o'$ one step forward can be broken up into a constant number of \emph{stages}, each of which consists of a constant number of $\LL$s, followed by a $\VL$ or $\SC$.      
            After all stages are successfully completed, the operation moves forward one location.
            This \emph{excludes} ``clearing the way'' by helping other operations advance (e.g., when two consecutive cells are unstable, the first cannot advance until the second becomes stable), as those steps will be charged to the other operations.

        Successful stages in the propagation of $o'$ are charged to $o'$.
        On the other hand, \emph{failed} stages (i.e., those that end with a failed $\VL$ or $\SC$) are charged to:
        \begin{itemize}
            \item Operation $o'$, if the successful $\SC$ that caused the failure is part of the propagation of $o'$; or
            \item An operation $o''$ whose initial write caused the failure.
        \end{itemize}
        These are the only two possibilities for a failed stage.
        The total charge to each operation is at most $O(1)$ per process and location in the run, as each operation propagates at most once across each location in the run.
        Therefore the total charge for an operation is $O( N_i \cdot P_i) = O(N_i^2 \cdot c)$.
   \end{itemize}

The total charge to an operation $o$ is bounded by $O(N_i^3 \cdot c)$.
Ignoring constants, the total cost of scheduling all $n$ operations is bounded by
\begin{align*}
    c \cdot \sum_{i = 1}^n N_i^3.
\end{align*}
This is true for \emph{any} schedule, and therefore it is true also for the maximum over all schedules.
Since $\E[N_i^3] = 1$, in expectation the total cost is $O(n \cdot c)$. 
\end{proof}

%% file: lower-bounds.tex
\section{Lower Bounds for History-Independent Dictionaries}
\label{sec:lbs}

To prove the lower bound (formalized below in Theorem~\ref{thm:natural-assig-formal}) 
we first use two properties of canonical memory assignments that serve 
as obstacles for wait-free history-independent dictionary implementations
(see Theorem~\ref{thm:waitfreelb} below).
We focus on the case where the number of memory cells $\ncell$ is smaller than the universe size $|\CalU|$. (If
$\ncell = |\CalU|$ then there is a trivial implementation where every memory cell represents a single element in the universe, and is set to 1 if the element is in the set, or 0 otherwise.)
Let $\can(q)$ denote the canonical memory representation of state $q\in Q$.


The first property captures a scenario 
where for some element $v \in \CalU$,
no single memory cell indicates whether $v$ is in the set:
for every memory cell $\ell$,
there exist two states $q, q'$ of the set
where memory cell $\ell$ has the same value ($\can\parens{q}[\ell] = \can\parens{q'}[\ell]$),
but state $q$ includes $v$ while $q'$ does not.

\begin{property}
\label{prop:lb-core}
    For element $v \in \calU$,
    there exists an initial configuration $C_0$
    such that for every $0 \leq \ell < \ncell$, there are states $q, q'\in Q$ such that $v\in q$, $v\notin q$ and $\can\parens{q}[\ell] = \can\parens{q'}[\ell]$.
\end{property}

When the implementation uses \emph{read-modify-write} 
primitives,
we can show that any wait-free history-independent implementation induces a canonical memory assignment where no element satisfies Property~\ref{prop:lb-core}.
These primitives read a memory cell and write a new value that depends on the old value in one atomic step.
In the case of the $\LL/\SC$ primitive, 
where the new value may depend not only on the current state of the cell but also on the history of previous operations,
we must introduce an additional property to reach the same impossibility result.
The second property captures a scenario 
where for some element $v \in \CalU$,
a memory cell cannot have the same value across all states that include 
$v$ or all states that do not include $v$.

\begin{property}
\label{prop:llsc}
For element $v \in \calU$,
    there exists an initial configuration $C_0$ 
    such that for every $0 \leq \ell < \ncell$,
    there are states $q_1, q'_1, q_2, q'_2\in Q$ 
    where $\can\parens{q_1}[\ell] \neq \can\parens{q'_1}[\ell]$ and $\can\parens{q_2}[\ell] \neq \can\parens{q'_2}[\ell]$,
    but $v \in q_1,q'_1$ and $v \notin q_2,q'_2$.   
\end{property}


Although Properties~\ref{prop:lb-core} and~\ref{prop:llsc}
are properties of individual elements, we abuse the terminology by saying that an \emph{implementation} satisfies Property~\ref{prop:lb-core} (resp.\ Property~\ref{prop:llsc}) if some element satisfies Property~\ref{prop:lb-core} (resp.\ Property~\ref{prop:llsc}).
We abuse the terminology even further by saying that an implementation satisfies \emph{both} properties if there is some element that satisfies both at the same time.

In the proof of the next theorem, we construct two executions $\alpha_0, \alpha_1$, such that in $\alpha_0$,
element $v$ is not in the dictionary throughout, and in $\alpha_1$, element $v$ \emph{is} in the dictionary throughout;
however, there is a $\lookup(v)$ operation that cannot distinguish $\alpha_0$ from $\alpha_1$, and therefore cannot return.
To ``confuse'' the lookup operation, every time it accesses a cell $\ell$
using an operation whose behavior depends only on the cell's current value (such as read/write or compare-and-swap),
we use Property~\ref{prop:lb-core} to extend executions $\alpha_0, \alpha_1$ in a way that the reader cannot distinguish between them and gain clear-cut information about the presence of $v$ in the dictionary;
and if the access is using $\VL$ or $\SC$,
we use Property~\ref{prop:llsc} to first ``overwrite'' the value of cell $\ell$, ensuring the $\VL$ or $\SC$ returns \emph{false} in both executions.
Since the $\lookup$ operation cannot distinguish $\alpha$ from $\alpha'$,
it is doomed to return an incorrect answer in one of them (or to never return).
The full proof appears in Section~\ref{sec:proof of waitfreelb}.

\begin{restatable}{theorem}{waitfreelb}
\label{thm:waitfreelb}
    There is no SQHI implementation of a dictionary with a wait-free $\lookup$ operation and obstruction-free $\insr$ and $\del$ operations from read-modify-write and $\LL/\SC$ primitives that induce 
    a canonical memory assignment satisfying both Property~\ref{prop:lb-core} and Property~\ref{prop:llsc}.
\end{restatable}


\subsection{Proof of Theorem~\ref{thm:waitfreelb}}
\label{sec:proof of waitfreelb}

Consider a SQHI implementation with wait-free $\lookup$ operation that, for initial configuration $C_0$, induces an assignment $\can$ that satisfies both Property~\ref{prop:lb-core} and Property~\ref{prop:llsc}.
The properties imply there is an element $v \in \CalU$ such that for every $0 \leq \ell < \ncell$ there are states $q, q'\in Q$ such that $v\in q$, $v'\notin q$ and $\can\parens{q}[\ell] = \can\parens{q'}[\ell]$, and states $q_1, q'_1, q_2, q'_2\in Q$ 
where $\can\parens{q_1}[\ell] \neq \can\parens{q'_1}[\ell]$ and $\can\parens{q_2}[\ell] \neq \can\parens{q'_2}[\ell]$,
but $v \in q_1,q'_1$ and $v \notin q_2,q'_2$. 

 We consider executions of the implementation with two processes, a ``reader'' process $r$ that executes a single $\lookup(v)$ operation and a ``writer" process $w$ that repeatedly executes $\insr$ and $\del$ operations.
The executions that we construct have the following form:
\begin{equation*}
    S_1, r_1, S_2, r_2, \ldots, S_k, r_k,
\end{equation*}
where $S_i$ is a (possibly empty) sequence of operations executed by the writer process, during which the reader process takes no steps,
and $r_i$ is a single step by the reader process.
We abuse notation and for $k=0$ we get an empty execution.

 In any linearization of $\alpha = S_1, r_1, \ldots, S_k, r_k$, the operations in the sequences $S_1, \ldots, S_k$ must be linearized in order, as they do not overlap.
Furthermore, the $\lookup(v)$ operation carried out by the reader is not state-changing.
Thus, the linearization of $\alpha$ ends with the object in state $q$, where $q$ is the state reached by applying the operation' sequence $S_1, \ldots, S_k$ from the initial state.
We abuse the terminology by saying that the execution ``ends at state $q$''.

We say that an execution $\alpha = S_1, r_1, S_2, r_2, \ldots, S_k, r_k$ \emph{avoids} (\emph{never avoids}) $v$ if for every state $q$ traversed during the sequence of operations $S_1, S_2, \ldots, S_k$, not including the initial state, we have $v\notin q$ ($v\in q$).
An empty execution $\alpha$ both avoids and never avoids $v$; this is fine for our purposes, because the reader 
only starts running after the first operation' sequence $S_1$.

\begin{restatable}{lemma}{diffreturn}
    \label{lem:diff-return}
        If execution $\alpha = S_1, r_1, \ldots, S_k, r_k$ avoids (resp., never avoids) $v\in \CalU$, then
        the $\lookup(v)$ operation can only return the value \textit{false} (resp., \textit{true}) at any point in $\alpha$.
\end{restatable}

\begin{proof}
    Fix an execution $\alpha = S_1, r_1, S_2, r_2, \ldots, S_k, r_k$ and note that in any linearization,
    the operations in the sequences $S_1, \ldots, S_k$
    must be linearized in-order,
    as they are non-overlapping operations by the same process.
    The $\lookup(v)$ operation cannot be linearized before
    the first sequence $S_1$,
    because it is only invoked after this sequence of operations completes.
    Thus, the $\lookup(v)$ operation 
    either does not return in $\alpha$,
    or it is linearized after some operation in a sequence
    $S_j$, $1\leq j\leq k$. 
    In the latter case, if $\alpha$ avoids $v$, the $\lookup(v)$ operation is linearized after a state that does not contain element $v$ and returns \textit{false}. Similarly, if $\alpha$ never avoids $v$, the $\lookup(v)$ operation is linearized after a state that contains element $v$ and returns \textit{true}.
\end{proof}

To build an execution that avoids or never avoids $v$, we need to traverse between two states in a way that preserves the absence or presence of $v$ in the set:

\begin{restatable}{lemma}{opseq}
        \label{lem:op-seq}
        Let $v\in \CalU$.
        For every two states $q_1,q_2\in Q$ such that $v\in q_2$ (resp., $v\notin q_2$), there is a sequence of operations $\seq(q_1,q_2)$ that takes the set from state $q_1$ to state $q_2$ and goes through only states that include $v$ (resp., do not include $v$), not including the starting state $q_1$.
\end{restatable}

\begin{proof}
    If $v\in q_2$ and $v\notin q_1$, the sequence $\seq(q_1,q_2)$ begins with an $\insr(v)$ operation.
    If $v\notin q_2$ and $v\in q_1$, the sequence $\seq(q_1,q_2)$ begins with an $\del(v)$ operation.
    The sequence $\seq(q_1,q_2)$ continues
    with a $\del(u)$ operation for each element $u\in q_1 \setminus q_2$ in an arbitrary order, and then with a $\insr(u)$ operation for each element $u\in q_2 \setminus q_1$ in an arbitrary order.
    Since all the elements not in $q_2$ and in $q_1$ are removed and elements in $q_2$ and not in $q_1$ are inserted, this sequence takes the set from state $q_1$ to state $q_2$. If $v\in q_2$, element $v$ is inside the first inner state, and since it is never removed, any inner state between the two states also includes element $v$.
    Similarly, if $v\notin q_2$, element $v$ is not inside the first inner state, and since it is never inserted, any inner state between the two states also does not include element $v$.
\end{proof}

For the purpose of the impossibility result, we assume, that the local state of a process $p$ contains the complete history of $p$'s invocations and responses. 
Two finite executions $\alpha_1$ and $\alpha_2$ are \emph{indistinguishable} 
to the reader, denoted $\alpha_1 \localind{r} \alpha_2$, 
if the reader is in the same state in the final configurations of $\alpha_1$ and $\alpha_2$.

\begin{restatable}{lemma}{extendexecwf}
    \label{lem:extend-exec-wf}
    Fix $k \geq 0$,
    and suppose we are given two
    executions of the form
    $\alpha_1 = S^1_1, r_1, \ldots, S^1_k, r_k$ and $\alpha_2 = S^2_1, r_1, \ldots, S^2_k, r_k$
    such that $\alpha_1 \localind{r} \alpha_2$, $\alpha_1$ avoids $v$ and $\alpha_2$ never avoids $v$.
    Then we can extend $\alpha_1$
    into an execution $\alpha'_1 = 
    S^1_1, r_1, \ldots, S^1_k, r_k, S^1_{k+1}, r_{k+1}$ that also avoids $v$ and $\alpha_2$ into an execution $\alpha'_2 = 
    S^2_1, r_1, \ldots, S^2_k, r_k, S^2_{k+1}, r_{k+1}$ that also never avoids $v$,
    such that $\alpha'_1 \localind{r} \alpha'_2$.
\end{restatable}

\begin{proof}
    By assumption, the reader is in the same local state at the end of both executions $\alpha_1$ and $\alpha_2$
    and so its next step is the same in all of them. Let cell $\ell$, $0\leq \ell < \ncell$,
    be the memory cell
    accessed by the reader in its next step
    in both executions.

    Since $\can$ satisfies property~\ref{prop:lb-core} for element $v$, there are states $q, q'\in Q$ such that $v\in q$, $v\notin q'$ and $\can\parens{q}[\ell] = \can\parens{q}[\ell]$.
    Let $q_1$ be the state $\alpha_1$ ends at, and $q_2$ the state $\alpha_2$ ends at.
    Since $\can$ satisfies property~\ref{prop:llsc} for element $v$, there a state $q'_1$ such that $v\notin q'_1$ and $\can\parens{q_1}[\ell] \neq \can\parens{q'_1}[\ell]$. Similarly, there a state $q'_2$ such that $v\in q'_1$ and $\can\parens{q_2}[\ell] \neq \can\parens{q'_2}[\ell]$.
    
    We extend $\alpha_1 = S^1_1, r_1, \ldots, S^1_k, r_k$
    into $\alpha'_1 = S^1_1, r_1, \ldots, S^1_{k+1}, r_{k+1}$,
    by appending
    a sequence of complete operations $S^1_{k+1} = \seq(q_1, q'_1)\seq(q'_1, q')$ according to Lemma~\ref{lem:op-seq}, followed by a single
    step of the reader.
    The resulting execution still avoids $v$ by Lemma~\ref{lem:op-seq} as $v\notin q'_1, q'$.
    We extend  $\alpha_2 = S^2_1, r_1, \ldots, S^2_k, r_k$
    into $\alpha'_2 = S^2_1, r_1, \ldots, S^2_{k+1}, r_{k+1}$,
    by appending
    a sequence of complete operations $S^2_{k+1} = \seq(q_2, q'_2)\seq(q'_2, q)$ according to Lemma~\ref{lem:op-seq}, followed by a single step of the reader.
    Since we assume inserts and deletes are obstruction-free, there is such an execution.
    The execution still never avoids $v$ by Lemma~\ref{lem:op-seq}, as $v\in q'_2, q$.
    If the reader's next step is an $\SC$ or $\VL$ operation, since 
    $\can\parens{q_1}[\ell] \neq \can\parens{q'_1}[\ell]$ and $\can\parens{q_2}[\ell] \neq \can\parens{q'_2}[\ell]$, the value of cell $\ell$ changes since the reader's most recent $\LL$ in both executions. Hence, the $\SC$ or $\VL$ operation must fail in both $\alpha'_1$ and $\alpha'_2$.
    Otherwise,
    since $\alpha'_1$ ends at state $q'$ and $\alpha'_2$ ends at state $q$, and $\can\parens{q}[\ell] = \can\parens{q'}[\ell]$, 
    when the reader takes its step,
    it observes the same state and response for the memory cell $\ell$ that it accesses in both executions.
    Therefore, the reader cannot distinguish
    the new executions from one another.
\end{proof}

\begin{proof}
[Proof of Theorem~\ref{thm:waitfreelb}]
    We construct two
    arbitrarily long executions, in each of which a $\lookup$
    operation takes infinitely many steps but never returns.
    The construction uses Lemma~\ref{lem:extend-exec-wf}
    inductively:
    we begin with empty executions,
    $\alpha_1^0$ and $\alpha_2^0$.
    These executions trivially satisfy the conditions of Lemma~\ref{lem:extend-exec-wf},
    as the reader has yet to take a single step
    in any of them and is in the same local state in both executions.
    We repeatedly apply Lemma~\ref{lem:extend-exec-wf}
    to extend these executions,
    obtaining for each $k \geq 0$ two executions
     $\alpha_1^k$ and $\alpha_2^k$,
    such that $\alpha_1^k$ avoids $v$, $\alpha_2^k$ never avoids $v$ and the reader cannot distinguish the executions from one another.

    Suppose for the sake of contradiction
    that the reader returns a value $bool$
    at some point in $\alpha_1^k$.
    Then it returns the same value $bool$
    at some point in execution $\alpha_2^k$,
    as it cannot distinguish these executions,
    and its local state encodes
    all the steps it has taken,
    including whether it has 
    returned a value,
    and if so, what value.
    However, by Lemma~\ref{lem:diff-return}, the reader can only return the value \textit{false} in execution $\alpha_1^k$ and the value \textit{true} in execution $\alpha_2^k$, in contradiction.    
\end{proof}

\subsection{Natural Assignments}

We now consider \emph{natural assignments},
where each memory cell can hold either a hash of an element that is currently in the dictionary,
or $\bot$, indicating a vacancy
(it is permitted to store multiple copies of the same element in memory.)
The next definition formalizes the notion of natural assignments:

\begin{definition}
    An assignment $\can$ is \emph{natural} if for every state $q\in Q$:
    \begin{enumerate}
        \item For every $0\leq \ell <\ncell$, either $\can\parens{q}[\ell] = \bot$ or $\can\parens{q}[\ell] \in q$.
        \item For every $v\in q$ there is an $0\leq \ell < \ncell$ such that $\can\parens{q}[\ell] = v$.
    \end{enumerate}
\end{definition}

A \emph{natural assignment with $k$-bits metadata} is an assignment $\can$
such that for every $q\in Q$, $\can(q) = (\can_{\mathit{nat}}(q), \mathit{meta}(q))$, where $\can_{\mathit{nat}}$ is a natural assignment and $\mathit{meta} : Q \to \set{0,1}^k$ is an arbitrary mapping of $Q$ to $k$ bits of information.

The proof of theorem~\ref{thm:natural-assig} 
uses Property~\ref{prop:lb-core} and Property~\ref{prop:llsc}.
First,
we show that a natural assignment using $\ncell < \ndomain$ memory cells cannot avoid satisfying
Property~\ref{prop:lb-core}, and for sufficiently large $\ndomain$, the same holds for a natural assignment with $O(\log \ndomain)$-bit metadata using $\ncell < \sqrt{\ndomain}$ memory cells.
A \emph{$(\ndomain,\nelem)$-dictionary} is a set that contains at most $\nelem$ elements from a domain of size $\ndomain$.

\begin{lemma}
\label{lem:natural-assig}
    Consider a natural assignment with $k$-bits metadata to the state space of an $(\ndomain,\nelem)$-dictionary, $k\geq 0$, $\nelem \geq 2$ and $\ncell < \ndomain$.
    If $k = 0$, the assignment satisfies Property~\ref{prop:lb-core}. Furthermore, for  $k = \lceil \log \ndomain \rceil + c$, where $c$ is a constant, $\nelem > 3$ and $\ncell < \sqrt{\ndomain}$, 
    there exists a sufficiently large $\ndomain$ for which
    the assignment satisfies Property~\ref{prop:lb-core}.
\end{lemma}

\begin{proof}
    We show the following claim: 

    \begin{claim}
    \label{clm:natural-assig1}
        If $\ncell < \ndomain / t$ such that $t < \frac{\sum_{\ell =0}^{\min(t+1,n)} {t+1 \choose \ell}}{2^k} - 2$, then $\can$ satisfies Property~\ref{prop:lb-core}.
    \end{claim}

    \begin{proof}
         Let $0\leq \ell < \ncell$ and assume that for $t+1$ distinct elements $v_1, \ldots, v_{t+1} \in \CalU$, for every $j\in [t+1]$ and all states $q,q'\in Q$ such that $v_j\in q$ and $v_j \notin q'$, 
    $\can\parens{q}[\ell] \neq \can\parens{q'}[\ell]$.

     We can partition the state set $Q$ according to the elements $v_1 ,\ldots, v_{t+1}$ as follows: The set $X_{i_1,\ldots i_{t+1}}$, where $i_j \in \set{0,1}$, $j \in [t+1]$,  and for at most $\min(t+1,\nelem)$
    locations $i_j = 1$, contains all the states $q\in Q$ such that $v_j \in q$ if and only if $i_j = 1$.
    The number of sets in this partition is $\sum_{l=0}^{\min(t+1,\nelem)} \binom{t+1}{l}$.
    Consider two sets $X_{i_1,\ldots i_{t+1}} \neq X_{i'_1,\ldots i'_{t+1}}$ and states $q \in X_{i_1,\ldots,i_{t+1}}$ and $q' \in X_{i'_1,\ldots,i'_{t+1}}$.
    As for some $j\in [t+1]$ $i_j \neq i'_j$, either $v_j\in q$ and $v_j \notin q'$ or $v_j\in q'$ and $v_j \notin q$, and $\can\parens{q}[\ell] \neq \can\parens{q'}[\ell]$.
    This implies that
    there are at least $\sum_{l=0}^{\min(t+1,\nelem)} \binom{t+1}{l}$ distinct values in the $\ell$-th entry of the canonical representations of all the states.

    Consider states that contain only elements of $\set{v_1, \ldots, v_{t+1}}$. For every such state $q$, since $\can$ is natural, $\can\parens{q}[\ell][0] \in \set{\bot, v_1, \ldots, v_{t+1}}$.
    Hence, the number of distinct values in the $\ell$-th memory cell for these states is at most $(t+2) 2^k$. However, by the lemma assumption, $(t+2) 2^k < \sum_{l=0}^{\min(t+1,\nelem)} \binom{t+1}{l}$, in contradiction.

    This shows that for every $0\leq \ell < \ncell$ there are at most $t$ elements $v \in \CalU$, where for all states $q,q'\in Q$ such that $v\in q$ and $v\notin q'$, $\can\parens{q}[\ell] \neq \can\parens{q'}[\ell]$. Thus, this inequality holds for at most $t \cdot \ncell < \ndomain$ elements over all memory cells and there is at least one element for which Property~\ref{prop:lb-core} holds.
    \end{proof}

    For $k = 0$, the inequality in Claim~\ref{clm:natural-assig1} holds for $t = 1$, which proves the first part of the lemma.
    Assume that $k =\lceil \log \ndomain \rceil + c$ for a constant $c$.
    For $\nelem\geq 4$ and $t\geq 3$, we get that
    $$\sum_{\ell =0}^{\min(t+1,\nelem)} {t+1 \choose \ell} > {t+1 \choose 4} > \frac{t^4}{24} .$$ 
     By setting $t = \sqrt{\ndomain}$, we get
     $$\frac{\sum_{\ell =0}^{\min(t+1,n)} {t+1 \choose \ell}}{2^k} - 2 >  \frac{t^4}{24 \cdot 2^k} = \frac{\ndomain^2}{24 \cdot 2^{\lceil \log \ndomain \rceil} \cdot 2^c} - 2.$$ Since $c$ is a constant and $\ndomain^2/2^{\lceil \log \ndomain \rceil} = \Theta(\ndomain)$, for a sufficiently large $\ndomain$, the following inequality holds: $\sqrt{\ndomain} < {\ndomain^2}/{24 \cdot 2^{\lceil \log \ndomain \rceil} \cdot 2^c} - 2$. Thus, by Claim~\ref{clm:natural-assig1},
     the assignment satisfies Property~\ref{prop:lb-core} for $\ncell < \sqrt{\ndomain} = \ndomain/t$. \qedhere
\end{proof}

Next, we show that if the dictionary is allowed to contain $\ncell$ elements,
then Property~\ref{prop:lb-core} implies Property~\ref{prop:llsc} for the same element.

\begin{lemma}
\label{lem:natural-assig-prop2}
    Consider a natural assignment with $k$-bits metadata $\can$ to the state space of a $(\ndomain,\ncell)$-dictionary.
    If Property~\ref{prop:lb-core} holds for element $v$ then Property~\ref{prop:llsc} also holds for element $v$.
\end{lemma}

\begin{proof}
    Let element $v$ be
    some element for which Property~\ref{prop:llsc} does not hold.
    This implies that Property~\ref{prop:llsc} also does not hold for the underlying natural assignment $\can_{\mathit{nat}}$.
    We show that this implies that Property~\ref{prop:lb-core} also does not hold for element $v$, which proves the lemma.
    Since Property~\ref{prop:llsc} does not hold, 
    for some memory cell $\ell$,
    either there is some value $x$
    such that for all $q \in Q$ where $v \in q$ we have $\can_{\mathit{nat}}\parens{q}[\ell] = x$,
    or there is some value $x$ such that for all $q \in Q$ where $v \notin q$ we have $\can_{\mathit{nat}}\parens{q}[\ell] = x$.

    Consider the first case, where for all states $q$ that contain $v$ we have $\can_{\mathit{nat}}\parens{q}[\ell] = x$.
    If $x = v$, then Property~\ref{prop:lb-core} is violated for $v$: in every state $q'$ such that $v \notin q'$ we have $\can_{\mathit{nat}}\parens{q'}[\ell] \neq v$ (by definition of natural assignments), and in every state $q$ such that $v \in q$ we have $\can_{\mathit{nat}}\parens{q}[\ell] = v$,
    so in every pair of states $q, q'$ such that $v \in q$ but $v \notin q'$ we have $\can_{\mathit{nat}}\parens{q}[\ell] \neq \can\parens{q'}[\ell]$, violating Property~\ref{prop:lb-core}.
    The only other possible value for $x$ is $\bot$,
    because there exists a state, $q = \set{v}$, where every memory cell stores either $v$ or $\bot$, and since $v \in q$, we have $\can_{\mathit{nat}}\parens{q}[\ell] = x$.
    But this also cannot be: let $q'$ be a set of $\ncell$ elements, including $v$. In $\can_{\mathit{nat}}\parens{q}$, every memory cell must store some distinct element, and in particular we cannot have $\can_{\mathit{nat}}\parens{q}[\ell] = \bot$.

    Now suppose that for all states $q$ that do \emph{not} contain $v$ we have $\can_{\mathit{nat}}\parens{q}[\ell] = x$ for some value $x$.
    Then we must have $x = \bot$, because in state $\emptyset$,
    which does not include $v$, all base elements store $\bot$.
    But on the other hand, there exists a state $q'$ containing $\ncell$ elements distinct from $v$ (since $\ncell < u$),
    and in $\can_{\mathit{nat}}\parens{q'}$, we cannot have any cells that store $\bot$.
    Therefore this case is impossible.
\end{proof}

Theorem~\ref{thm:waitfreelb}, which we prove below,
together  with
Lemma~\ref{lem:natural-assig-prop2} and Lemma~\ref{lem:natural-assig}
prove the following theorem:
\begin{theorem}[Theorem~\ref{thm:natural-assig}, stated formally.]
\label{thm:natural-assig-formal}
Given any collection of memory cells supporting any read-modify-write shared memory primitive operations
as well as $\LL/\VL/\SC$,
    there is a no wait-free SQHI implementation of a $(\ndomain, \ncell)$-dictionary from $\ncell < \ndomain$ memory cells whose canonical memory representation induces a natural assignment.
    Furthermore, for a sufficiently large $\ndomain$, the same holds for implementations from $3 < \ncell < \sqrt{\ndomain}$ memory cells,
    whose canonical memory representation induces a natural assignment with $k$-bits metadata such that $k = \lceil \log \ndomain \rceil + c$ for some constant $c$.
\end{theorem}

\subsection{Designing a History-Independent  Linear-Probing Hash Table}
\label{app:linearprob}

In this section, we consider a general linear probing hashing scheme with an arbitrary \emph{total priority ordering} for each cell $0\leq i < \ncell$, which can be compared using the binary relation $<_{p_i}$.
Using the linear probing scheme of~\cite{NaorTe01, BlellochGo07}, 
the hash function $h$ and the collection of priorities determine a unique canonical representation for each state in $Q$,
let $\can$ be this assignment.
(See the description for Robin Hood hashing in Section~\ref{sec:hash_overview}, replacing the specific priorities with general ones.)
For simplicity, we assume that the maximal number of elements that can be stored in the table is $\ncell - 1$, that is, it is not fully occupied.
In previous work, both the insert and delete operations of this scheme assume there is an empty cell in the table~\cite{BlellochGo07}.
Specifically for Robin Hood hashing, we showed that deletion also works when there is no empty cell, and the propagation of the operation ends before reaching the initial hash location of the deleted element (Lemma~\ref{lem:finishprop}).
In this section, we prove that the priority scheme in Robin Hood hashing is the only scheme that has the key property on which our algorithm is based.
First, we define \emph{linear} assignments, which are assignments where elements move at most one cell during an insertion or deletion.
We show that if an assignment is not linear, it cannot satisfy the property.
Finally, we show that linear assignments and Robin Hood hashing are equivalent.

\subsubsection{Linear Assignments}

Our algorithm requires that for every $q\in Q$ and $v \notin q$, we can determine that $v$ is not part of $q$ by only reading two consecutive cells in $\can(q)$.
We show that the underlying assignment $\can$ must be 
\emph{linear} to satisfy this property. In a linear assignment, 
for every $q\in Q$ and $v \notin q$, 
the position of every $v' \in q$ moves down at most one cell in 
$\can\parens{q\cup \set{v}}$ compared to $\can\parens{q}$.

The next lemma gives an equivalent definition to linear assignments that considers deletion instead of insertion.
\begin{lemma}
\label{lem:linassigeq}
    Consider an assignment $\can$, the following two statements are equivalent:
    \begin{enumerate}
        \item Let $q\in Q$, $q \neq \emptyset$, and $v \in q$, the position of every $v' \in q \setminus \set{v}$ moves up at most one cell in $\can\parens{q\setminus \set{v}}$ compared to $\can\parens{q}$. 
        \item Let $q\in Q$ and $v \notin q$, the position of every $v' \in q$ moves down at most one cell in $\can\parens{q\cup \set{v}}$ compared to $\can\parens{q}$. 
    \end{enumerate}
\end{lemma}

\begin{proof}
    Assume that for every $q\in Q$, $q \neq \emptyset$, and $v \in q$, the position of every $v' \in q \setminus \set{v}$ has moved up at most one cell in $\can\parens{q}$ compared to $\can\parens{q\setminus \set{v}}$. 
    Let $q\in Q$ and $v \notin q$, by the assumption the position of every $v' \in q$ has moved up at most one cell in $\can\parens{q}$ compared to $\can\parens{q\cup \set{v}}$. This implies that every $v' \in q$ has moved down at most one cell in $\can\parens{q\cup \set{v}}$ compared to $\can\parens{q}$.

     Assume that for every $q\in Q$ and $v \notin q$, the position of every $v' \in q$ has moved down at most one cell in $\can\parens{q\cup \set{v}}$ compared to $\can\parens{q}$.
     Let $q\in Q$, $q \neq \emptyset$, and $v \in q$, by the assumption the position of every $v' \in Q$ has moved down at most one cell in $\can\parens{q}$ compared to $\can\parens{q\setminus \set{v}}$. This implies that every $v' \in q$ has moved up at most one cell in $\can\parens{q\setminus \set{v}}$ compared to $\can\parens{q}$.
\end{proof}

We prove that if $\can$ is not linear, then for some $q\in Q$ such that $v\notin q$, we cannot determine that $v$ is not in $q$ only by reading two consecutive cells in $\can(q)$. 
This is because the values of any two consecutive cells in $\can(q)$ are the same as those in $\can(q')$ for some state $q'$ where $v\in q'$.
Intuitively, if a deletion causes an element to move two cells up compared to its initial location, an assignment with a lookahead only to the next cell will not detect this deletion in the cells in between this change (except for the last cell before the initial location of the element).
This allows us to ``hide'' this element in different parts of the table, making a part of the table look identical to a table without this element.
Let $\dist(j,i) = \parens{i - j} \bmod \ncell$, $0\leq i,j < \ncell$, be the distance of position $i$ from starting position $j$,
and let $\pos(q,v)$, $v\in q$, be the position of element $v$ in $\can\parens{q}$.

\begin{lemma}
\label{lem:nonlinearnoimpl}
    If the assignment $\can$ is not linear, 
    then there is state $q\in Q$ and element $v \notin q$ such that for every $0 \leq i < \ncell$ there is $q'\in Q$ such that $v\in q'$ and $\can(q)[i] = \can(q')[i]$ and $\can(q)[i + 1] = \can(q')[i + 1]$.
\end{lemma}

\begin{proof}
    Since $\can$ is not linear, there are $q\in Q$, $q \neq \emptyset$ and $v \in q$, and element $v' \in q \setminus \set{v}$ that moves up at least two cells in $\can\parens{q}$ compared to $\can\parens{q\setminus \set{v}}$.
    Deleting element $v$ from $q$ results in a chain of elements that move up along the table, $v = v_0, \ldots, v_k$, $k\geq 1$.
    Element $v_{j-1}$ is replaced with element $v_j$, $1\leq j\leq k$, and the value $v_{k+1} = \bot$ is placed in the cell where the last element $v_k$ was.
    By the assumption, element $v'$ is part of the chain, hence, $v_i = v'$ for $1\leq i \leq k$ and 
    $\dist(\pos(q,v_{i-1}), \pos(q,v_{i})) > 1$.
    For any element $u\in q$ such that $\pos(q,v_{j-1})< \pos(q,u) <\pos(q,v_{j})$, i.e., $u$ is placed in between $v_{j-1}$ and $v_j$ in $\can(q)$, it must hold that $\rank(u,\pos(q,v_{j-1})) > \rank(u, \pos(q,u))$, hence, such element can only possibly move because of a deletion of an element between these positions. 

    Since $|q|\leq \ncell-1$, there is an empty cell in the $\can(q)$. This implies the deletion of element $v_j$ from $q$, $1\leq j\leq k$ results in the deletion chain $v_j, \ldots, v_k$.
    Let $u$ be the element in position $\pos(q,v_{i}) + 1$ in $\can(q)$ and consider the deletion of this element from $q$. In $\can\parens{q\setminus\set{u}}$, the position of $v_i$ is $\pos(q,v_{i}) + 1$ and the position of $v_{j+1}$ for $i\leq j \leq k$ is $\pos(q,v_{j})$. The position of all other elements remain the same in $\can\parens{q\setminus \set{u}}$ as in $\can\parens{q}$.

    Consider the deletion of $v_i$ from $q$. In $\can\parens{q\setminus\set{v_i}}$, 
    the position of $v_{j+1}$ is $\pos(q,v_{j})$, $i\leq j \leq k$, and the position of all other elements remain the same in $\can\parens{q\setminus \set{v_i}}$ as in $\can\parens{q}$.
    Let $\canla$ be the assignment $\can$, where each cell also includes the value of the next cell. 
    We divide the table into three parts, in each part $\canla\parens{q\setminus\set{v_i}}$ is equal to $\canla\parens{q}$, $\canla\parens{q\setminus\set{v_0}}$ or $\canla\parens{q\setminus\set{u}}$ (see Figure~\ref{fig:nonlinearnoimpl}):
    \begin{enumerate}
        \item Between $\pos(q,v_{0})$ and $\pos(q,v_{i-1}) + 1$, $\can\parens{q\setminus\set{v_i}}$ is equal to $\can\parens{q}$, which implies that between $\pos(q,v_{0})$ and $\pos(q,v_{i-1})$, $\canla\parens{q\setminus\set{v_i}}$ is equal to $\canla\parens{q}$.

        \item Between $\pos(q,v_{i-1}) + 1$ and $\pos(q,v_{i})$, $\can\parens{q\setminus\set{v_i}}$ is equal to $\can\parens{q\setminus\set{v_0}}$, which implies that between $\pos(q,v_{i-1}) + 1$ and $\pos(q,v_{i-1}) - 1$, $\canla\parens{q\setminus\set{v_i}}$ is equal to $\canla\parens{q\setminus\set{v_0}}$.

        \item Between $\pos(q,v_{i})$ and $\pos(q,v_{0})$, $\can\parens{q\setminus\set{v_i}}$ is equal to $\can\parens{q\setminus\set{u}}$, which implies that between $\pos(q,v_{i})$ and $\pos(q,v_{0}) - 1$, $\canla\parens{q\setminus\set{v_i}}$ is equal to $\canla\parens{q\setminus\set{u}}$.
    \end{enumerate}
     Since $v_i \notin q\setminus\set{v_i}$ and $v_i \in q, q\setminus\set{v_0}, q\setminus\set{u}$, this concludes the proof.
\end{proof}

\begin{table}
\centering
\begin{tabular}{c c c c}
$\canla\parens{q\setminus\set{v_i}}$ & $\canla\parens{q}$ & $\canla\parens{q\setminus\set{u}}$ & $\canla\parens{q\setminus\set{v_0}}$ \\
\begin{tabular}{|c|c|}
    \hline
    \textcolor{Red}{$v_0$} & \textcolor{Red}{$\ldots$} \\ \hline
    \textcolor{Red}{$\ldots$} & \textcolor{Red}{$v_1$} \\ \hline
    \textcolor{Red}{$v_1$} & \textcolor{Red}{$\ldots$}\\ \hline
    \textcolor{Red}{$\ldots$} & \textcolor{Red}{$v_{i-1}$} \\ \hline  
    \textcolor{Red}{$v_{i-1}$} & \textcolor{Red}{$\ldots$}\\ \hline
    \textcolor{blue}{$\ldots$} & \textcolor{blue}{$u$}\\ \hline   
    \textcolor{blue}{$u$} & \textcolor{blue}{$v_{i+1}$}\\ \hline
    \textcolor{Green}{$v_{i+1}$} & \textcolor{Green}{$\ldots$} \\ \hline   
    \textcolor{Green}{$\ldots$} & \textcolor{Green}{$\bot$} \\ \hline   
    \textcolor{Green}{$\bot$} & \textcolor{Green}{$\ldots$} \\ \hline   
    \textcolor{Green}{$\ldots$} & \textcolor{Green}{$v_0$} \\ \hline   
\end{tabular}
&
\begin{tabular}{|c|c|}
    \hline
    \textcolor{Red}{$v_0$} & \textcolor{Red}{$\ldots$} \\ \hline
    \textcolor{Red}{$\ldots$} & \textcolor{Red}{$v_1$} \\ \hline
    \textcolor{Red}{$v_1$} & \textcolor{Red}{$\ldots$} \\ \hline
    \textcolor{Red}{$\ldots$} & \textcolor{Red}{$v_{i-1}$} \\ \hline   
    \textcolor{Red}{$v_{i-1}$} & \textcolor{Red}{$\ldots$}\\ \hline
    $\ldots$ & \\ \hline   
    $u$ & \\ \hline
    $v_i$ & \\ \hline   
    $\ldots$ & \\ \hline   
    $v_k$ & \\ \hline   
    $\ldots$ & \\ \hline   
\end{tabular}
&
\begin{tabular}{|c|c|}
    \hline
    $v_0$ & \\ \hline
    $\ldots$ & \\ \hline
    $v_1$ & \\ \hline
    $\ldots$ & \\ \hline   
    $v_{i-1}$ & \\ \hline
    $\ldots$ & \\ \hline   
    $v_i$ & \\ \hline
    \textcolor{Green}{$v_{i+1}$} & \textcolor{Green}{$\ldots$} \\ \hline   
    \textcolor{Green}{$\ldots$} & \textcolor{Green}{$\bot$} \\ \hline   
    \textcolor{Green}{$\bot$} & \textcolor{Green}{$\ldots$} \\ \hline   
    \textcolor{Green}{$\ldots$} & \textcolor{Green}{$v_0$} \\ \hline   
\end{tabular}
&
\begin{tabular}{|c|c|}
    \hline
    $v_1$ & \\ \hline
    $\ldots$ & \\ \hline
    $v_2$ & \\ \hline
    $\ldots$ & \\ \hline   
    $v_{i}$ & \\ \hline
    \textcolor{blue}{$\ldots$} & \textcolor{blue}{$u$} \\ \hline   
    \textcolor{blue}{$u$} & \textcolor{blue}{$v_{i+1}$} \\ \hline
    $v_{i+1}$ & \\ \hline   
    $\ldots$ & \\ \hline   
    $\bot$ & \\ \hline   
    $\ldots$ & \\ \hline   
\end{tabular}
\end{tabular}
\caption{Illustrating the proof of Lemma~\ref{lem:nonlinearnoimpl}}
\label{fig:nonlinearnoimpl}
\end{table}

\subsubsection{Linear Assignments are Equivalent to Robin Hood Hashing}
\label{app:robinhood}


For a given hash function $h$ consider the next collection of priorities:

\begin{definition}[Age-rules~\cite{NaorTe01}] 
\label{def:age-rule}
    For every $0\leq i < \ncell$ and elements $v_1,v_2\in \CalU$,
    if $\rank(v_1,i) > \rank(v_2,i)$ then $v_1 >_{p_i} v_2$, and if $\rank(v_1,i) < \rank(v_2,i)$ then $v_1 <_{p_i} v_2$, otherwise $\rank(v_1,i) = \rank(v_2,i)$ (i.e., $h(v_1) = h(v_2)$), the priority is determined according to some fixed total order $<$ on the elements in $\CalU$.
\end{definition}

It is easy to see that using age-rule is equivalent to \emph{Robin Hood hashing}~\cite{CelisLaMu85}, with tie-breaking using some given total order, where the priority is given to elements that are further away from their initial hash location.
Next, we prove that linear assignments \emph{must} use age rules.

The next lemma shows that age rules yield a linear assignment:

\begin{restatable}{lemma}{ageruleslin}
\label{lem:agerules->lin}
    An assignment $\can$ induced by a hash function $h$ and age-rules is linear.
\end{restatable}

\begin{proof}
    Assume the priority functions are determined by age rules and $\can$ is not linear.
    Then there are $q\in Q$, $q \neq \emptyset$ and $v \in q$, and element $v' \in q \setminus \set{v}$ that moves up at least two cells in $\can\parens{q}$ compared to $\can\parens{q\setminus \set{v}}$.
    Deleting element $v$ from $q$ results in a chain of elements that move up along the table, $v = v_0, \ldots, v_k$, $k\geq 1$.
    By the assumption, element $v'$ is part of the chain, hence, $v_i = v'$ for $1\leq i \leq k$ and $\dist(\pos(q,v_{i-1}), \pos(q,v_{i})) > 1$.
    Let $u$ be the element in $\pos(q,v_{i-1}) - 1$. Since $u$ is not part of the deletion chain and is positioned between $v_{i-1}$ and $v_i$, $\rank(u, \pos(q,v_{i-1})) > \rank(u, \pos(q,v_{i-1}) - 1)$. This can only happen if $h(u) = \pos(q,v_{i-1}) - 1 = \pos(q,u)$.
    In addition, since the value $v_{i-1}$ is replaced with $v_i$ in the deletion chain, $\rank(v_i, \pos(q,v_{i-1})) < \rank(v_i,\pos(q,v_{i}))$.
    In the deletion of $v_i$ from $q$, $u$ stays in the same position in $\can(q\setminus\set{v_i})$ as in $\can(q)$. Since
    $\rank(v_i, \pos(q,v_{i-1}) - 1) < \rank(v_i, \pos(q,v_{i}))$, in the insertion of $v_i$ to $q\setminus\set{v_i}$ the insertion process examines $u$ and finds that $u >_{\pos(q,v_{i-1}) - 1} v_i$. Therefore, by the choice of priority functions, 
    $\rank(u,\pos(q,v_{i-1}) - 1) > \rank(v_i,\pos(q,v_{i-1}) - 1)$.
    However, we get that  $0 = \rank(u,\pos(q,v_{i-1}) - 1) > \rank(v_i,\pos(q,v_{i-1}) - 1) > \rank(v_i,\pos(q,v_{i-1})) \geq 0$, in contradiction.
\end{proof}

Next, we show that a choice of priorities other than 
age-rules yields a non-linear assignment:

\begin{restatable}{lemma}{linagerules}
\label{lem:lin->agerules}
    Consider a linear assignment $\can$,
    for any $q\in Q$ and $u,v \neq \bot$, if $u$ is in position $i$ in $\can(q)$ and $v$ is in position $i+1 \neq h(v)$ in $\can(q)$, then $\rank(u,i) \geq \rank(v,i)$.
\end{restatable}

\begin{proof}
    Assume there is a $q\in Q$ and $u,v \neq \bot$ such that $u$ is in position $i$ in $\can(q)$, $v$ is in position $i+1 \neq h(v)$ in $\can(q)$ and $\rank(u,i) < \rank(v,i)$. 
    All positions between $h(u)$ and $i$ must be occupied by some element. Similarly, all positions between $h(v)$ and $i+1$ must be occupied by some element. Let $q'$ be the state resulting from deleting all elements between $h(u)$ and $i - 1$ in $\can(q)$. This results in $u$ being in position $h(u)$ in $\can(q')$ and $v$ in position $h(u) + 1$ in $\can(q')$ since $v$ is not in its hash location in $\can(q)$ and $\rank(v,h(u) + 1) < \rank(v,i)$. 
    Let $q''$ be the state resulting from deleting the element in position $h(u)+1$ in $\can(q')$.
    Since $u$ is in its hashing position in $q'$ it will not go up a cell in $\can(q'')$. It must hold that $\rank(v, h(u)-1) < \rank(v, h(u) + 1)$ and $v$ prefers to be two positions above its current position in $\can(q')$. However, this contradicts the linearity of the assignment.  
\end{proof}





%% file: related.tex
\section{Related Work}
\label{sec:related}

\paragraph{History-independent data structures.}
Efficient sequential history-independent data structure constructions include
fast HI constructions for cuckoo hash tables~\cite{NaorSeWi08}, 
linear-probing hash tables~\cite{BlellochGo07,GoodrichKoMi17}, 
other hash tables~\cite{BlellochGo07,NaorTe01}, 
trees~\cite{Micciancio97,AcarBlHa04},
memory allocators~\cite{NaorTe01,GoodrichKoMi17}, 
write-once memories~\cite{MoranNaSe07},
priority queues~\cite{BuchbinderPe03}, 
union-find data structures~\cite{NaorTe01}, 
external-memory dictionaries~\cite{Golovin08,Golovin09,Golovin10,BenderBeJo16}, 
file systems~\cite{BajajSi13b, BajajSi13a,BajajChSi15, RocheAvCh15},
cache-oblivious dictionaries~\cite{BenderBeJo16}, 
order-maintenance data structures~\cite{BlellochGo07},
packed-memory arrays/list-labeling data structures~\cite{BenderBeJo16,BenderCoFa22},
and geometric data structures~\cite{Tzouramanis12}.
Given the strong connection between history independence and unique representability~\cite{HartlineHoMo02,HartlineHoMo05}, 
some earlier data structures 
can also be made history independent, 
including 
hashing variants~\cite{AmbleKn74,CelisLaMu85}, 
skip lists~\cite{Pugh90}, treaps~\cite{AragonSe89}, 
and other less well-known deterministic data structures 
with canonical representations~\cite{SundarTa90,AnderssonOt91,AnderssonOt95,PughTe89,Snyder77}.
There is work on the time-space tradeoff for strongly history-independent hash tables, 
see e.g.,~\cite{Kus23,LLYZ24, LLYZ23} and references therein.

The algorithmic work on history independence has 
found its way into systems.
There are now voting machines~\cite{BethencourtBoWa07}, 
file systems~\cite{BajajSi13b,BajajChSi16,BajajChSi15}, 
databases~\cite{BajajSi13a,PoddarBoPo16,RocheAvCh16},
and other storage systems~\cite{ChenSi15} that support 
history independence as an essential feature.

\paragraph{History independence for concurrent data structures.}
To our knowledge, the only work to study history independence in a fully concurrent setting is~\cite{AtBeFaCoOsSc24},
which investigated several possible definitions for history independence in a concurrent system that might never quiesce.
Several lower bounds and impossibility results are proven in~\cite{AtBeFaCoOsSc24}.
In particular, it is shown that if the observer is allowed to examine the memory at any point in the execution,
then in any concurrent implementation that is obstruction free,
for any two logical states $q, q'$ of the ADT
such that some operation can cause a transition from $q$ to $q'$,
the memory representations of $q, q'$
in the implementation must differ by exactly one base object.
This rules out open-addressing hash tables, unless their size is $|\CalU|$,
because it does not allow us to move elements around.

\paragraph{Sequential hash tables.}
There are several popular design paradigms for hash tables, including
probing, both linear and other probing patterns (e.g.,~\cite{AmbleKn74,Knuth63,CelisLaMu85,HopgoodDa72,Maurer68});
cuckoo hashing~\cite{PaghRodlerCuckooHash};
and chained hashing~\cite{Kruse84}.
These can all be made history independent~\cite{NaorSeWi08,NaorTe01,BlellochGo07,GoodrichKoMi17},
typically by imposing a canonical order on the elements
in the hash table.
This involves moving elements around,
which can be challenging for concurrent implementations.
Notably, for sequential history-independent 
hash tables,
history independence does not come at the cost of increasing the width of memory cells---it suffices to
store one element per cell;
our work shows that in the concurrent setting, storing one element per cell is not enough.

\paragraph{Concurrent hash tables.}
There is a wealth of literature on 
non-history-independent
concurrent hash tables;
we discuss here only the
work directly related to ours,
which includes concurrent implementations of linear probing hash tables in general, and Robin Hood hashing in particular.
We note that while~\cite{AtBeFaCoOsSc24} gives a universal wait-free history-independent construction that can be used to implement any abstract data type, including a dictionary,
the construction does not allow for true concurrency:
processes compete to perform operations,
with only one operation going through at a time.
In addition,
the universal construction of~\cite{AtBeFaCoOsSc24} uses memory cells whose width is linear in the state of the data type (in the case of a dictionary, the size of the set that it represents),
and in the number of processes.

There are several implementations of concurrent linear-probing hash tables, e.g.,~\cite{SSBHW10,GGH05}.
It is easier to implement a concurrent linear-probing hash table if one is not concerned with history independence: 
since the elements do not need to be stored in a canonical order, we can simply place newly-inserted elements in the first empty cell we find, and we do not need to move elements around. Tombstones can be used to mark deleted elements, thus avoiding the need to move elements once inserted.  

As for Robin Hood hashing, which gained popularity due to its use in the Rust programming language,
there are not many concurrent implementations.
Kelly et al.~\cite{KellyPM2018} present a lock-free concurrent Robin Hood hash table 
that is not history independent.
Their design is based on a $k$-CAS primitive,
which is then replaced with an optimized implementation 
from CAS~\cite{ArbelravivBrown2017}.
This implementation uses standard synchronization techniques that require large memory cells:
timestamps, operation descriptors, and so on.
Bolt~\cite{Kahssay2021} is a concurrent resizable Robin Hood hash table,
with a lock-free fast path and a lock-based slow path.
This implementation is not lock-free, although it does avoid the use of locks in lightly-contended cases (the fast path).


Shun and Blelloch~\cite{ShunBlellochSPAA14} present a \emph{phase-concurrent} deterministic 
hash table based on linear probing, which supports one type of operation within a synchronized phase. Different types of operations (inserts, deletes and lookups)
may not overlap.
The hash table of~\cite{ShunBlellochSPAA14} is based
on the same history-independent scheme used in~\cite{NaorTe01, BlellochGo07}.
In this scheme, each memory cell stores a single element;
this does not contradict our impossibility result, because none of these implementations support overlapping lookups and operations that modify the hash table (inserts and/or deletes).




%% file: main.bbl
\begin{thebibliography}{10}

\bibitem{AcarBlHa04}
Umut~A Acar, Guy~E Blelloch, Robert Harper, Jorge~L Vittes, and Shan Leung~Maverick Woo.
\newblock Dynamizing static algorithms, with applications to dynamic trees and history independence.
\newblock In {\em Proc.\ of the 15th Annual ACM-SIAM Symposium on Discrete Algorithms (SODA)}, pages 531--540, 2004.

\bibitem{AmbleKn74}
Ole Amble and Donald~Ervin Knuth.
\newblock Ordered hash tables.
\newblock {\em The Computer Journal}, 17(2):135--142, January 1974.

\bibitem{AnderssonOt91}
Arne Andersson and Thomas Ottmann.
\newblock Faster uniquely represented dictionaries.
\newblock In {\em Proc.\ of the 32nd Annual IEEE Symposium on Foundations of Computer Science (FOCS)}, pages 642--649, 1991.

\bibitem{AnderssonOt95}
Arne Andersson and Thomas Ottmann.
\newblock New tight bounds on uniquely represented dictionaries.
\newblock {\em SIAM Journal on Computing}, 24(5):1091--1103, 1995.

\bibitem{AragonSe89}
Cecilia~R Aragon and Raimund~G Seidel.
\newblock Randomized search trees.
\newblock In {\em Proc.\ of the 30th Annual IEEE Symposium on Foundations of Computer Science (FOCS)}, pages 540--545, 1989.

\bibitem{ArbelravivBrown2017}
Maya Arbel-Raviv and Trevor Brown.
\newblock {Reuse, Don't Recycle: Transforming Lock-Free Algorithms That Throw Away Descriptors}.
\newblock In {\em 31st International Symposium on Distributed Computing (DISC)}, Leibniz International Proceedings in Informatics (LIPIcs), pages 4:1--4:16, 2017.

\bibitem{AtBeFaCoOsSc24}
Hagit Attiya, Michael~A. Bender, Mart\'{\i}n Farach-Colton, Rotem Oshman, and Noa Schiller.
\newblock History-independent concurrent objects.
\newblock In {\em Proceedings of the 43rd ACM Symposium on Principles of Distributed Computing (PODC)}, page 14–24, 2024.

\bibitem{BajajChSi15}
Sumeet Bajaj, Anrin Chakraborti, and Radu Sion.
\newblock The foundations of history independence.
\newblock {\em arXiv preprint arXiv:1501.06508}, 2015.

\bibitem{BajajChSi16}
Sumeet Bajaj, Anrin Chakraborti, and Radu Sion.
\newblock Practical foundations of history independence.
\newblock {\em {IEEE} Trans. Inf. Forensics Secur.}, 11(2):303--312, 2016.

\bibitem{BajajSi13b}
Sumeet Bajaj and Radu Sion.
\newblock {HIFS}: History independence for file systems.
\newblock In {\em Proc.\ of the ACM SIGSAC Conference on Computer \& Communications Security (CCS)}, pages 1285--1296, 2013.

\bibitem{BajajSi13a}
Sumit Bajaj and Radu Sion.
\newblock Ficklebase: Looking into the future to erase the past.
\newblock In {\em Proc.\ of the 29th IEEE International Conference on Data Engineering (ICDE)}, pages 86--97, 2013.

\bibitem{BayerS1977}
Rudolf Bayer and Mario Schkolnick.
\newblock Concurrency of operations on b-trees.
\newblock {\em Acta informatica}, 9:1--21, 1977.

\bibitem{BenderBeJo16}
Michael~A. Bender, Jon Berry, Rob Johnson, Thomas~M. Kroeger, Samuel McCauley, Cynthia~A. Phillips, Bertrand Simon, Shikha Singh, and David Zage.
\newblock Anti-persistence on persistent storage: History-independent sparse tables and dictionaries.
\newblock In {\em Proc.\ 35th ACM Symposium on Principles of Database Systems (PODS)}, pages 289--302, June 2016.

\bibitem{BenderCoFa22}
Michael~A. Bender, Alex Conway, Mart{\'\i}n Farach-Colton, Hanna Koml{\'o}s, William Kuszmaul, and Nicole Wein.
\newblock Online list labeling: Breaking the $\log^2 n$ barrier.
\newblock In {\em Proc.\ 63rd IEEE Annual Symposium on Foundations of Computer Science (FOCS)}, November 2022.

\bibitem{BethencourtBoWa07}
John Bethencourt, Dan Boneh, and Brent Waters.
\newblock Cryptographic methods for storing ballots on a voting machine.
\newblock In {\em Proc.\ of the 14th Network and Distributed System Security Symposium (NDSS)}, 2007.

\bibitem{BlellochGo07}
Guy~E Blelloch and Daniel Golovin.
\newblock Strongly history-independent hashing with applications.
\newblock In {\em Proc.\ of the 48th Annual IEEE Symposium on Foundations of Computer Science (FOCS)}, pages 272--282, 2007.

\bibitem{BuchbinderPe03}
Niv Buchbinder and Erez Petrank.
\newblock Lower and upper bounds on obtaining history independence.
\newblock In {\em Advances in Cryptology}, pages 445--462, 2003.

\bibitem{CelisLaMu85}
Pedro Celis, Per{-}{\AA}ke Larson, and J.~Ian Munro.
\newblock {Robin Hood} hashing (preliminary report).
\newblock In {\em 26th Annual Symposium on Foundations of Computer Science (FOCS'85)}, pages 281--288, Portland, Oregon, USA, 21--23~October 1985.

\bibitem{ChenSi15}
Bo~Chen and Radu Sion.
\newblock Hiflash: {A} history independent flash device.
\newblock {\em CoRR}, abs/1511.05180, 2015.

\bibitem{Click2007}
Cliff Click.
\newblock A lock-free wait-free hash table.
\newblock \url{http://www.stanford.edu/class/ee380/Abstracts/070221}, 2007.
\newblock Accessed on 2024-07-04.

\bibitem{GGH05}
H.~Gao, J.~F. Groote, and W.~H. Hesselink.
\newblock Lock-free dynamic hash tables with open addressing.
\newblock {\em Distrib. Comput.}, 18(1):21–42, 2005.

\bibitem{Golovin08}
Daniel Golovin.
\newblock {\em Uniquely Represented Data Structures with Applications to Privacy}.
\newblock PhD thesis, Carnegie Mellon University, Pittsburgh, PA, 2008, 2008.

\bibitem{Golovin09}
Daniel Golovin.
\newblock {B}-treaps: A uniquely represented alternative to {B}-trees.
\newblock In {\em Proc.\ of the 36th Annual International Colloquium on Automata, Languages, and Programming (ICALP)}, pages 487--499. Springer Berlin Heidelberg, 2009.

\bibitem{Golovin10}
Daniel Golovin.
\newblock The {B}-skip-list: A simpler uniquely represented alternative to {B}-trees.
\newblock {\em arXiv preprint arXiv:1005.0662}, 2010.

\bibitem{GKMT16}
Michael~T. Goodrich, Evgenios~M. Kornaropoulos, Michael Mitzenmacher, and Roberto Tamassia.
\newblock More practical and secure history-independent hash tables.
\newblock In {\em {ESORICS} {(2)}}, volume 9879 of {\em Lecture Notes in Computer Science}, pages 20--38. Springer, 2016.

\bibitem{GoodrichKoMi17}
Michael~T. Goodrich, Evgenios~M. Kornaropoulos, Michael Mitzenmacher, and Roberto Tamassia.
\newblock Auditable data structures.
\newblock In {\em Proc.\ {IEEE} European Symposium on Security and Privacy (EuroS{\&}P)}, pages 285--300, April 2017.

\bibitem{HartlineHoMo02}
Jason~D. Hartline, Edwin~S. Hong, Alexander~E. Mohr, William~R. Pentney, and Emily Rocke.
\newblock Characterizing history independent data structures.
\newblock In {\em Proceedings of the Algorithms and Computation, 13th International Symposium (ISAAC)}, pages 229--240, November 2002.

\bibitem{HartlineHoMo05}
Jason~D Hartline, Edwin~S Hong, Alexander~E Mohr, William~R Pentney, and Emily~C Rocke.
\newblock Characterizing history independent data structures.
\newblock {\em Algorithmica}, 42(1):57--74, 2005.

\bibitem{HerlihyShavitBook}
Maurice Herlihy, Nir Shavit, Victor Luchangco, and Michael Spear.
\newblock {\em The Art of Multiprocessor Programming, Second Edition}.
\newblock Elsevier, 2020.

\bibitem{HerlihyShavitTzafrir08}
Maurice Herlihy, Nir Shavit, and Moran Tzafrir.
\newblock Hopscotch hashing.
\newblock In {\em Distributed Computing: 22nd International Symposium (DISC)}, pages 350--364, 2008.

\bibitem{HerlihyWing90}
Maurice~P. Herlihy and Jeannette~M. Wing.
\newblock Linearizability: a correctness condition for concurrent objects.
\newblock {\em ACM Transactions on Programming Languages and Systems}, 12(3):463–492, jul 1990.

\bibitem{HopgoodDa72}
F.~R.~A. Hopgood and J.~Davenport.
\newblock The quadratic hash method when the table size is a power of 2.
\newblock {\em The Computer Journal}, 15(4):314--315, 1972.

\bibitem{Kahssay2021}
Endrias Kahssay.
\newblock A fast concurrent and resizable robin-hood hash table.
\newblock Master's thesis, Massachusetts Institute of Technology, 2021.

\bibitem{KellyPM2018}
Robert Kelly, Barak~A. Pearlmutter, and Phil Maguire.
\newblock {Concurrent Robin Hood Hashing}.
\newblock In {\em 22nd International Conference on Principles of Distributed Systems (OPODIS)}, Leibniz International Proceedings in Informatics (LIPIcs), pages 10:1--10:16, 2019.

\bibitem{Knuth63}
Don Knuth.
\newblock Notes on ``open'' addressing, 1963.

\bibitem{Kruse84}
Robert~L. Kruse.
\newblock {\em Data Structures and Program Design}.
\newblock Prentice-Hall Inc, Englewood Cliffs, New Jersey, USA, 1984.

\bibitem{Kus23}
William Kuszmaul.
\newblock Strongly history-independent storage allocation: New upper and lower bounds.
\newblock In {\em 2023 IEEE 64th Annual Symposium on Foundations of Computer Science (FOCS)}, pages 1822--1841, Los Alamitos, CA, USA, 2023. IEEE Computer Society.

\bibitem{LLYZ23}
Tianxiao Li, Jingxun Liang, Huacheng Yu, and Renfei Zhou.
\newblock { Tight Cell-Probe Lower Bounds for Dynamic Succinct Dictionaries }.
\newblock In {\em 2023 IEEE 64th Annual Symposium on Foundations of Computer Science (FOCS)}, pages 1842--1862, Los Alamitos, CA, USA, November 2023. IEEE Computer Society.

\bibitem{LLYZ24}
Tianxiao Li, Jingxun Liang, Huacheng Yu, and Renfei Zhou.
\newblock Dynamic dictionary with subconstant wasted bits per key.
\newblock In {\em Proceedings of the 2024 Annual ACM-SIAM Symposium on Discrete Algorithms (SODA)}, pages 171--207. Society for Industrial and Applied Mathematics, 2024.

\bibitem{Maurer68}
W.~D. Maurer.
\newblock An improved hash code for scatter storage.
\newblock {\em Communications of the ACM}, 11(1):35--38, 1968.

\bibitem{Micciancio97}
Daniele Micciancio.
\newblock Oblivious data structures: applications to cryptography.
\newblock In {\em Proc.\ of the 29th Annual ACM Symposium on Theory of Computing (STOC)}, pages 456--464, 1997.

\bibitem{Michael02}
Maged~M. Michael.
\newblock High performance dynamic lock-free hash tables and list-based sets.
\newblock In {\em Proceedings of the Fourteenth Annual ACM Symposium on Parallel Algorithms and Architectures (SPAA)}, page 73–82, 2002.

\bibitem{MoranNaSe07}
Tal Moran, Moni Naor, and Gil Segev.
\newblock Deterministic history-independent strategies for storing information on write-once memories.
\newblock In {\em Proc.\ 34th International Colloquium on Automata, Languages and Programming (ICALP)}, 2007.

\bibitem{NaorSeWi08}
Moni Naor, Gil Segev, and Udi Wieder.
\newblock History-independent cuckoo hashing.
\newblock In {\em Proc.\ of the 35th International Colloquium on Automata, Languages and Programming (ICALP)}, pages 631--642. Springer, 2008.

\bibitem{NaorTe01}
Moni Naor and Vanessa Teague.
\newblock Anti-persistence: history independent data structures.
\newblock In {\em Proc.\ of the 33rd Annual ACM Symposium on Theory of Computing (STOC)}, pages 492--501, 2001.

\bibitem{PaghRodlerCuckooHash}
Rasmus Pagh and Flemming~Friche Rodler.
\newblock Cuckoo hashing.
\newblock {\em Journal of Algorithms}, 51(2):122 -- 144, 2004.

\bibitem{PoddarBoPo16}
Rishabh Poddar, Tobias Boelter, and Raluca~Ada Popa.
\newblock Arx: {A} strongly encrypted database system.
\newblock {\em {IACR} Cryptol. ePrint Arch.}, page 591, 2016.

\bibitem{Pugh90}
William Pugh.
\newblock Skip lists: a probabilistic alternative to balanced trees.
\newblock {\em Communications of the ACM}, 33(6):668--676, 1990.

\bibitem{PughTe89}
William Pugh and Tim Teitelbaum.
\newblock Incremental computation via function caching.
\newblock In {\em Proc.\ of the 16th ACM SIGPLAN-SIGACT Symposium on Principles of Programming Languages (POPL)}, pages 315--328, 1989.

\bibitem{PurcellH05}
Chris Purcell and Tim Harris.
\newblock Non-blocking hashtables with open addressing.
\newblock In {\em Distributed Computing, 19th International Conference, {DISC}}, volume 3724 of {\em Lecture Notes in Computer Science}, pages 108--121. Springer, 2005.

\bibitem{RocheAvCh15}
Daniel~S Roche, Adam~J Aviv, and Seung~Geol Choi.
\newblock Oblivious secure deletion with bounded history independence.
\newblock {\em arXiv preprint arXiv:1505.07391}, 2015.

\bibitem{RocheAvCh16}
Daniel~S. Roche, Adam~J. Aviv, and Seung~Geol Choi.
\newblock A practical oblivious map data structure with secure deletion and history independence.
\newblock In {\em {IEEE} Symposium on Security and Privacy (SP)}, pages 178--197. {IEEE} Computer Society, 2016.

\bibitem{ShalevShavit03}
Ori Shalev and Nir Shavit.
\newblock Split-ordered lists: lock-free extensible hash tables.
\newblock In {\em Proceedings of the Twenty-Second Annual Symposium on Principles of Distributed Computing (PODC)}, page 102–111, 2003.

\bibitem{ShunBlellochSPAA14}
Julian Shun and Guy~E. Blelloch.
\newblock Phase-concurrent hash tables for determinism.
\newblock In {\em Proceedings of the 26th ACM Symposium on Parallelism in Algorithms and Architectures}, SPAA '14, page 96–107, New York, NY, USA, 2014. Association for Computing Machinery.

\bibitem{Snyder77}
Lawrence Snyder.
\newblock On uniquely represented data structures.
\newblock In {\em Proc.\ of the 18th Annual IEEE Symposium on Foundations of Computer Science (FOCS)}, pages 142--146, 1977.

\bibitem{SSBHW10}
Alex Stivala, Peter~J. Stuckey, Maria {Garcia de la Banda}, Manuel Hermenegildo, and Anthony Wirth.
\newblock Lock-free parallel dynamic programming.
\newblock {\em Journal of Parallel and Distributed Computing}, 70(8):839--848, 2010.

\bibitem{SundarTa90}
Rajamani Sundar and Robert~Endre Tarjan.
\newblock Unique binary search tree representations and equality-testing of sets and sequences.
\newblock In {\em Proc.\ of the 22nd Annual ACM Symposium on Theory of Computing (STOC)}, pages 18--25, 1990.

\bibitem{Thorup}
Mikkel Thorup.
\newblock Linear probing with 5-independent hashing.
\newblock {\em CoRR}, abs/1509.04549, 2015.

\bibitem{Tzouramanis12}
Theodoros Tzouramanis.
\newblock History-independence: a fresh look at the case of {R}-trees.
\newblock In {\em Proc.\ 27th Annual ACM Symposium on Applied Computing (SAC)}, pages 7--12, 2012.

\end{thebibliography}
